\PassOptionsToPackage{unicode}{hyperref}
\PassOptionsToPackage{hyphens}{url}
\PassOptionsToPackage{dvipsnames,svgnames,x11names}{xcolor}
\documentclass[
  12pt]{article}
\usepackage{etoc}
\usepackage{amsmath,amssymb}
\usepackage{iftex}
\ifPDFTeX
  \usepackage[T1]{fontenc}
  \usepackage[utf8]{inputenc}
  \usepackage{textcomp} 
\else 
  \usepackage{unicode-math}
  \defaultfontfeatures{Scale=MatchLowercase}
  \defaultfontfeatures[\rmfamily]{Ligatures=TeX,Scale=1}
\fi
\usepackage{lmodern}
\ifPDFTeX\else  
\fi
\IfFileExists{upquote.sty}{\usepackage{upquote}}{}
\makeatletter
\@ifundefined{KOMAClassName}{
  \IfFileExists{parskip.sty}{%
    \usepackage{parskip}
  }{
    \setlength{\parindent}{0pt}
    \setlength{\parskip}{6pt plus 2pt minus 1pt}}
}{
  \KOMAoptions{parskip=half}}
\makeatother
\usepackage{xcolor}
\usepackage{titletoc}
\makeatletter
\ifx\paragraph\undefined\else
  \let\oldparagraph\paragraph
  \renewcommand{\paragraph}{
    \@ifstar
      \xxxParagraphStar
      \xxxParagraphNoStar
  }
  \newcommand{\xxxParagraphStar}[1]{\oldparagraph*{#1}\mbox{}}
  \newcommand{\xxxParagraphNoStar}[1]{\oldparagraph{#1}\mbox{}}
\fi
\ifx\subparagraph\undefined\else
  \let\oldsubparagraph\subparagraph
  \renewcommand{\subparagraph}{
    \@ifstar
      \xxxSubParagraphStar
      \xxxSubParagraphNoStar
  }
  \newcommand{\xxxSubParagraphStar}[1]{\oldsubparagraph*{#1}\mbox{}}
  \newcommand{\xxxSubParagraphNoStar}[1]{\oldsubparagraph{#1}\mbox{}}
\fi
\makeatother

\usepackage{longtable,booktabs,array}
\usepackage{calc} 
\usepackage{etoolbox}
\makeatletter
\patchcmd\longtable{\par}{\if@noskipsec\mbox{}\fi\par}{}{}
\makeatother
\usepackage{footnotehyper}
\makesavenoteenv{longtable}
\usepackage{graphicx}
\usepackage{titletoc}
\makeatletter
\def\maxwidth{\ifdim\Gin@nat@width>\linewidth\linewidth\else\Gin@nat@width\fi}
\def\maxheight{\ifdim\Gin@nat@height>\textheight\textheight\else\Gin@nat@height\fi}
\makeatother
\setkeys{Gin}{width=\maxwidth,height=\maxheight,keepaspectratio}
\makeatletter
\def\fps@figure{htbp}
\makeatother

\makeatletter
\@ifpackageloaded{caption}{}{\usepackage{caption}}
\AtBeginDocument{%
\ifdefined\contentsname
  \renewcommand*\contentsname{Table of contents}
\else
  \newcommand\contentsname{Table of contents}
\fi
\ifdefined\listfigurename
  \renewcommand*\listfigurename{List of Figures}
\else
  \newcommand\listfigurename{List of Figures}
\fi
\ifdefined\listtablename
  \renewcommand*\listtablename{List of Tables}
\else
  \newcommand\listtablename{List of Tables}
\fi
\ifdefined\figurename
  \renewcommand*\figurename{Figure}
\else
  \newcommand\figurename{Figure}
\fi
\ifdefined\tablename
  \renewcommand*\tablename{Table}
\else
  \newcommand\tablename{Table}
\fi
}
\@ifpackageloaded{float}{}{\usepackage{float}}
\floatstyle{ruled}
\@ifundefined{c@chapter}{\newfloat{codelisting}{h}{lop}}{\newfloat{codelisting}{h}{lop}[chapter]}
\floatname{codelisting}{Listing}

\makeatother
\makeatletter
\@ifpackageloaded{caption}{}{\usepackage{caption}}
\@ifpackageloaded{subcaption}{}{\usepackage{subcaption}}
\makeatother

\ifLuaTeX
  \usepackage{selnolig}  
\fi
\usepackage[square,numbers,sort&compress]{natbib}
\usepackage{bookmark}

\IfFileExists{xurl.sty}{\usepackage{xurl}}{} 
\hypersetup{
  pdftitle={Modeling Dynamic Correlation Matrices with Shrinkage Priors},
  pdfauthor={Author 1; Author 2},
  pdfkeywords={3 to 6 keywords, that do not appear in the title},
  colorlinks=true,
  linkcolor={blue},
  filecolor={Maroon},
  citecolor={Blue},
  urlcolor={Blue},
  pdfcreator={LaTeX via pandoc}}

\newcommand{\anon}{1}

\usepackage{a4wide,amsmath,graphicx,amssymb,mathtools,float,array, mathbbol}
\usepackage{setspace}
\usepackage{comment}
\usepackage[T1]{fontenc}
\usepackage{authblk}
\usepackage{amsthm}

\newtheorem{Theorem}{Theorem}

\newtheorem{proposition}[Theorem]{Proposition}

\newcommand{\ignore}[1]{}

\begin{document}

\def\spacingset#1{\renewcommand{\baselinestretch}%
{#1}\small\normalsize} \spacingset{1}


\if1\anon
{
  \title{\bf Modeling Dynamic Correlation Matrices with Shrinkage Priors}

  \author{
    Daniel Andrew Coulson\\
    \texttt{dac382@cornell.edu}\\
    \begin{tabular}{c@{\qquad\qquad}c}
      David S. Matteson \& Martin T. Wells\\
      \texttt{dm484@cornell.edu} \& \texttt{mtw1@cornell.edu}
    \end{tabular}
  }

  \date{
    Department of Statistics and Data Science, Cornell University
  }

  \maketitle
} \fi

\bigskip
\begin{abstract}
Estimating time-varying correlation matrices is challenging because existing methods may adapt slowly to structural changes, impose insufficient regularization, or produce diffuse posterior uncertainty. In moderate dimensions, an additional difficulty is summarizing the estimated evolving dependence structure for downstream decision-making tasks. We propose a Bayesian approach based on a low-rank factor representation, with latent states evolving under a dynamic shrinkage prior and observation errors following a multivariate factor stochastic volatility model. This specification allows locally adaptive regularization of the estimated correlation structure over time and informative uncertainty quantification. We establish, to our knowledge, a first-of-its-kind posterior contraction result for dynamically regularized Bayesian models, showing contraction around the true model parameters at an explicit rate under averaged Hellinger distance. To summarize the estimated correlation matrices, we build on the information-theoretic concept of total correlation to obtain a scalar measure of cross-sectional dependence. Simulation studies show improved accuracy and responsiveness relative to competing methods in a range of challenging scenarios. We then apply our method to monitoring the correlation evolution of equity portfolios during periods of financial market stress.  
\end{abstract}

\noindent%
\vfill

\newpage
\spacingset{1.8} 

\section{Introduction}\label{sec-intro}

Existing non-Bayesian methods for estimating time-varying correlation matrices include rolling window and exponentially weighted estimators, regularized estimators, and multivariate GARCH-based dynamic correlation models \citep{ledoit2004well, bollerslev1986generalized, bickel2008covariance}. Rolling and exponentially weighted estimators remain popular because of their simplicity and ease of implementation, but their reliance on moving windows can smooth over abrupt changes and therefore may adapt only gradually when dependence changes quickly. In moderate dimensions, such estimators may also be noisy, since each window contains limited information relative to the number of pairwise dependence parameters to be estimated. 
Regularization, including shrinkage and sparsity-inducing penalties, can substantially improve stability of the estimators, but stability alone does not guarantee adequate adaptation to temporal changes in dependence, particularly when those changes are heterogeneous and abrupt. 

Parameterized dynamic correlation models, including standard Dynamic Conditional Correlation (DCC) specifications \cite{engle2002dynamic} and extensions such as asymmetric DCC \cite{cappiello2006asymmetric}, achieve tractability through structured low-dimensional dynamic formulations. Although these models are often parsimonious, they can be restrictive for multivariate financial time series in which dependence evolves differently across assets and may change sharply during periods of market stress. 

Uncertainty assessments for such procedures can be obtained using asymptotic or bootstrap methods. However, these methods are typically approximate, do not provide a unified probabilistic characterization of uncertainty over the full time-evolving correlation process, and can be empirically miscalibrated in challenging settings, as illustrated in our simulation studies.

Bayesian methods provide a unified probabilistic framework for inference on time-varying correlation matrices, allowing uncertainty quantification for latent states, model parameters, and derived quantities through the posterior distribution. Existing Bayesian approaches to modeling dynamic correlation matrices include matrix evolution models based on Wishart-type priors, factor stochastic volatility models, and Cholesky-based multivariate stochastic volatility specifications, in some cases supplemented with hierarchical shrinkage to stabilize estimation in higher dimensions \citep{cholsv,factorstochvol,philipov2006multivariate,triantafyllopoulos2008multivariate,aguilar2000bayesian,cogley2005drifts,primiceri2005time}. However, computational tractability is often achieved through structural assumptions, such as low-rank factor dependence, recursive decompositions, or specific matrix evolution laws, which may be restrictive when dependence changes unevenly across variables or exhibits sudden local shifts. 

Moreover, when shrinkage is imposed in existing approaches, it is often static, so that the strength of regularization does not adapt over time. This can lead to poorer estimation accuracy during periods of rapid change in the dependence structure, because such changes may be overshrunk toward the prior center instead of accommodated through a temporary increase in latent state innovation variance. This is especially relevant in settings where dependence changes suddenly, such as periods of market stress in financial time series. 

We propose a dynamic correlation model with a factor specification in the conditional mean and a low-rank multivariate factor stochastic volatility structure for the joint observation error vector. We place the dynamic shrinkage prior (DSP) introduced by \citet{kowal2019dynamic} on the innovations in the time-varying parameter processes. This allows the degree of regularization to adapt over time, strongly shrinking small innovations while permitting larger movements when supported by the data. Because the parameters evolve separately, the model can accommodate heterogeneous changes across variables rather than imposing a common pattern of evolution throughout the system. The resulting model specification retains a parsimonious structure for tractability while allowing locally adaptive changes in the correlation dynamics. 

In addition, we show that the proposed Bayesian model satisfies a posterior contraction property. The posterior contraction property is not theory for its own sake, but a principled theoretical justification of a Bayesian method. In particular, it provides a Bayesian analogue of frequentist consistency and, through the contraction rate, quantifies the rate at which the procedure learns the underlying parameters from the data. It ensures that, as the sample size increases, the posterior distribution contracts around the true data-generating parameters, thereby formalizing the notion that the posterior optimally learns from the data. 

DSP-type priors, such as the dynamic horseshoe, are widely used in many areas, including time-varying parameter regression, change-point analysis, and stochastic volatility models \citep{kowal2019dynamic, hauzenberger2024dynamic, wu2025trend, huber2021dynamic}. However, relatively little is known about the theoretical properties of such hierarchical priors. To our knowledge, no posterior contraction results exist for DSP-type priors; our theorem provides the first such guarantee. 

Because a time-varying $N\times N$ correlation matrix contains $N(N-1)/2$ distinct pairwise correlations at each time point, even moderate-dimensional systems can quickly become a case of not seeing the forest for the trees. For example, a financial portfolio with 1,000 assets yields 499,500 pairwise correlation series to track over time. Visual inspection then becomes impractical and posterior analysis cumbersome when inference is based on thousands of Monte Carlo draws, since the total number of pairwise correlation series can easily grow into the hundreds of millions. 

To address this, we summarize each correlation matrix by a single scalar score that measures the overall strength of dependence in the system rather than any individual pairwise association. The score is 0 under an identity correlation matrix and increases toward 1 as co-movement strengthens. This provides a compact and interpretable way to track system-wide dependence over time, and it forms the central object of study in this paper rather than individual pairwise correlations. 

In quantitative portfolio management, researchers often develop predictive signals for a large universe of assets, and these signals must be combined into a trading strategy while controlling overall portfolio risk. A key consideration in portfolio construction is the dependence structure of asset returns, since highly correlated positions can weaken diversification benefits, increase portfolio variance, and heighten vulnerability to large drawdowns by increasing the likelihood of simultaneous losses across holdings. For this reason, it is important not only to forecast returns, but also to assess how portfolio co-movement evolves over time. Because our method is a smoother, it can be applied to backtested returns across portfolio assets to determine whether a proposed strategy would historically have exhibited periods of elevated dependence and, hence, greater fragility, while also providing covariance estimates for out-of-sample portfolio construction. As discussed above, large portfolios generate hundreds of thousands of pairwise correlation series, and our scalar score provides the portfolio manager with a convenient summary measure of overall correlation at each time point.

This paper makes three main contributions. First, we develop a Bayesian method for estimating time-varying correlation matrices that combines dynamic shrinkage priors with a multivariate factor stochastic volatility model. Second, we establish the first posterior contraction result for this class of dynamic shrinkage models, providing asymptotic consistency and frequentist justification. Third, we introduce a scalar summary that provides an information-theoretically grounded measure of overall dependence. Taken together, these results advance both the theory and application of Bayesian dynamic shrinkage methods for modeling time-varying dependence.

The remainder of this paper is organized as follows. In Section 2, we specify the observation equation and prior distribution. Section 3 introduces the proposed scalar summary measure. In Section 4, we establish a posterior contraction result and provide a proof sketch. Section 5 presents the results of simulation studies and real-world applications. Section 6 concludes.  

\section{Model for returns }
\subsection*{Terminology and notation}
We consider an equity portfolio consisting of $N$ assets indexed by $a \in \{ 1, \dots, N\}$. Let $P_{a,t}$ denote the adjusted closing price of asset $a$ on day $t \leq T$, where $T\in \mathbb{N}$. The simple return on asset $a$ at time $t$ is 
\begin{equation*}
    R_{a,t} = \frac{P_{a,t}-P_{a,t-1}}{P_{a,t-1}}.
\end{equation*}
Financial theory posits the existence of a risk-free asset, such as a short-term U.S. Treasury bill, with a one-period risk-free rate $\mu^{*}_{t}$. The excess return on asset $a$ is then $r_{a,t} := R_{a,t} - \mu^{*}_{t}$. In addition, we denote the excess return of a whole market portfolio by $r_{M,t}$. 

A common modeling framework for excess returns is the multivariate linear factor model. Let $\boldsymbol{r}_{t} = (r_{1,t}, \dots, r_{N,t})^\top$. The multivariate linear factor model is given by 
\begin{equation*}
    \boldsymbol{r}_{t} = \boldsymbol{\beta}_{0,t} + \boldsymbol{\beta}_{1,t} F_{1,t} + \dots + \boldsymbol{\beta}_{z,t} F_{z,t} + \boldsymbol{\epsilon}_{t}, \qquad \boldsymbol{\epsilon}_{t} = (\epsilon_{1,t}, \dots, \epsilon_{N,t})^\top. 
\end{equation*}

The variables $F_{1,t},\dots,F_{z,t}$ are factors, which are often
observed, although some models allow for latent factors, denoted by
$\boldsymbol f_t$. The vectors $\boldsymbol\beta_{1,t},\dots,\boldsymbol\beta_{z,t}$
are unknown factor loadings, and $\boldsymbol\epsilon_t$ is a vector of
idiosyncratic errors whose components need not be independent and
identically distributed. In the main paper, we consider only the excess
market return $r_{M,t}$, with additional observed factors considered in
Section~C of the supplement; Section~2.1.1 specifies the factor
structure and the distribution of $\boldsymbol\epsilon_t$.

\subsection{Model}
We model the time-varying coefficients using dynamic shrinkage priors (DSP) of \citet{kowal2019dynamic}, which extend global-local shrinkage priors by introducing a time-local component,  
so that the degree of shrinkage
varies across both parameters and time. As normal mean--scale mixtures,
these priors place substantial mass near zero while allowing occasional
large innovations. For the observation errors, we use a multivariate
factor stochastic volatility (MFSV) model, which imposes a low-rank
structure on the time-varying covariance matrix.

\subsubsection{Observation equation}
Our observation equation is a time-varying version of the Capital Asset Pricing Model (CAPM) \citep{sharpe1964capital,lintner1965security, lintner1965valuation,mossin1966equilibrium}. In particular,
\begin{align*}
    r_{a,t} &= \alpha_{a,t} + \beta_{a,t}r_{M,t} + \epsilon_{a,t},
    \qquad r_{M,t}|h_{M,t} = \exp\{\frac{h_{M,t}}{2} \}\epsilon_{M,t}, \epsilon_{M,t} \sim N(0,1), \\
    \boldsymbol{\epsilon}_{t} &= (\epsilon_{1,t},\dots, \epsilon_{N,t})^\top, \qquad \boldsymbol{\epsilon}_{t}| (\Lambda, \boldsymbol{f}_{t} , \bar{\Sigma}_{t}) \sim N_{N}(\Lambda\boldsymbol{f}_{t} , \bar{\Sigma}_{t}) \text{ where } \boldsymbol{f}_{t} | \tilde{\Sigma}_{t} \sim N_{r}(\boldsymbol{0}_{r}, \tilde{\Sigma}_{t}).
\end{align*}

The excess asset returns $r_{a,t}$ and market return $r_{M,t}$ are observed. The time-varying coefficients $\alpha_{a,t}$ and $\beta_{a,t}$, the error vector $\boldsymbol{\epsilon}_{t}$, the factor loadings $\Lambda$, the $r$ latent factors $\boldsymbol{f}_{t}$, and the stochastic volatility processes $\bar{\Sigma}_{t}$ and $\tilde{\Sigma}_{t}$ are unobserved. 

The model allows excess returns to depend on time-varying intercepts and market exposures, while permitting the unexplained component of returns to exhibit both cross-sectional dependence and time-varying volatility. At time t, $\alpha_{a,t}$ captures abnormal performance and $\beta_{a,t}$ measures the contemporaneous exposure of asset $a$ to the market factor. The error term $\epsilon_{a,t}$ represents the residual variation that is not explained by the market return. The factor stochastic volatility formulation parsimoniously models the time-varying covariance matrix: a small number of latent factors capture the dominant common movements in residual returns, while the stochastic volatility components allow both common and asset-specific uncertainty to evolve over time. Because the errors are conditionally Gaussian under a scale-mixture representation, the resulting marginal joint distribution is heavy-tailed, accommodating large shocks in financial returns. 

From the observation equation, one can then construct the time-varying covariance matrix. In particular, $  \text{cov}(\boldsymbol{r}_{t}) = \text{var}(\boldsymbol{\alpha}_{t}) + \text{var}(r_{M,t}) \mathbb{E}[\boldsymbol{\beta}_{t} \boldsymbol{\beta}_{t}^\top] + \text{cov}(\boldsymbol{\epsilon}_{t}).$
The expression follows from mutual independence and the mean-zero specification for $r_{M,t}$; standardization then yields the correlation matrix. 

\subsubsection{Prior Specification}
The evolution of $\{\beta_{a,t} \}_{1\leq t \leq T}$ is modeled by the prior stochastic process given by 
\begin{equation*}
    \beta_{a,t+1} = \beta_{a,t} + \omega_{\beta_{a},t},\qquad\omega_{\beta_{a},t} |\tau_{a,0}\tau_{\beta_{a}}, \{ \lambda_{\beta_{a},s} \} \sim N(0, \tau_{a,0}^{2}\tau_{\beta_{a}}^{2}\lambda_{\beta_{a},t}^{2}),
\end{equation*}
\begin{equation*}
    h_{\beta_{a},t} = \log(\tau_{a,0}^{2}\tau_{\beta_{a}}^{2}\lambda_{\beta_{a},t}^{2}), \qquad h_{\beta_{a},t} = \mu_{\beta_{a}} + \phi_{\beta_{a}}(h_{\beta_{a},t-1} - \mu_{\beta_{a}} ) + \eta_{\beta_{a},t},
\end{equation*}
\begin{equation*}
    \tau_{a,0} \sim C^{+}(0, \frac{1}{\sqrt{T}}), \tau_{\beta_{a}} \sim C^{+}(0,1),\eta_{\beta_{a},t} \sim Z(\frac{1}{2}, \frac{1}{2}, 0, 1), \frac{\phi_{\beta_{a}} + 1}{2} \sim Beta(10,2) , 
\end{equation*}
\begin{equation*}
    h_{M,t+1} = \mu_{M} + \phi_{M}(h_{M,t} - \mu_{M}) + \sigma_{M} \eta_{t}, \eta_{t}\sim N(0,1)
\end{equation*}
\begin{equation*}
    \mu_{M} \sim N(0,100), \frac{\phi_{M}+1}{2}\sim Beta(10,3), \sigma_{M}^{2} \sim Ga(\frac{1}{2}, \frac{1}{2}).
\end{equation*}
The prior specification for $\{ \alpha_{a,t} \}_{1\leq t\leq T}$ is identical to that for $\{ \beta_{a,t} \}_{1\leq t\leq T}$, with $\beta_{a}$ subscripts replaced by $\alpha_{a}$ subscripts. The random walk evolution provides a simple and interpretable benchmark for gradual parameter drift. A global-local shrinkage structure is used because unrestricted random-walk innovations can lead to excessive variation in the state paths. The global component, $\tau_{a,0}$, controls the overall degree of time-variation; the coefficient-specific components, $\tau_{\alpha_{a}}$ and $\tau_{\beta_{a}}$, allow differing levels of parameter instability; and the local components, $\lambda_{\alpha_{a,t}}$ and $\lambda_{\beta_{a,t}}$, permit isolated periods of pronounced movement. 

We place stochastic volatility priors on the innovation variances because parameter instability is unlikely to be constant over time. Dynamic innovation variances allow the model to distinguish stable periods, in which coefficients are nearly constant, from turbulent periods, in which larger parameter movements may occur. This is important for asset returns, where return volatility and risk-exposure stability can vary substantially across market conditions.

We model the observation error covariance process using the following prior process with $i = 1,\dots, N$ and $j = 1, \dots, r$: 
\begin{equation*}
    \bar{\Sigma}_{t} = diag(\exp(\bar{h}_{t,1}),..., \exp(\bar{h}_{t,N})),\qquad \tilde{\Sigma}_{t} = diag(\exp(\tilde{h}_{t,1}),..., \exp(\tilde{h}_{t,r})),
\end{equation*}
\begin{equation*}
    \bar{h}_{t,i} \sim N(\bar{\mu}_{i} + \bar{\psi}_{i} (\bar{h}_{t-1,i} - \bar{\mu}_{i}), \bar{\sigma}_{i}^{2}), \qquad\tilde{h}_{t,j} \sim N( \tilde{\psi}_{j}\tilde{h}_{t-1,j}, \tilde{\sigma}_{j}^{2})
\end{equation*}
\begin{equation*}
\Lambda_{i,j}|\tau_{i,j}^{2} \sim N(0, \tau_{i,j}^{2}), \tau_{i,j}^{2}|\lambda_{i}^{2} \sim Ga(0.1,\frac{0.1 \lambda_{i}^{2}}{2}), \lambda_{i}^{2} \sim Ga(1,1),
\end{equation*}
\begin{equation*}
    \bar{\sigma}_{i} \sim Ga(\frac{1}{2}, \frac{1}{2}), \tilde{\sigma}_{j} \sim Ga(\frac{1}{2}, \frac{1}{2}),\bar{\mu}_{i} \sim N(0,10), 
\end{equation*}
\begin{equation*}
 \frac{\bar{\psi}_{i} + 1}{2} \sim Beta(10,3), \frac{\tilde{\psi}_{j} + 1}{2} \sim Beta (10,3),
\end{equation*}
The factor stochastic volatility prior is motivated by well-known features of financial returns: cross-sectional dependence, heteroskedasticity, and heavy tails. An unrestricted time-varying covariance matrix would be prohibitively parameterized in large asset panels. The factor structure is a parsimonious alternative: a few latent factors capture common variation in residual risk, while asset-specific stochastic volatilities capture idiosyncratic movements. The priors on the persistence parameters encode the empirically standard belief that volatility is persistent; centering prior mass on persistent yet stationary dynamics allows substantial temporal dependence without permitting explosive behavior. 

We provide the full posterior sampling algorithm in the supplement. The algorithm uses the multivariate dynamic shrinkage process sampler from \citet{kowal2019dynamic} and the univariate and multivariate factor stochastic volatility samplers from \citet{factorstochvol}. We call this combination of observation model and prior specification DSP-MFSV-CAPM. For the stochastic volatility process we use the ancillarity-sufficiency interweaving strategy which interweaves the centered and non-centered SV parameterizations to improve SV parameter mixing. For MFSV, we impose no triangular or other restrictions on $\Lambda$. Factor scales are identified by normalizing the factor log-volatility levels with sign ambiguity handled ex post. We use the unrestricted specification because inference concerns the implied time-varying covariance matrix rather than individual factor loadings. Conditional on latent factors and volatilities, loading matrix rows are updated using Bayesian regression, with deep interweaving improving sampling efficiency. We use the largest absolute loading per factor column for the deep-interweaving transformation. 

We set the number of latent factors to $r=3$; Section~A of the
supplement reports a sensitivity study.

\subsubsection{Prior structure and interpretation}
\label{sec:prior-short}

The prior specification contains two shrinkage components: the DSP
component regularizes the temporal evolution of the state variables
$(\alpha_{a,t},\beta_{a,t})$, and the normal--gamma prior on $\Lambda$,
together with the stochastic volatility processes, regularizes the
cross-sectional dependence of the observation errors. The two
components interact in the full conditional distribution of the state
variables, where the posterior precision combines the DSP innovation
precisions with the observation error precisions implied by the MFSV
component. Shrinkage is therefore a ratio of prior precision to noise
precision with an estimated, time-varying noise scale; the effective shrinkage strengthens as residual volatility increases. Entrywise, for $i\ne j$,
\begin{equation}
\operatorname{cov}(\boldsymbol r_t)_{ij}
=\operatorname{var}(r_{M,t})\,
\mathrm{E}[\beta_{i,t}]\,\mathrm{E}[\beta_{j,t}]
+\sum_{k=1}^{r}\Lambda_{ik}\Lambda_{jk}\,e^{\tilde h_{t,k}},
\label{eq:offdiag-main}
\end{equation}
so every correlation, and hence the proposed score, sums a systematic
component governed by the state variables and a residual component
governed by the prior on $\Lambda$. The allocation 
between the components in \eqref{eq:offdiag-main} is weakly identified and is regularized by the relative strength of the two priors (Supplement Section~B). The normal--gamma prior encodes the belief that most residual comovement is negligible, and its continuous shrinkage of the columns of $\Lambda$, rather than $r$, determines the effective number of factors (Supplement Section~A).

The roles of the global, coefficient-specific, and local components
described above can be made precise. Since
$\mu_{\beta_a}=\log(\tau_{a,0}^2\tau_{\beta_a}^2)$, the half-Cauchy
global scales set the stationary mean of the log-variance process, the
$Z$-distributed innovations generate deviations from this mean and
$\phi_{\beta_a}$ determines their persistence; the scale of
$1/\sqrt T$ encodes a moderate number of genuine regime shifts per
coefficient path, and $\lambda_{\beta_a,t}$ acts as a bandwidth that
either strongly shrinks an innovation or admits it without attenuation. The
implied shrinkage weight
$\kappa_{a,t}=(1+\tau_{a,0}^{2}\tau_{\beta_a}^{2}\lambda_{\beta_a,t}^{2})^{-1}$
places most of its prior mass near zero and one, so an innovation is
either strongly shrunk or left essentially untouched; this property
underlies the short response lags in Section~5.1. Setting
$\phi_{\beta_a}=0$ recovers independent horseshoe shrinkage at each
time point, and the polynomial tails and logarithmic origin behavior of
the marginal prior on an innovation hold for every value of
$\phi_{\beta_a}$ (Supplement Proposition~I.1); therefore, the persistence
parameter only pools the treatment of intermediate
innovations across neighboring time points. This pooling is the source
of the well-calibrated credible intervals in Tables~1--5,
while the heavy tails provide the robustness. Beyond centering
the prior mass on persistent yet stationary dynamics, the $\mathrm{Beta}(a,b)$, $a>b$, form of the persistence priors provides a
tail bound used in the sieve construction in the proof of Theorem \ref{thm:correlation_posterior_contraction}. Section~I of the Supplement develops the full account, including the marginal distributions, an identity for eliciting $\phi_{\beta_a}$, the
comparison between static and dynamic shrinkage, and the justification
of each component of the prior.

\subsubsection{Limitations}
The main limitation of DSP-MFSV-CAPM is computational, not statistical. As a large latent-variable state-space model estimated by MCMC, each iteration must update many unknowns at every time point. The DSP component introduces asset-specific time-varying parameters and local shrinkage variables, so the latent dimension grows on the order of $N \times T$. The MFSV component avoids estimating an unrestricted $N\times N$ time-varying covariance matrix by imposing a low-rank factor structure. But posterior sampling still requires updating factor trajectories, stochastic volatility paths for both factor and idiosyncratic variances, and the $N \times r$ loading matrix. As a result, computational cost and mixing worsen as $T$ and $N$ increase. We therefore focus on moderate-sized panels with $N=30$ and $T \leq 1257$, where the model remains tractable while retaining its main advantages: locally adaptive shrinkage of time-varying exposures and a parsimonious, interpretable representation of time-varying dependence. In Section G of the supplement, we apply DSP-MFSV-CAPM to a portfolio of 1,000 assets. 

Competing methods differ substantially in computational cost. DCC-type models are typically fast to estimate by numerical optimization. Loggle uses efficient ADMM and blockwise estimation. Sparse graphical and penalized covariance methods can be faster than simulation-based Bayesian methods in high dimensions. Variational Bayes reduces computation by replacing MCMC with approximate inference. However, non-Bayesian competitors may require procedures like bootstrapping for uncertainty quantification, which can increase computational cost.

Using a small set of observed factors, as in CAPM-style models, has an important limitation. If the factor specification is incomplete or misspecified, residuals may retain substantial systematic co-movement, for example from sectoral shocks, potentially biasing inference on regression coefficients through omitted-factor effects. In our model, however, a parsimonious observed-factor specification is justified because the market factor typically explains a large share of equity co-movement while MFSV directly models the remaining dependence in $\boldsymbol{\epsilon}_{t}$. A small observed-factor dimension also reduces computation. Additional observed factors introduce further time-varying parameters and shrinkage paths, increasing computation and potentially slowing mixing. Multiple time-varying parameters associated with correlated factors can also worsen ill-conditioning relative to a one-factor model. As seen in Table 3 in Section 5.1, our model achieves strong performance compared to competing methods despite incomplete factors. Section C of the supplement extends the model to the five-factor Fama–French specification.

Our use of continuous shrinkage priors reflects a trade-off between
regularization and exact sparsity. Compared with spike-and-slab priors,
continuous shrinkage strongly shrinks small innovations while allowing
occasional large ones, and it avoids discrete selection indicators,
which can produce multimodal posteriors and poor MCMC mixing in an
already high-dimensional model. However, it does not yield exact
posterior zeros, so posterior correlation matrices remain dense. Since
the number of pairwise correlations grows quadratically with the
cross-section, this complicates network interpretation and the
distinction between economically meaningful and negligible
relationships. Exact sparsity may therefore be preferable when edge
selection or an interpretable dependence network is the primary goal,
whereas continuous shrinkage is preferable for stable estimation of the
full time-varying correlation structure and its functionals, as in our
analysis. Sparsity can also be recovered by post-processing, for
example by thresholding or other decision-theoretic selection rules
that trade off faithfulness to the estimated dependence structure
against sparsity; this retains the computational advantages of
continuous shrinkage, provided that the positive definiteness of the
correlation matrices is preserved.

Finally, DSP-MFSV-CAPM is a smoother designed to infer how portfolio
dependence evolves over time, which makes it well suited to
backtesting, particularly for evaluating portfolio diversification;
Section~F of the supplement examines out-of-sample performance.

\section{Summary of correlation matrices}
Once we have an estimated covariance matrix, we can derive the associated correlation matrix by standardization. However, for even moderately sized correlation matrices, it is not obvious how to extract a useful summary for portfolio allocation. A common approach is to plot estimated pairwise correlations over time. However, multiple pairwise correlation series quickly become cumbersome and uninformative, particularly for assessing overall portfolio correlation.   

To summarize a correlation matrix $R_t$ by a scalar measure of overall correlation, we first need a definition of `overall correlation'. Rather than average pairwise entries, we view overall correlation as the strength of joint dependence in the portfolio; that is, how far the joint excess return distribution is from what it would look like if all assets were independent. Under this view, stronger comovement, as during crises, corresponds to a joint distribution that is increasingly constrained to fewer effective directions, implying greater departure from independence and less diversification. We summarize $R_{t}$ by the dimension-normalized log-determinant score
\begin{equation}
\label{score}
    \text{score}(R_{t}) = 1-\det(R_{t})^{\frac{1}{N}} = 1- \exp\{\frac{1}{N}\log(\det(R_{t}))\} \in [0,1]. 
\end{equation}
If $R=I_{N}$, $\text{score}(R)=0$. At one extreme, if $R = \boldsymbol{1}_{N}\boldsymbol{1}_{N}^\top$, a singular case corresponding to perfect comovement, then $\det(R)=0$ and $\text{score}(R)=1$. For an equicorrelation matrix $R(\rho) = (1-\rho)I_{N} + \rho\boldsymbol{1}_{N} \boldsymbol{1}_{N}^\top$, positive semidefiniteness requires $\rho \geq -1/(N-1)$. At $\rho_{\min} = -1/(N-1)$, the matrix is singular, so $\det(R(\rho_{\min})) = 0$ and hence $\text{score}(R(\rho_{\min}))=1$.  

The score also has an information-theoretic interpretation through total correlation, $TC(\cdot)$, a form of mutual information. Consider $\boldsymbol{X}\sim N_{N}(\boldsymbol{0}_{N}, R)$ with correlation matrix $R$, where $X_{1},\dots, X_{N}$ have unit variance. Let $f(\boldsymbol{x}) = \phi_{p}(\boldsymbol{x}; \boldsymbol{0}_{N}, R)$ and $g(\boldsymbol{x}) = \prod\limits_{i=1}^{N} \phi_{1}(x_{i}; 0 , 1)$. The Kullback-Leibler divergence between $f$ and $g$ is given by $D_{KL}(f||g)= - \frac{1}{2}\log(\det(R))$, by \citet{pardo2018statistical}. 
Then $\mathcal{I}(R) :=  \frac{1}{N}TC(\boldsymbol{X}) = - \frac{1}{2N}\log(\det(R))$ is the dimension-normalized total correlation, with $\text{score}(R) = 1- \exp (-2\mathcal{I}(R))$. Also, $\text{score}(R)$ measures how much the correlation structure compresses the joint return cloud relative to independence: stronger comovement implies smaller effective volume, and hence a larger score. See Supplementary Section H for the details. Additional scalar summaries are considered in Sections D.2 and E.2. Section F.2 also uses such summaries as out-of-sample leverage decision rules.

\section{Posterior Contraction}
We state the posterior contraction result, outline the proof, and discuss its implications. The full proof is in Section J of the supplement. Posterior contraction is important for several reasons. First, it provides a frequentist justification for Bayesian inference by establishing consistency and guarding against misleading conclusions in the large-sample limit. Second, the contraction rate quantifies statistical efficiency; when it matches the minimax-optimal rate, it shows that the Bayesian procedure attains optimal statistical performance. Third, posterior contraction underpins the calibration of Bayesian credible sets. Under suitable conditions, they can also serve as valid frequentist confidence sets. Finally, the contraction result illuminates the role of the prior, clarifying how prior specification balances flexibility and regularization in determining posterior behavior. In hierarchical shrinkage models, such results verify that prior regularization does not asymptotically overwhelm the likelihood.

Our contraction result theoretically justifies DSP-MFSV-CAPM by establishing consistent estimation and supporting reliable inference. The proof applies Theorem 3 of \cite{ghosal2007convergence}. We establish the existence of consistent hypothesis tests (Proposition 2) and a prior mass condition (Proposition 3). We define sieves $\Theta_{n} \subseteq \Theta$, and $H_{n} \subseteq H$, where $\Theta$ and $H$ denote the parameter and latent-state spaces, respectively. Let $\boldsymbol{\theta}$ denote the vector of model parameters. For $\boldsymbol{\theta} \in \Theta$ and a latent state path $\boldsymbol{s} \in H$, consider the conditional experiment $Q_{\boldsymbol{\theta},\boldsymbol{s}_{n}}^{(n)} (\cdot) := \mathbb{P}_{\boldsymbol{\theta}}(\boldsymbol{r}_{1:n}\in \cdot|\boldsymbol{s}_{n}, \boldsymbol{r}_{M, 1:n})$, where $\boldsymbol{s}_{n}$ denotes the latent states through time $n$, $\boldsymbol{r}_{1:n} = (\boldsymbol{r}_{1}, \dots, \boldsymbol{r}_{n})$ and $\boldsymbol{r}_{M,1:n} = (r_{M,1},\dots, r_{M,n})$. By conditional independence, 
\begin{equation}
\label{condexp}
    Q_{\boldsymbol{\theta}, \boldsymbol{s}_{n}}^{(n)} = \bigotimes_{t=1}^{n} Q_{\boldsymbol{\theta}, \boldsymbol{s}_{n,t}}, q_{\boldsymbol{\theta},\boldsymbol{s}_{n}}^{(n)}(\boldsymbol{x}^{(n)}) = \prod\limits_{t=1}^{n} q_{\boldsymbol{\theta}, \boldsymbol{s}_{n,t}}(\boldsymbol{x}_{t}), \text{ where } q_{\boldsymbol{\theta},\boldsymbol{s}_{n,t}} (\cdot) = \phi_{N}(\cdot; \mu_{t}(\boldsymbol{\theta}), \Sigma_{t}(\boldsymbol{\theta}))
\end{equation}
is the Gaussian density with mean vector $\mu_{t}(\boldsymbol{\theta})$ and covariance matrix $\Sigma_{t}(\boldsymbol{\theta})$. This product-measure structure allows us to establish consistent tests under a Hellinger-type metric. We remove the conditioning in our final posterior contraction result by a standard total-expectation argument. We define with respect to Lebesgue measure,
\begin{equation*}
    d_{n,\boldsymbol{s}_{n}}^{2}(\boldsymbol{\theta}, \boldsymbol{\theta}') := \frac{1}{n}\sum\limits_{t=1}^{n}h_{t}^{2}(\boldsymbol{\theta}, \boldsymbol{\theta}'), h_{t}^{2}(\boldsymbol{\theta}, \boldsymbol{\theta}') = \int(\sqrt{q_{\boldsymbol{\theta}, \boldsymbol{s}_{n,t}}}-\sqrt{q_{\boldsymbol{\theta}', \boldsymbol{s}_{n,t}}})^{2} d\lambda_{N}(\boldsymbol{x}), \boldsymbol{x}\in \mathbb{R}^{N}. 
\end{equation*}
We also define
\begin{equation*}
    B_{n}(\boldsymbol{\theta}_{0}, \epsilon_{n};k):= \{\boldsymbol{\theta} \in \Theta: K(Q_{\boldsymbol{\theta}_{0}, \boldsymbol{s}_{n}}^{(n)}, Q_{\boldsymbol{\theta}, \boldsymbol{s}_{n}}^{(n)}) \leq n\epsilon_{n}^{2}, V_{k,0}(Q_{\boldsymbol{\theta}_{0}, \boldsymbol{s}_{n}}^{(n)},Q_{\boldsymbol{\theta}, \boldsymbol{s}_{n}}^{(n)}) \leq n^{\frac{k}{2}} \epsilon_{n}^{k} \}, 
\end{equation*}
where $K(f,g) = \int f\log(f/g) d\mu$ and $V_{k,0}(f,g) = \int f|\log(f/g) - K(f,g)|^{k}d\mu,k>1$ are defined with respect to a measure $\mu$. Unconditionally, $\Pi_{n}$ denotes the prior measure; conditional on $\boldsymbol{r}_{1:n}$, it denotes the posterior measure, as seen in Theorem 5. We first prove the following metric entropy result, where $N(a,b,c)$ denotes the covering number of the set $b$ by $c$-balls of radius $a$.
\begin{proposition}
$\sup\limits_{\boldsymbol{s}_{n}\in H_{n}}\sup\limits_{\epsilon>\epsilon_{n}}\log(N(\frac{\epsilon}{36}, \{ \boldsymbol{\theta} \in \Theta_{n}: d_{n,\boldsymbol{s}_{n}}(\boldsymbol{\theta}, \boldsymbol{\theta}_{0})\leq \epsilon\}, d_{n,\boldsymbol{s}_{n}})) \leq n\epsilon_{n}^{2}$, \text{ where } $\epsilon_{n}:= c_{\epsilon}n^{-\frac{1}{2+c_{\boldsymbol{\theta}}}}$, $c_{\epsilon}>0$ is a constant and $c_{\boldsymbol{\theta}}$ is the number of parameters in the model.
\end{proposition}
The proof first shows that, uniformly over $\boldsymbol{s}_{n} \in H_{n}$, the covariance map $\boldsymbol{\theta} \rightarrow \Sigma_{t}(\boldsymbol{\theta})$ satisfies a polynomial-logarithmic Lipschitz bound on $\Theta_{n}$, while the mean is fixed given $\boldsymbol{s}_{n}$. For suitable functions of $n$,  $m_{n}$ and $M_{n}$, we establish $m_{n}I_{N} \preceq \Sigma_{t}(\boldsymbol{\theta}) \preceq M_{n}I_{N}$. This yields a quadratic bound on the Gaussian Kullback--Leibler divergence. Hence $d_{n,\boldsymbol{s}_{n}}(\boldsymbol{\theta}_{1}, \boldsymbol{\theta}_{2}) \lesssim \tilde{L}_{n}||\boldsymbol{\theta}_{1}-\boldsymbol{\theta}_{2}||$ for some $\tilde{L}_{n}$. Therefore, a Euclidean ball of radius $\delta = \epsilon/36\tilde{L}_{n}$ is contained in a $d_{n,\boldsymbol{s}_{n}}$-ball of radius at most $\epsilon/36$. This implies that the metric entropy is bounded by the Euclidean covering number of $\Theta_{n}$. Since $\Theta_{n}$ lies in a $c_{\boldsymbol{\theta}}$-dimensional box with side lengths $O(\log(n))$, the covering number is at most $(C\log(n)/\delta)^{c_{\boldsymbol{\theta}}}$ for some constant $C>0$. Taking logarithms and choosing $c_{\epsilon}$ sufficiently large yields the result. 

\begin{proposition}
    Let $S_{n,j} := \{ \boldsymbol{\theta} \in \Theta_{n}: j\epsilon_{n}< d_{n,\boldsymbol{s}_{n}}(\boldsymbol{\theta}, \boldsymbol{\theta}_{0}) \leq 2j\epsilon_{n} \}$. Then, for every $\boldsymbol{s}_{n} \in H_{n}$, there exists a sequence of tests $\phi_{n,\boldsymbol{s}_{n}}$ and a constant $K>0$ such that $\sup\limits_{\boldsymbol{s}_{n} \in H_{n}} Q_{\boldsymbol{\theta}_{0}, \boldsymbol{s}_{n}}^{(n)} \phi_{n,\boldsymbol{s}_{n}} \rightarrow 0$ and for every $j \geq 2$, $\sup\limits_{\boldsymbol{s}_{n}\in H_{n}} \sup\limits_{\boldsymbol{\theta} \in S_{n,j}}Q_{\boldsymbol{\theta}, \boldsymbol{s}_{n}}^{(n)} (1-\phi_{n,\boldsymbol{s}_{n}}) \leq e^{-Kj^{2}n\epsilon_{n}^{2}}$.
\end{proposition}

We apply Proposition 1 to each shell $S_{n,j}$ with $\epsilon = 2j\epsilon_{n}$ to obtain a $d_{n,\boldsymbol{s}_{n}}$-cover of $S_{n,j}$ by at most $\exp(n\epsilon_{n}^{2})$ balls of radius $j\epsilon_{n}/18$. Lemma 2 of \cite{ghosal2007convergence} yields, for each center, a local test whose type I and II errors are bounded by $\exp\{ -(1/2) j^{2}n\epsilon_{n}^{2}\}$. Taking the maximum over finitely many local tests in the shell preserves the type II bound uniformly over $S_{n,j}$, while a union bound yields a shellwise type I error bound of at most $\exp\{ - (j^{2}/4)n\epsilon_{n}^{2} \}$. Taking the supremum over $j \geq 2$ yields a global test. Since $n\epsilon_{n}^{2} \rightarrow \infty$, the summed shellwise type I error vanishes, while the type II error on each shell remains exponentially small in $j^{2}n\epsilon_{n}^{2}$. 

\begin{proposition}
    For $k =2$ and $0<K< \frac{1}{2}$, $\exists J_{0} \in \mathbb{N}$ such that, $\forall j \geq J_{0}, \forall n \in \mathbb{N},$ and $\forall \boldsymbol{s}_{n} \in H_{n},$
    \begin{equation*}
        \frac{\Pi_{n}(\boldsymbol{\theta}\in\Theta_{n}:j\epsilon_{n} < d_{n,\boldsymbol{s}_{n}}(\boldsymbol{\theta}, \boldsymbol{\theta}_{0})\leq2j\epsilon_{n})}{\Pi_{n}(B_{n}(\boldsymbol{\theta}_{0},\epsilon_{n};k))} \leq \exp\{\frac{1}{2}Kn\epsilon_{n}^{2}j^{2} \}. 
    \end{equation*}
\end{proposition}
Since the numerator is bounded above by 1, it suffices to bound the denominator below. For $\boldsymbol{s}_{n} \in H_{n}$, Section J.4 of the supplement shows that there exist constants $C,a>0$ and $n_{0} \in \mathbb{N}$ such that, $\forall n\geq n_{0}$, $U_{n} := \{\boldsymbol{\theta} \in \Theta: ||\boldsymbol{\theta} - \boldsymbol{\theta}_{0}||_{2} \leq \epsilon_{n} C^{-1}n^{-a} \} \subseteq B_{n}(\boldsymbol{\theta}_{0}, \epsilon_{n};2).$ Let $r_{n} := \epsilon_{n}(C\sqrt{c_{\boldsymbol{\theta}}}n^{a})^{-1}$. We lower bound the prior mass of $U_{n}$ using the coordinatewise neighborhoods of radius $r_{n}$. Each pair $(\mu_{\alpha_{i}}, \mu_{\beta_{i}})$ is treated jointly because it shares $\tau_{i,0}$. The remaining scalar parameters and loading matrix are treated using their local prior densities. Thus, for some fixed $\tilde{m}>0$, $
    \Pi_{n}(B_{n}(\boldsymbol{\theta}_{0}, \epsilon_{n};2)) \geq \Pi_{n}(U_{n}) \geq (\tilde{m}n^{-(a+\frac{1}{2})}\epsilon_{n})^{c_{\boldsymbol{\theta}}}. $
Hence, $-\log(\Pi_{n}(B_{n}(\boldsymbol{\theta}_{0}, \epsilon_{n};2))) = O(\log(n))$, where $n\epsilon_{n}^{2} = c_{\epsilon}^{2}n^{\frac{c_{\boldsymbol{\theta}}}{2+c_{\boldsymbol{\theta}}}}$. Thus, Condition 2.5 of \cite{ghosal2007convergence} holds for all sufficiently large $j$, uniformly over $n \geq n_{0}$. The finitely many $n<n_{0}$ follow from $d_{n,\boldsymbol{s}_{n}} \leq \sqrt{2}$; see the full proof. 
\begin{proposition}
$\mathbb{P}_{\boldsymbol{\theta}_{0}}(H_{n}^{c}) \rightarrow 0.$
\end{proposition}
Let $H_{n} = A_{M,n} \cap A_{\alpha,n} \cap A_{\beta,n}\cap A_{\tilde{h},n}\cap A_{\bar{h},n}$. By a union bound, it suffices to show that each complement has vanishing probability. For the Gaussian AR(1) log-volatility processes $(h_{M}, \tilde{h}, \bar{h})$, Gaussian tail bounds and a union bound over $t$ imply maxima of order $\log(n)$ with probability tending to one. For the $Z$-innovation AR(1) process $h_{\alpha_{a},t}$, the moment generating function $M(t) = \sec(\pi t)$ yields exponential Chernoff tails, and hence vanishing tail probability. For $\beta_{i,t}$, we first control $h_{\beta_{i}}$ on a high-probability event $G_{n}$, which bounds the conditional variance of $\beta_{a,t}-\beta_{a,t,0}$, and then Gaussian tail bounds and a union bound show that the complement probability vanishes. We now state the main result.  

\begin{Theorem}\label{thm:correlation_posterior_contraction}
    For every $M_{n} \rightarrow \infty,\Pi_{n}(\boldsymbol{\theta} \in \Theta_{n}:d_{n,\boldsymbol{S}_{0,n}}(\boldsymbol{\theta},\boldsymbol{\theta}_{0})\geq M_{n}\epsilon_{n}|\boldsymbol{r}_{1:n}, r_{M,1:n}) \xrightarrow{\mathbb{P}_{\boldsymbol{\theta}_{0}}} 0$.
\end{Theorem}
Note $N$ and $r$ are fixed as $n\rightarrow \infty$. For any fixed latent path $\boldsymbol{s}_{n}$, Theorem 3 of \cite{ghosal2007convergence} implies that the posterior mass outside the $d_{n,\boldsymbol{s}_{n}}$-ball of radius $M_{n}\epsilon_{n}$ has vanishing expectation uniformly over $H_{n}$. Since this posterior mass lies in $[0,1]$, Markov's inequality yields conditional convergence in probability uniformly over $\boldsymbol{s}_{n} \in H_{n}$. To extend this to the random true path $\boldsymbol{S}_{0,n}$, we decompose over $H_{n}$ and $H_{n}^{c}$. On $H_{n}$, the uniform bound controls the conditional error, while the contribution from $H_{n}^{c}$ is bounded by $\mathbb{P}_{\boldsymbol{\theta}_{0}}(H_{n}^{c}) \rightarrow 0$. Since both terms vanish, posterior contraction follows in $\mathbb{P}_{\boldsymbol{\theta}_{0}}$-probability. 

Theorem \ref{thm:correlation_posterior_contraction} establishes posterior contraction of DSP-MFSV-CAPM. As the sample size grows, DSP-MFSV-CAPM consistently recovers the true parameters. This provides theoretical support for its use in modeling stock portfolios. Beyond the observation equation and prior specification, the theorem imposes no additional structural assumptions. The result concerns contraction around the true parameter $\boldsymbol{\theta}_{0}$. Under model misspecification, however, $\boldsymbol{\theta}_{0}$ should be interpreted as the pseudo-true value rather than the data-generating truth. 
\section{Numerical Results}
\subsection{Simulation study}
We consider five simulation designs to evaluate DSP-MFSV-CAPM against competing methods. Each design simulates excess returns for $N= 30$ assets over $T=1000$ time steps. For each of 100 seeds, we compute full-sample performance metrics and average them across seeds. We report the root mean squared error (RMSE) of the estimated score path, mean interval width (width), and empirical coverage (coverage). These metrics assess point estimation and uncertainty quantification: lower RMSE indicates greater accuracy, narrower intervals indicate greater sharpness, and coverage near the nominal $95\%$ level indicates better calibration. Coverage and width should be assessed jointly. An interval spanning the parameter space would attain $100\%$ coverage but be uninformative. Conversely, a narrow interval with low coverage is poorly calibrated. Code is available at \url{https://github.com/DanielCoulson/Modeling-Dynamic-Correlation-Matrices-with-Shrinkage-Priors}.

Scenario 4 reports mean absolute error (MAE), while the remaining scenarios report two MAE-based summaries. Steady-state MAE (MAE1) is computed over post-break time points 21--70 and measures accuracy after the new regime has stabilized. Transient MAE (MAE2) is computed over the first 50 post-break time points and measures accuracy immediately after a regime change. Because the rolling methods produce no estimates for the first 59 periods, all metrics use time points 60-1000. In Scenarios 1--3 and 5, we also report mean response lag, the number of post-break periods until the estimated score moves at least $50\%$ of the true jump from its pre-break level in the correct direction. Mean settling lag is the number of periods after a break until the estimated score first enters a tolerance band equal to $10\%$ of the jump size around the truth and remains within that band for five consecutive periods. If the thresholds are not reached within the observed post-break period, the corresponding lags are reported as $\infty$.

Scenario 1 considers abrupt breaks in asset dependence under a one-factor return model.
Returns are generated as $r_{a,t} = \alpha_{a} + \beta_{a} r_{M,t} + \epsilon_{a,t}$, where $\alpha_{a} \sim  N(0,0.03^{2})$, $\beta_{a} \sim \text{Unif}(0.7,1.4)$, $r_{M,t} \sim N(0,\sigma_{M,t}^{2})$, and $\boldsymbol{\epsilon}_{t} \sim N_{N}(\boldsymbol{0}_{N}, \Sigma_{\epsilon,t})$, with $\Sigma_{\epsilon,t}$ calibrated so the implied asset correlation matrix follows the desired regime-specific pattern. We impose four 250-period regimes, with breaks at $t = 250$, $500$, and $750$, and market-factor standard deviations $(0.80, 2.20, 1.10, 1.60)$, corresponding to low equicorrelation $(\rho = 0.10)$, a crisis jump to high equicorrelation ($\rho = 0.70$), a reversion to moderate equicorrelation ($\rho = 0.25$), and a three-block dependence structure with 10 assets per block and within- and between-block correlations $0.65$ and $0.35$. Each asset standard deviation is drawn once, $\sigma_{a}\sim \text{Unif}(1.20, 2.00)$, for $a= 1,\dots,N$, and scaled by the regime-specific multipliers $(1,1.35, 1.10, 1.25)$. 

Scenario 2 studies sparse, time-varying dependence with repeated crisis clustering. Returns follow the same one-factor model, $r_{a,t} = \alpha_{a} + \beta_{a}r_{M,t} + \epsilon_{a,t}$, $\alpha_{a} \sim N(0,0.03^{2})$, $\beta_{a} \sim \text{Unif}(0.7,1.4)$, $r_{M,t} \sim N(0, \sigma_{M,t}^{2})$, $\sigma_{a} \sim \text{Unif}(1.20,2.00)$, and $\boldsymbol{\epsilon}_{t}\sim N_{N}(\boldsymbol{0}_{N},\Sigma_{\epsilon,t})$, with $\Sigma_{\epsilon,t}$ calibrated so that the implied asset correlation matrix follows the desired regime-specific pattern. Assets are partitioned into three groups of 10. In calm periods, the correlation matrix is sparse: within-group correlation is 0.20 and between-group correlations are 0. Crisis episodes occur in periods 251--350, 601--700, and 851--925. During crises, within-group correlation rises to 0.75, correlation between groups 1 and 2 increases to 0.35, and other between-group correlations are set to 0.05. Market volatility rises from 0.80 to 2.20 in crises, and asset standard deviations are scaled by 1.00 in calm periods and 1.25 in crises. For each state, the target covariance matrix is constructed from the corresponding regime-specific correlations and volatilities. 

Scenario 3 studies abrupt changes in latent factor dimension and loading structure under a multi-factor return model. Returns are generated as 
\begin{equation*}
    r_{a,t} = \alpha_{a} + \beta_{a}r_{M,t} + \gamma_{a}z_{2,t} + \delta_{a}z_{3,t} + \eta_{a} z_{4,t} + \theta_{a} z_{5,t} + \kappa_{a} z_{6,t} + e_{a,t},
\end{equation*}
with $\alpha_{a} \sim N(0,0.03^{2})$, $\beta_{a} \sim \text{Unif}(0.7,1.4)$, $r_{M,t} \sim N(0,\sigma_{M,t}^{2})$, $z_{j,t} \sim N(0, \sigma_{j,t}^{2})$,  $ \qquad \qquad \boldsymbol{e}_{t} \sim N_{N}(\boldsymbol{0}_{N},D_{t}^{2})$, and baseline idiosyncratic standard deviations $\sigma_{e,a} \sim \text{Unif}(0.70,1.20)$, where $D_{t}^{2} = c_{t}^{2}\text{diag}(\sigma_{e,1}^{2},\dots, \sigma_{e,30}^{2})$, with $c_{t} \in \{1, 1.35, 1.10 \}$ corresponding to three different regimes. We partition assets into three groups $G_1,G_2,G_3$ of 10 and consider regimes of lengths 300, 400, and 300. In particular, $\gamma_{a} \sim \boldsymbol{1}\{ a \in G_{1} \cup G_{2} \}\max\{ 0, N(0.75,0.08^{2}) \}$, $\delta_{a} \sim \boldsymbol{1}\{ a \in G_{1} \} \max \{ 0, N(0.90, 0.08^{2})\}$, $\eta_{a} \sim \boldsymbol{1}\{ a \in G_{2} \} \max\{ 0, N(0.80, 0.08^{2}) \}$, $\theta_{a} \sim \boldsymbol{1}\{ a \in G_{3} \} \max\{ 0, N(0.80, 0.08^{2}) \}$, $\kappa_{a}\sim \max\{0, N(0.65, 0.08^{2}) \}.$ The first regime is a calm one-factor regime with only the market factor active and $\sigma_{M,t} = 0.80$. The second is a severe six-factor regime with $\sigma_{M,t} = 2$, in which $z_{2,t}, z_{3,t}, z_{4,t}, z_{5,t}$, and $z_{6,t}$ are active with standard deviations (1.10, 0.95, 0.85, 0.80, 0.75), respectively. The final regime is a milder two-factor period with only $r_{M,t}$ and $z_{2,t}$ active, with standard deviations 1.20 and 0.65, respectively, and the sector loading falls to $70\%$ of its regime 2 value.

Scenario 4 studies smoothly evolving dependence under a correctly specified DCC model with heavy-tailed innovations. Let $\boldsymbol{w}_{t} = (r_{M,t}, r_{1,t}, \dots, r_{N,t})^\top$ denote the joint market and excess-return vector. We generate $\boldsymbol{w}_{t} = \boldsymbol{\mu} + D\boldsymbol{z}_{t}$, where $\boldsymbol{z}_{t}|\mathcal{F}_{t-1}$ follows a multivariate Student-t distribution with 8 degrees of freedom and conditional correlation matrix $R_{t}$. Then $R_{t}$ evolves according to a DCC(1,1) recursion with parameters $(a,b) = (0.03, 0.95)$. The unconditional joint correlation target is calibrated so that the long-run market-asset correlation is 0.35 and the asset-asset correlation matrix is equicorrelated  with $\rho = 0.15$. Marginal scales are heterogeneous: market standard deviation is $1.20$, asset standard deviations are drawn independently from $\text{Unif}(0.80, 1.40)$, and asset-specific means $\alpha_{a} \sim N(0,0.03^{2})$. After a 300-period burn-in, we retain 1,000 observations. 

Scenario 5 retains Scenario 1's four-regime dependence structure but allows time-varying factor-model coefficients. Each asset's $\alpha_{a,t}$ and $\beta_{a,t}$ follow persistent Gaussian AR(1) processes around long-run means $\bar{\alpha}_{a} \sim N(0,0.03^{2})$ and $\bar{\beta}_{a} \sim \text{Unif}(0.7,1.4)$. Persistence is $\phi_{\alpha} = \phi_{\beta} = 0.995$. Calm-period innovation standard deviations are 0.004 for $\alpha_{a,t}$ and 0.012 for $\beta_{a,t}$, scaled by regime-dependent multipliers $(1.0, 3.0, 1.2, 1.8)$. Market factor variance also follows a stochastic volatility process, so dependence changes reflect both abrupt regime shifts and smooth variation in factor exposures and market risk. 

We benchmark DSP-MFSV-CAPM against nine alternatives. The first is a multivariate factor stochastic volatility (MFSV) model fitted directly to returns using 3 latent factors, 12,000 MCMC draws, 1,500 burn-in iterations, and a thinning rate of 4. The second is an exponentially weighted rolling covariance estimator (Rolling) using a 60-period window and decay factor 0.94. The third is a rolling Ledoit-Wolf shrinkage estimator (Ledoit-Wolf), also using a 60-period window. The fourth through sixth are multivariate GARCH-based dynamic correlation models with constant-mean univariate sGARCH(1,1) margins and DCC(1,1)-type correlation dynamics: Gaussian DCC model (DCC), Student-t DCC model (DCC-t), and asymmetric DCC (ADCC). The seventh is a Bayesian Cholesky stochastic volatility model (Chol-SV), implemented as a VAR(1) with stochastic volatility and a horseshoe prior on the Cholesky coefficients, using 12,000 MCMC draws, 1,500 burn-in iterations, and a thinning rate of 4. The eighth is a hard-thresholded rolling correlation estimator (Rolling-threshold), using a 60-period rolling sample correlation matrix, thresholding off-diagonal entries below 0.10 in absolute value, and projecting to the nearest positive definite correlation matrix. The ninth is a local group graphical lasso (Loggle,\citet{yang2020estimating}) estimator of time-varying precision matrices, fitted by pseudo-likelihood with refitting, with bandwidth \(h\), neighborhood width \(d\), and sparsity \(\lambda\) selected by three-fold cross-validation over \(h\in\{0.20,0.25\}\), \(d\in\{0,0.05,0.15,1\}\), and \(\lambda\in\{0.20,0.25\}\) at nine equally spaced time points. For DSP-MFSV-CAPM, we use 3 latent factors and retain 3,000 posterior draws after 1,500 burn-in iterations, with a thinning rate of 4. For the rolling and Loggle, intervals use a moving-block bootstrap with 200 resamples and block length 20, whereas uncertainty for the Bayesian procedures is derived from posterior output, using $95\%$ credible intervals for DSP-MFSV-CAPM, and normal approximations for MFSV and Chol-SV. Tables report metrics averaged over 100 seeds. DCC interval estimates follow \citet{fresoli2016uncertainty}. 
\begin{table}[t]
\centering
\caption{Results for Scenario 1}
\label{tab:sim1}
\begin{tabular}{lccccccc}
\hline
Model & RMSE & width & coverage & response lag & MAE1 & settling lag  & MAE2  \\
\hline
MFSV & 0.061 & 0.360 & 0.981 & 1.530 & 0.042 & 15.20& 0.061  \\
Rolling &0.292  & 0.381  & 0.531 & 1.130 & 0.239 & 177.0& 0.255  \\
Loggle & 0.198 & 0.250 & 0.249 & 34.27 & 0.156 & 168.2 & 0.178 \\
Ledoit-Wolf &0.126  &0.553  & 0.766 & 26.00 & 0.125 & 57.30&0.246  \\
DCC & 0.201 &0.032  & 0.001& 21.68 & 0.205 & $\infty$& 0.232 \\
DCC-t & 0.210 & 0.024 & 0.142 & 17.82 & 0.206&  $\infty$&0.227 \\
ADCC &0.204 & 0.030 & 0.002& 20.63 & 0.199 & 642.1 &0.225  \\
Chol-SV &0.197 &0.220 & 0.327& 3.860 & 0.138& 352.0&0.150  \\ 
Rolling-threshold &0.272 & 0.357 & 0.531 & 34.10 & 0.202& 94.30&0.222  \\ 
\hline
DSP-MFSV-CAPM & 0.057  & 0.217  & 0.952 & 1.040 & 0.040&14.10&0.057    \\
\hline
\end{tabular}
\end{table}
\begin{table}[t]
\centering
\caption{Results for Scenario 2}
\label{tab:sim2}
\begin{tabular}{lccccccc}
\hline
Model & RMSE & width & coverage & response lag & MAE1 & settling lag & MAE2  \\
\hline
MFSV & 0.059 & 0.345  &0.957 &0.505   &  0.030 &8.530  & 0.046 \\
Rolling & 0.363  & 0.383 & 0.292 &0.000  &0.287 & 219.0 &0.329  \\
Loggle & 0.215 & 0.305 & 0.004 & 0.000 & 0.141 & $\infty$ & 0.142\\
Ledoit-Wolf & 0.201 & 0.541  &0.708 &29.50  & 0.158  &59.20 &0.319  \\
DCC & 0.228 & 0.014  & 0.001 &14.06  & 0.250& $\infty$&0.267  \\
DCC-t & 0.233 & 0.012 & 0.000 &2.572   &0.255 & $\infty$&0.267   \\
ADCC & 0.230 & 0.013 &0.000& 31.23  & 0.248&$\infty$ &0.264 \\
Chol-SV & 0.212 & 0.258 & 0.101 & 2.110  &0.191 &  $\infty$&0.203 \\ 
Rolling-threshold & 0.437 & 0.622 & 0.292 &  31.00&0.444&231.0 & 0.495  \\ 
\hline
DSP-MFSV-CAPM & 0.064  & 0.187   & 0.929 &0.562  &0.039 &11.70 &0.056   \\
\hline
\end{tabular}
\end{table} 
\begin{table}[t]
\centering
\caption{Results for Scenario 3}
\label{tab:sim3}
\begin{tabular}{lccccccc}
\hline
Model & RMSE & width & coverage & response lag & MAE1 & settling lag & MAE2 \\
\hline
MFSV & 0.054  & 0.402  &0.996          & 3.340& 0.036   & 28.60 & 0.044   \\
Rolling & 0.191  & 0.246  & 0.425     & 26.90& 0.158  & 65.60& 0.165 \\
Loggle & 0.057 & 0.136 & 0.169 & 0.000 & 0.030 & 87.84 & 0.035\\
Ledoit-Wolf & 0.070  & 0.275&0.740     &27.70&0.066    & 83.10&0.117\\
DCC & 0.114 & 0.045  & 0.000          & 4.556 &0.083  & 338.0& 0.092\\
DCC-t & 0.120 & 0.030&0.000           & 2.214&0.079& 326.8&0.085\\
ADCC & 0.114 & 0.046 & 0.000          &3.541&0.083  & 382.9&0.092\\
Chol-SV & 0.084 & 0.206 & 0.758       &1.330&0.043  &121.0 &0.053  \\ 
Rolling-threshold &0.130 &0.235&0.441  &57.60&0.112   & 39.90&0.129 \\ 
\hline
DSP-MFSV-CAPM & 0.049&0.202&0.954     & 4.060&0.033& 27.50& 0.042\\
\hline
\end{tabular}
\end{table}
\begin{table}[t]
\centering
\caption{Results for Scenario 4}
\label{tab:sim4}
\begin{tabular}{lcccc}
\hline
Model & RMSE & width & coverage & MAE    \\
\hline
MFSV &0.146  &0.046  &0.271&0.139  \\
Rolling &0.402  & 0.184 &0.000&0.401  \\
Loggle & 0.053 & 0.096 & 0.128 & 0.033\\
Ledoit-Wolf &0.053  &0.222  &0.813&0.041 \\
DCC & 0.018 &0.047  &0.369&0.016 \\
DCC-t & 0.025 &0.047 &0.695&0.024\\
ADCC & 0.019& 0.047&0.405&0.016 \\
Chol-SV & 0.120 &0.245 &0.784&0.114\\ 
Rolling-threshold &0.582 & 0.301& 0.000&0.579\\ 
\hline
DSP-MFSV-CAPM &0.117  & 0.259  & 0.670&0.109   \\
\hline
\end{tabular}
\end{table}

\begin{table}[t]
\centering
\caption{Results for Scenario 5}
\label{tab:sim5}
\begin{tabular}{lccccccc}
\hline
Model & RMSE & width & coverage & response lag & MAE1 & settling lag  & MAE2  \\
\hline
MFSV & 0.065  &0.375  & 0.987  & 1.710  &0.044  & 21.80 & 0.075   \\
Rolling & 0.297 & 0.389  &0.531 & 19.70  &0.225  & 86.30 & 0.231  \\
Loggle & 0.110 & 0.306 & 0.156 & 0.000 & 0.091 & 248.3 & 0.106 \\
Ledoit-Wolf &0.135  &0.570 & 0.694&30.60 &0.146  &63.80 & 0.281 \\
DCC & 0.207  & 0.028  &0.080 &33.63  &0.217  & $\infty$ & 0.244 \\
DCC-t  &0.219&0.020  &0.032  &36.32  &0.217 & $\infty$ &  0.235 \\
ADCC &0.207 & 0.027  & 0.086&28.11 &0.213  & 697.0  & 0.239\\
Chol-SV & 0.211&0.274 & 0.316 &2.020  &0.149 & 382.0 &0.160  \\ 
Rolling-threshold &0.290 &0.389  &0.530  &37.90  &0.195 & 97.20 & 0.211\\ 
\hline
DSP-MFSV-CAPM &0.064   &0.226   &0.947  &1.730  & 0.044& 22.20 &0.074    \\
\hline
\end{tabular}
\end{table}

Overall, DSP-MFSV-CAPM performs strongly under abrupt, nonstationary changes in dependence and remains competitive under smooth dynamics. In Scenario 1, it has the lowest RMSE and MAE summaries and the shortest response and settling lags. Its intervals are also well calibrated, with coverage near the nominal $95\%$ level and substantially narrower intervals than competing methods, indicating a favorable coverage--width tradeoff. DCC-type methods produce narrow intervals but with very low coverage. In Scenario 2, which emphasizes sparse and localized correlation changes rather than broader market-wide changes in correlation, DSP-MFSV-CAPM no longer ranks first on every metric; MFSV performs slightly better in RMSE and MAE. Nonetheless, DSP-MFSV-CAPM remains competitive overall. In Scenario 3, DSP-MFSV-CAPM again has the lowest RMSE and MAE and narrow intervals with near-nominal coverage. Although some competitors have shorter response lags in this design, these gains do not translate into improved estimation accuracy or interval performance. Scenario 4 is less favorable for the proposed method because the data follows a correctly specified heavy-tailed DCC process with smooth dynamics and no breaks; accordingly, the DCC models and Ledoit-Wolf perform best. Even in that setting, however, DSP-MFSV-CAPM improves on MFSV and Chol-SV in RMSE and MAE, outperforms the rolling estimators, and provides better interval coverage. In Scenario 5, DSP-MFSV-CAPM again performs competitively, achieving the lowest RMSE and MAE, interval coverage near the nominal $95\%$ level, and substantially narrower intervals than competitors. These results indicate that DSP-MFSV-CAPM performs well in the break-driven environments with broad changes in correlation that motivate the proposed method, while remaining competitive under smooth dynamics that favor competing models. 

\subsection{Real-world examples}
\subsubsection{U.S. subprime mortgage crisis}
The U.S. subprime mortgage crisis of 2007--2010 followed the collapse of the U.S. housing bubble. Mortgage securitization and collateralized debt obligations (CDOs) played a central role, as highly rated instruments proved substantially riskier than expected. We consider two 30-stock portfolios. The first portfolio comprises large technology-oriented stocks. The second is more diversified, spanning 10 sectors and a range of firm sizes. We compute excess returns using data from Yahoo Finance and the Fama-French data library \cite{fama2023production}. We fit DSP-MFSV-CAPM to returns from 4 January 2006 to 31 December 2009. We retain 3,000 posterior draws of the score series. Figure 1 shows the posterior mean score for both portfolios, and Figure 1b shows the diversified portfolio's posterior mean score, $95\%$ highest density intervals (HDIs), and the VIX.
\begin{figure}[htbp]
    \centering
    \begin{subfigure}[t]{0.48\linewidth}
        \centering
        \includegraphics[width=\linewidth]{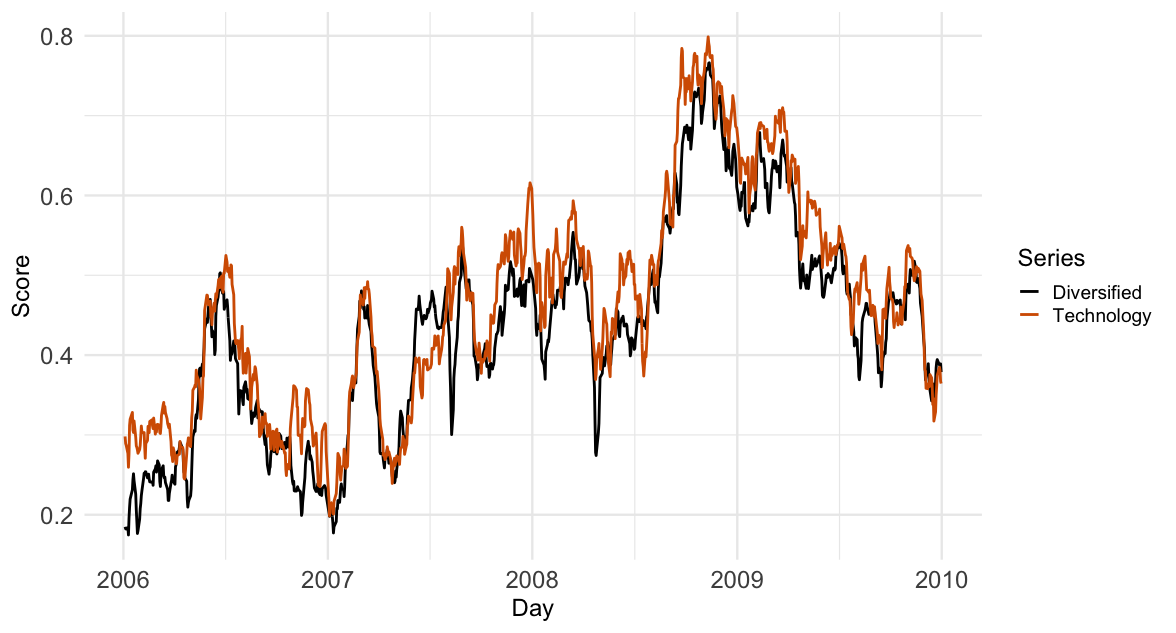}
        \caption{Posterior mean score time series for the technology portfolio
        (vermilion) and the diversified portfolio (black).}
        \label{fig:2008_tech}
    \end{subfigure}
    \hfill
    \begin{subfigure}[t]{0.48\linewidth}
        \centering
        \includegraphics[width=\linewidth]{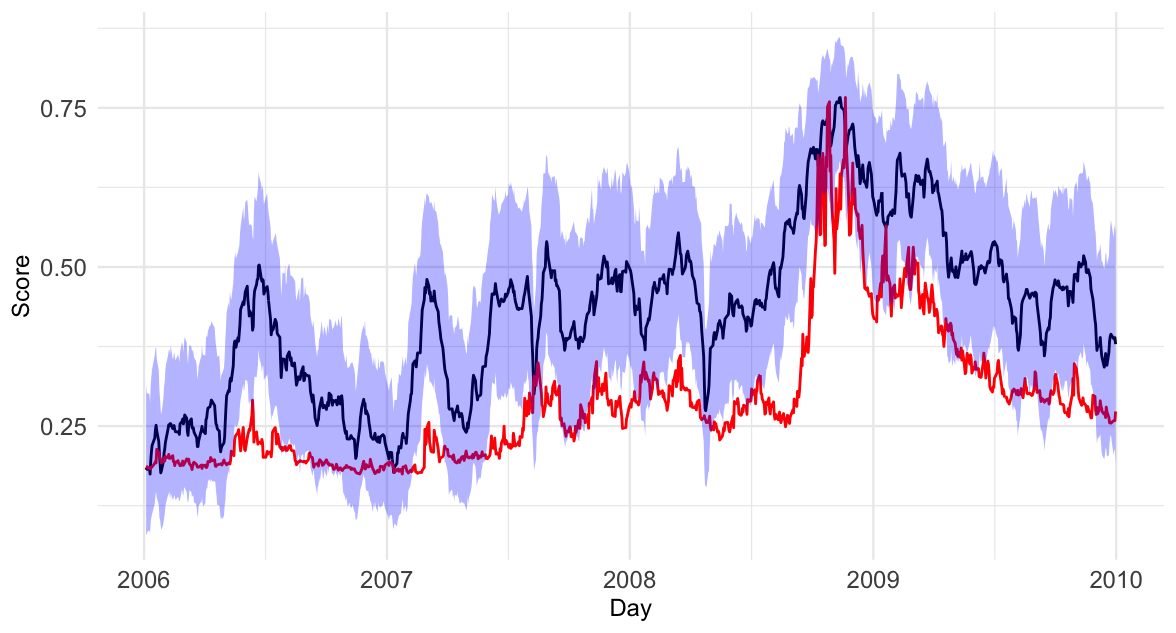}
        \caption{Posterior mean score time series for the diversified portfolio
        (black) with 95\% HDIs and the VIX index (red).}
        \label{fig:2008_vix}
    \end{subfigure}
    \caption{Posterior mean score comparisons for the diversified portfolio.}
    \label{fig:2008_comparisons}
\end{figure}

Figure 1a shows that the technology portfolio generally has higher overall correlation than the diversified portfolio. For example, this is evident in mid-2006. This pattern changes in 2007; for example, in the first quarter of 2007, both portfolios show a correlation spike, with the diversified portfolio exceeding the technology portfolio. Overall, their correlation dynamics are similar during the crisis. Both series culminate in pronounced spikes, peaking on 10 and 12 November 2008 for the technology and diversified portfolios, respectively. 

Figure 1b shows several spikes, including in June--July 2006, coinciding with an increase in the federal funds rate and emerging CDO-market difficulties, including growth in CDO-linked credit default swaps. The series spikes again in March 2007, amid growing warnings of financial instability, including a speech by Ben Bernanke discussing systemic risk associated with Fannie Mae and Freddie Mac. Correlation then rises persistently from around June 2007 into 2008. The diversified portfolio score is about 0.372 on 1 May 2008. It then rises rapidly during September 2008, amid the government takeover of Fannie Mae and Freddie Mac and Lehman Brothers' bankruptcy, reaching 0.614 on 15 September 2008. The series then fluctuates but remains elevated, peaking at 0.766 on 12 November 2008. This may reflect the aftermath of these events and uncertainty surrounding the policy response to the crisis, including government lending and refinancing measures. 

\subsubsection{2020 COVID-19 pandemic}
The COVID-19 pandemic began in Wuhan, China, in December 2019. It quickly spread, prompting widespread public-health restrictions and substantial economic disruption. We consider two 30-stock portfolios. The first comprises large technology-oriented stocks. The second is more diversified, spanning 10 sectors and a range of firm sizes. We compute excess asset returns using data from Yahoo Finance and the Fama-French data library \cite{fama2023production}. We use daily data from 3 January 2019 to 28 December 2023. We fit DSP-MFSV-CAPM and obtain 3,000 posterior draws of the score series. Figure 2a shows the posterior mean score for both portfolios, and Figure 2b shows the diversified portfolio's posterior mean score, $95\%$ HDIs, and the VIX.
\begin{figure}[htbp]
    \centering

    \begin{subfigure}[t]{0.48\linewidth}
        \centering
        \includegraphics[width=\linewidth]{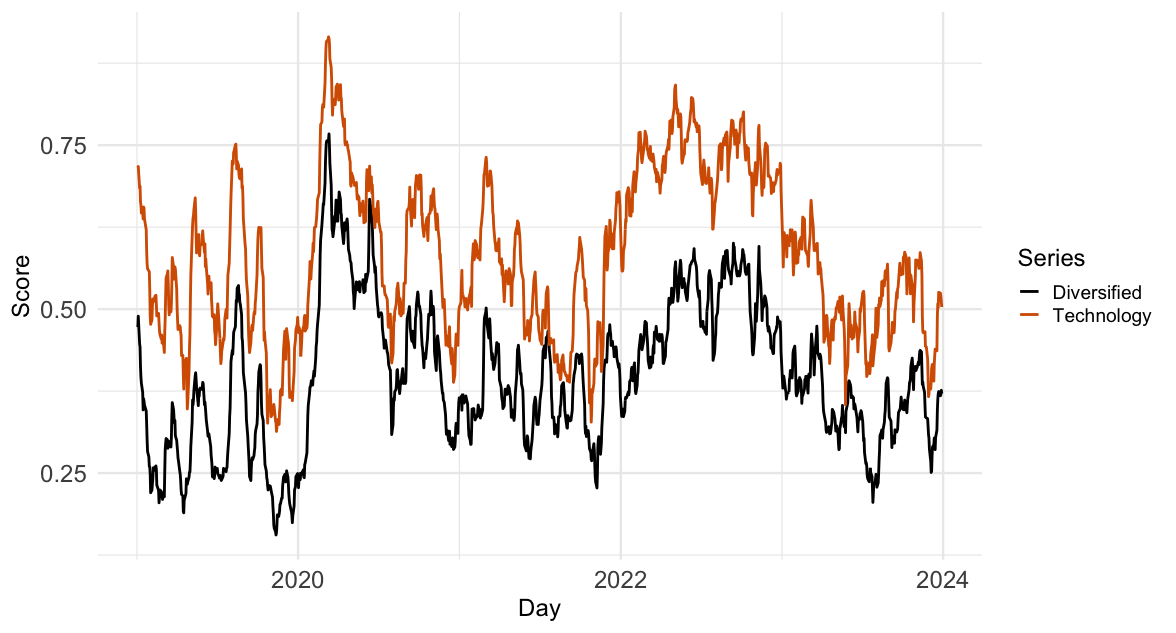}
        \caption{Posterior mean score time series for the technology portfolio
        (vermilion) and the diversified portfolio (black).}
        \label{fig:2020_tech}
    \end{subfigure}
    \hfill
    \begin{subfigure}[t]{0.48\linewidth}
        \centering
        \includegraphics[width=\linewidth]{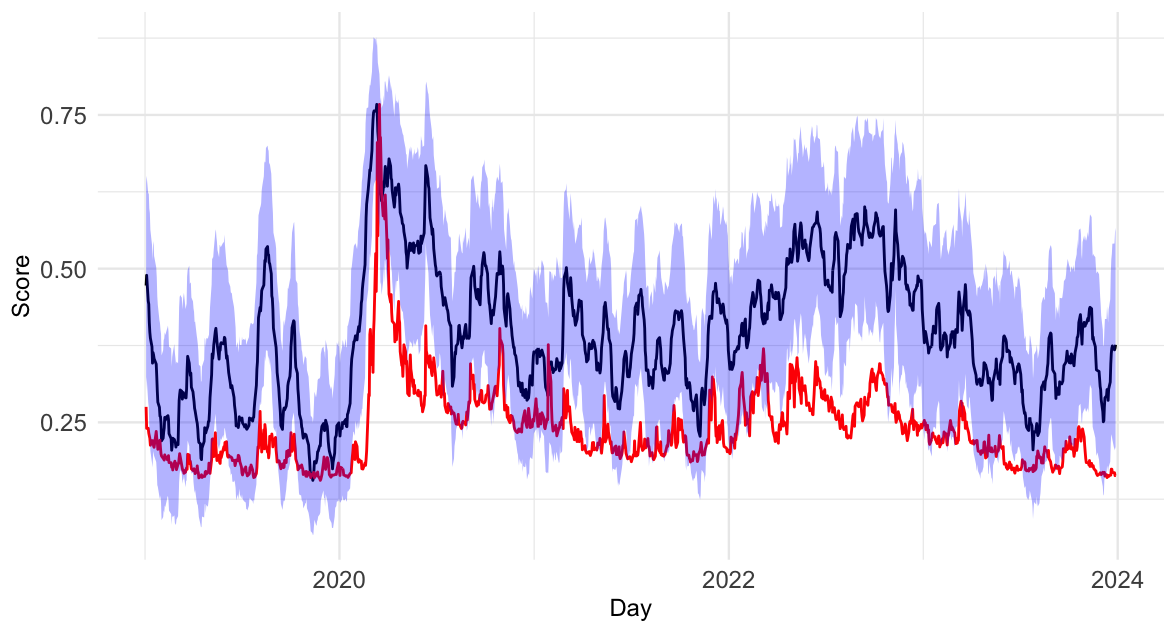}
        \caption{Posterior mean score time series for the diversified portfolio
        (black) with 95\% HDIs and the VIX index (red).}
        \label{fig:2020_vix}
    \end{subfigure}

    \caption{Posterior mean score comparisons for the diversified portfolio.}
    \label{fig:2020_comparisons}
\end{figure}

Figure 2a shows large spikes in the technology and diversified portfolio scores on 10 and 11 March 2020, respectively, shortly before widespread national lockdowns, 
including across the United Kingdom, Europe, and several U.S. states, and the World Health Organization's (WHO's) pandemic declaration on 11 March 2020. Initially, the diversified portfolio has a lower score than the technology portfolio, but both exhibit similar correlation spikes.  

 Figure 2b shows several spikes. An early spike reaches 0.536 on 19 August 2019. This may reflect several developments, including new U.S. tariffs on Chinese goods, signs of a global downturn, and the Federal Reserve's late-July rate cut. The first Federal Reserve rate cut since 2008 may have signaled growing economic concern. The period also coincided with other signs of possible economic weakness, including an inverted U.S. Treasury yield curve. International stresses were also present, including the sharp decline in the Argentine stock market on 12 August 2019. The series rose in early 2020 following initial COVID-19 reports and the WHO's declaration of a public health emergency on 30 January 2020, by which point the estimated score had reached 0.386. The subsequent increase coincides with major February developments, including rising global death counts and sharp stock-market declines from 24--28 February. The Federal Reserve cut rates on 3 March amid growing COVID-19 concerns, and U.S. stocks fell sharply on 9 March 2020. The score peaks at 0.767 on 11 March 2020. The VIX peaks later, on 16 March 2020. 

Both examples show that the score can serve as a portfolio-level indicator of financial stress in retrospective backtests, closely tracking and sometimes preceding VIX spikes. This suggests that the score captures changes in portfolio-wide correlation before they are fully reflected in broader market measures, providing a portfolio-specific view of market instability. The examples also suggest that simple stock diversification offered little protection, with both portfolios similarly affected during downturns. These results suggest that the method may be useful for portfolio managers as a backtesting tool to assess how diversification strategies would have performed under crisis conditions and identify when alternative diversification may have been warranted. 

\section{Conclusion}\label{sec-conc}
We proposed a novel approach for estimating the evolution of the overall correlation over time. To do so, we developed an information-theoretically grounded score to summarize a given correlation matrix, and we constructed a Bayesian hierarchical model with dynamic shrinkage to model time-varying correlation matrices. Through a posterior contraction result, we established theoretical support for the model. Across five simulation scenarios, our approach achieved strong performance relative to competing methods. In two real-world examples, we showed how our model can be used in practice to provide novel insights for investors and economic historians. Future work could build on the theoretical results in this paper, for example by considering alternative separation metrics or modes of stochastic convergence. Another direction is to explore computational approaches for extending the hierarchical model to higher dimensions and develop more scalable posterior approximations for larger cross sections.

\clearpage
\appendix

\setcounter{section}{0}
\setcounter{subsection}{0}
\setcounter{subsubsection}{0}
\setcounter{equation}{0}
\setcounter{figure}{0}
\setcounter{table}{0}

\renewcommand{\thesection}{\Alph{section}}
\renewcommand{\thesubsection}{\thesection.\arabic{subsection}}
\renewcommand{\thesubsubsection}{\thesection.\arabic{subsection}.\arabic{subsubsection}}

\numberwithin{equation}{section}
\numberwithin{figure}{section}
\numberwithin{table}{section}

\part*{Supplementary Material}

\etocsetlocaltop.toc{part}

\etocsettocstyle{\section*{Supplementary Contents}}{}
\localtableofcontents
\clearpage
\section{Sensitivity study for number of latent factors}
We use three latent factors in our model to balance computational complexity and parsimony. Below, we provide results from a small sensitivity study. In particular, we repeat simulation scenario 1 using only five assets and increase the number of latent factors from one to five. We then compare the average RMSE across 30 seeds.

\begin{table}[ht]
\centering
\caption{RMSE by Number of Latent Factors}
\label{tab:factor_rmse}
\begin{tabular}{cc}
\hline
\textbf{Number of Factors} & \textbf{RMSE} \\
\hline
1 & 0.081 \\
2 & 0.077 \\
3 & 0.073 \\
4 & 0.071 \\
5 & 0.071 \\
\hline
\end{tabular}
\end{table}

We perform two hypothesis tests to assess the choice of three latent factors. First, using a paired Monte Carlo sign-flip permutation test, we test whether the three-factor specification yields a lower RMSE than the one- and two-factor specifications. RMSE differences are calculated within each of the 30 simulation seeds. For each comparison, the mean paired RMSE difference is approximated using 1,000,000 random sign flips. The three-factor model reduces mean RMSE by 0.00807 relative to the one-factor model (p < 0.000001) and by 0.00322 relative to the two-factor model (p = 0.000409), indicating significantly better performance.

Second, using a paired Monte Carlo non-inferiority permutation test with a prespecified RMSE margin of $\Delta$ = 0.005, representing the largest increase considered practically negligible, we test whether the three-factor specification has an increase in RMSE of less than 0.005 relative to the four- and five-factor specifications. RMSE differences are calculated within each of the 30 simulation seeds, and the null distribution is approximated using 1,000,000 random sign flips. The three-factor model has a mean RMSE only 0.00183 higher than the four-factor model (p < 0.000001) and 0.00268 higher than the five-factor model (p = 0.000103). Both differences are significantly below the prespecified non-inferiority margin, supporting the three-factor model as a parsimonious specification.

\section{Identification}
Following \cite{factorstochvol}, we normalize the scale of the latent factors by centering each factor log-volatility process at zero: 

\begin{equation*}
    \tilde{h}_{t,j} \sim N(\tilde{\psi}_{j}\tilde{h}_{t-1,j}, \tilde{\sigma}_{j}^{2}), j = 1, \dots, r.
\end{equation*}
Without this normalization, multiplying the $j$th column of $\Lambda$ by a positive constant could be offset by inversely rescaling the corresponding latent factor, leaving the model for the observations unchanged. Centering the factor log-volatility at zero removes this scale ambiguity, so the factor scale is carried by $\Lambda_{\cdot j}$.

The covariance matrix $\tilde{\Sigma}_t$ is diagonal. When the factor-volatility paths are distinct, general rotations of the factors do not preserve this diagonal structure and therefore do not preserve the model specification.
 Therefore, such rotations are excluded. Even after imposing these restrictions, the data cannot determine the signs or ordering of the latent factors. Simultaneously changing the sign of a factor and its loading column, or jointly reordering the factors, loading columns, and factor volatility processes, leaves the model unchanged. Our implementation leaves the factor loading matrix unrestricted. This means no triangular zero restriction is imposed on $\Lambda$. Factor signs are normalized ex post using the default maximin sign-identification procedure, but factor ordering is not assigned a structural interpretation. These sign and ordering ambiguities do not affect the inferential targets of the paper because $\text{cov}(\boldsymbol{\epsilon}_{t}) = \Lambda\tilde{\Sigma}_{t}\Lambda^\top + \bar{\Sigma}_{t}$ is unchanged by such transformations. The resulting covariance matrix, correlation matrix, and proposed determinant score are therefore also invariant to factor signs and ordering.

The dynamic shrinkage and stochastic volatility components act on different parts of the model. The DSP controls the innovation variances of the time-varying regression parameters, $\alpha_{a,t+1}-\alpha_{a,t}$ and $\beta_{a,t+1}-\beta_{a,t}$, and thereby regularizes the evolution of the regression mean $\alpha_{a,t}+\beta_{a,t}r_{M,t}$. Conditional on the regression parameters and volatility states, and after integrating out the latent factors, the factor and idiosyncratic stochastic volatility processes govern the residual covariance $\Lambda\tilde{\Sigma}_t\Lambda^\top+\bar{\Sigma}_t$.

The components also have different temporal and cross-sectional implications. A regression-parameter innovation changes the fitted mean and carries forward through the parameter's random-walk evolution. A change in $\beta_{a,t}$ affects the fitted return in proportion to the observed excess market return. Observation volatility instead governs residual dispersion. The common residual covariance has the low-rank structure $\Lambda\tilde{\Sigma}_t\Lambda^\top$, with time-invariant loadings, while $\bar{\Sigma}_t$ is diagonal. These features provide information across time and assets for distinguishing regression-parameter movements from changes in residual volatility.

An isolated large observation may be consistent with either a regression-parameter movement or an increase in observation volatility, leaving some finite-sample uncertainty about their allocation. The priors regularize this allocation, while the model draws information from the observed market-return path, the different temporal behavior of the regression parameters and volatilities, and the cross-sectional structure of residual dependence. These complementary sources of information help distinguish the components through their implications across time and assets.

\section{Five-factor Fama--French extension}

We extend our model to the Fama--French five-factor asset pricing model by including the following factors: the excess market return ($r_{M,t}$), small-minus-big (SMB), high-minus-low (HML), robust-minus-weak (RMW), and conservative-minus-aggressive (CMA) . In particular, 
\begin{align*}
    \boldsymbol{x}_{t} &= (r_{M,t}, SMB_{t}, HML_{t}, RMW_{t}, CMA_{t})^\top\\
    \boldsymbol{r}_{t} &= \boldsymbol{\alpha}_{t} + B_{t}\boldsymbol{x}_{t}+ \boldsymbol{\epsilon}_{t}, B_{t} \in \mathbb{R}^{N \times 5}, 
\end{align*}
where $B_{t}$ contains the time-varying factor loadings for the $N$ assets. 

Unlike in the main paper, the model now contains five factors instead of only one. Therefore, we model the five-factor covariance instead of a one-factor variance. In particular, to model the covariance among the factors, we fit a three-latent-factor MFSV model which yields the time-varying factor covariance matrix $H_{t}$. Then, under the same zero-mean and independence assumptions as the main paper, the model-based covariance matrix formula becomes 
\begin{equation*}
\text{cov}(\boldsymbol{r}_{t}) = \text{var}(\boldsymbol{\alpha}_{t}) + \mathbb{E}[B_{t}H_{t}B_{t}^\top] + \text{cov}(\boldsymbol{\epsilon}_{t}). 
\end{equation*}
The R code, along with a working example, is available at the following link: \url{https://github.com/DanielCoulson/Modeling-Dynamic-Correlation-Matrices-with-Shrinkage-Priors}.

\begin{figure}[H]
    \centering
    \includegraphics[width = 0.65\linewidth]{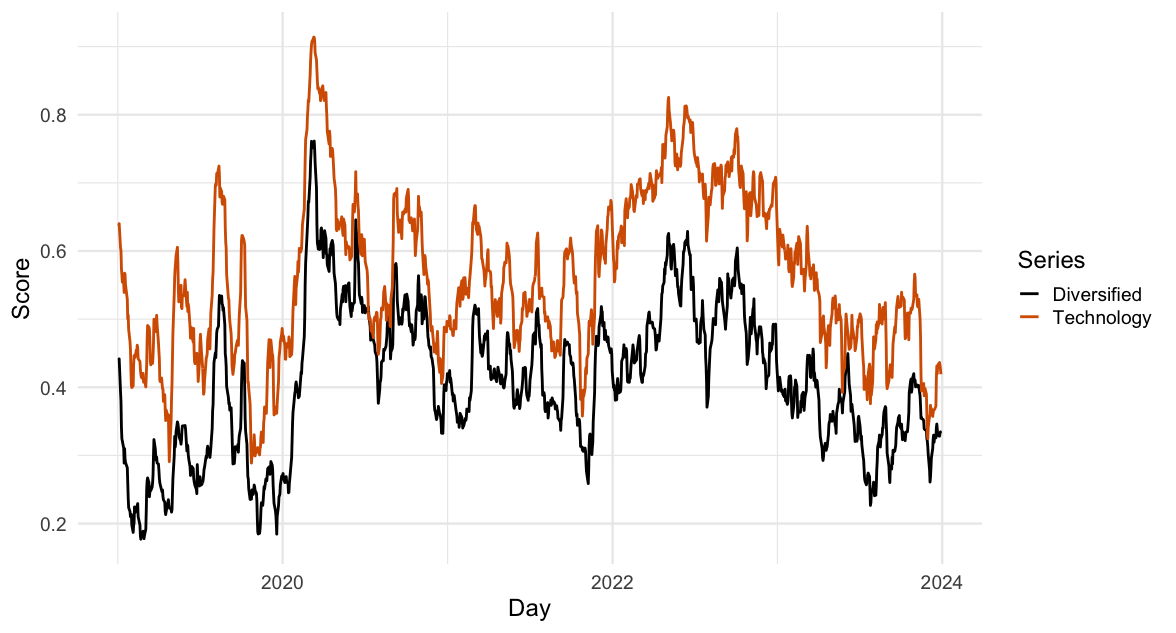}
    \caption{Posterior mean score time series for the technology portfolio (vermilion) and the diversified portfolio (black).}
    \label{fig:2020_5_factor}
\end{figure}

\begin{table}[t]
\centering
\caption{Comparison of the CAPM and Fama--French five-factor specifications}
\label{tab:capm_ff5}
\begin{tabular}{lcccc}
\hline
& \multicolumn{2}{c}{Diversified Portfolio}
& \multicolumn{2}{c}{Technology Portfolio} \\
Metric & CAPM & FF5 & CAPM & FF5 \\
\hline
CAPM--FF5 score correlation
& \multicolumn{2}{c}{0.946}
& \multicolumn{2}{c}{0.960} \\

Mean absolute score difference
& \multicolumn{2}{c}{0.031}
& \multicolumn{2}{c}{0.032} \\

Peak score date
& 2020-03-11 & 2020-03-11
& 2020-03-10 & 2020-03-10 \\

Peak score
& 0.767 & 0.761
& 0.915 & 0.914 \\

Contemporaneous VIX correlation
& 0.764 & 0.759
& 0.696 & 0.737 \\

Strongest VIX correlation
& 0.772 & 0.770
& 0.703 & 0.747 \\

Lag of strongest VIX correlation
& 3 & 3
& 3 & 3 \\
\hline
\end{tabular}
\end{table}

The five-factor extension produces results that are consistent with those obtained under the CAPM specification. The sample correlations between the CAPM and Fama--French five-factor posterior mean score series are 0.946 for the diversified portfolio and 0.960 for the technology portfolio, with mean absolute differences of 0.031 and 0.032, respectively. Moreover, the peak score occurs on the same date under both specifications: March 10, 2020, for the technology portfolio and March 11, 2020, for the diversified portfolio. 

The relationship with the VIX is also consistent across the two specifications. For the diversified portfolio, the contemporaneous correlation between the posterior mean score and the VIX is 0.764 under the CAPM specification and 0.759 under the five-factor specification; for the technology portfolio, the corresponding correlations are 0.696 and 0.737. In all four cases, the strongest posterior mean score--VIX correlation occurs when the score leads the VIX by three trading days. The corresponding correlations are 0.772 and 0.770 for the diversified portfolio under the CAPM and five-factor specifications, respectively, and 0.703 and 0.747 for the technology portfolio.

We also see that the dependence dynamics shown in Figure S1 are consistent with those shown in Figure 2a from the main paper. 

\section{Extended Subprime mortgage crisis analysis}
This section extends the subprime mortgage crisis analysis of
Section~5.2.1 of the main paper. Section~D.1 presents the estimated
determinant score time series for each method. Section~D.2 introduces
four alternative scalar dependence measures and assesses the robustness
of the conclusions to the choice of summary. Section~D.3 quantifies the
relationship between the estimated scores and the VIX. Section~D.4
displays the estimated time-varying shrinkage, which adapts between the
stable and turbulent periods of the sample.
\subsection{Point estimate score plots for each method}
Below, we present the estimated determinant score time series for each of the methods under consideration. 
\begin{figure}[H]
    \centering
    \includegraphics[width = 0.65\linewidth]{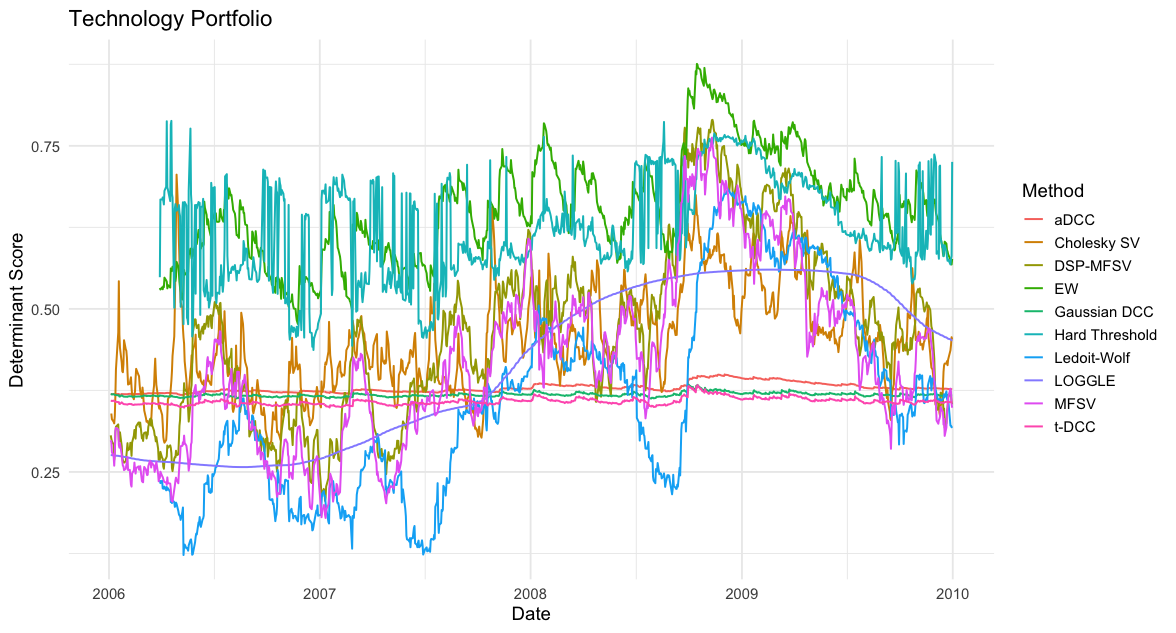}
    \caption{Time series of point estimates of the determinant score for all considered methods}
    \label{fig:2008_tech_extended}
\end{figure}
\begin{figure}[H]
    \centering
    \includegraphics[width = 0.65\linewidth]{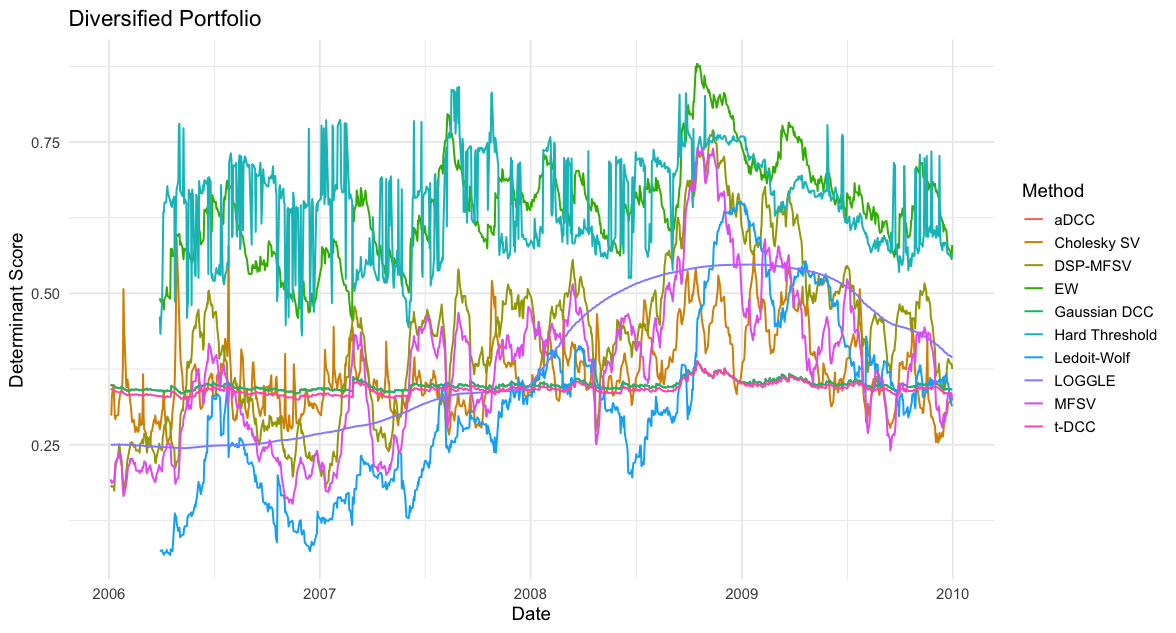}
    \caption{Time series of point estimates of the determinant score for all considered methods}
    \label{fig:2008_diversify_extended}
\end{figure}

The estimated determinant scores differ considerably across methods in both their baseline level and degree of temporal variation. The DCC specifications are relatively stable over the sample, with only a modest increase during 2008, whereas EW and Hard Threshold estimate persistently greater dependence and therefore a less pronounced contrast between the crisis and the preceding periods. Ledoit--Wolf, Cholesky SV, MFSV, and DSP-MFSV-CAPM exhibit substantially greater variation over time and identify a much stronger increase in dependence during 2007--2008. Loggle's estimated time series is noticeably smooth, with its crisis-period increase spread over a wider interval. That is, the score begins rising earlier than several other methods and remains elevated longer as other models' estimates decrease. Thus, although the methods broadly associate late 2008 with stronger dependence, they differ materially in the magnitude, timing, and persistence of the estimated change relative to the pre-crisis period. 

\subsection{Time series of different scores}
We consider our proposed determinant score denoted by $S_{t}^{D}(\cdot)$ and four additional scalar multivariate dependence measures. Let $N$ be the number of assets, $\lambda_{1,t},\dots, \lambda_{N,t}$ denote the eigenvalues of the correlation matrix $R_{t}$, and define 
\begin{equation*}
    q_{i,t} := \frac{\max(\lambda_{i,t},0)}{\sum\limits_{j=1}^{N}\max(\lambda_{j,t},0)}, i = 1, \dots, N. 
\end{equation*}
Then the average pairwise correlation is given by 
\begin{equation*}
    S_{t}^{\text{AP}} = \frac{2}{N(N-1)} \sum\limits_{i<j} (R_{t})_{i,j}. 
\end{equation*}
The effective rank is given by 
\begin{equation*}
    S_{t}^{\text{ER}} = \exp\{ - \sum\limits_{i:q_{i,t}>0} q_{i,t}\log(q_{i,t})\}. 
\end{equation*}
The eigenvalue concentration is given by 
\begin{equation*}
    S_{t}^{\text{EC}} = \max\limits_{i}q_{i,t}. 
\end{equation*}
The participation ratio is given by 
\begin{equation*}
    S_{t}^{\text{PR}} = \frac{1}{\sum\limits_{i=1}^{N}q_{i,t}^{2}}  . 
\end{equation*}

In the following plots, all scores are scaled to lie on $[0,1]$ and oriented so that larger values correspond to stronger dependence.

\begin{figure}[H]
    \centering
    \includegraphics[width = 0.65\linewidth]{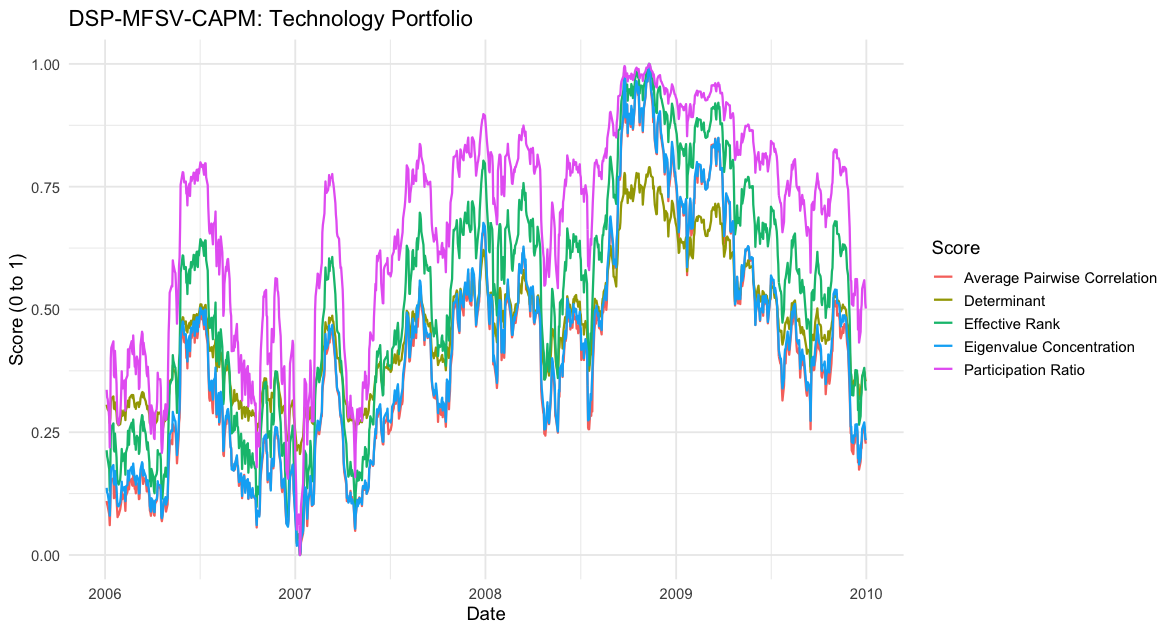}
    \caption{Posterior mean score estimates for the technology portfolio}
    \label{fig:2008_tech_different_scores}
\end{figure}

\begin{figure}[H]
    \centering
    \includegraphics[width = 0.65\linewidth]{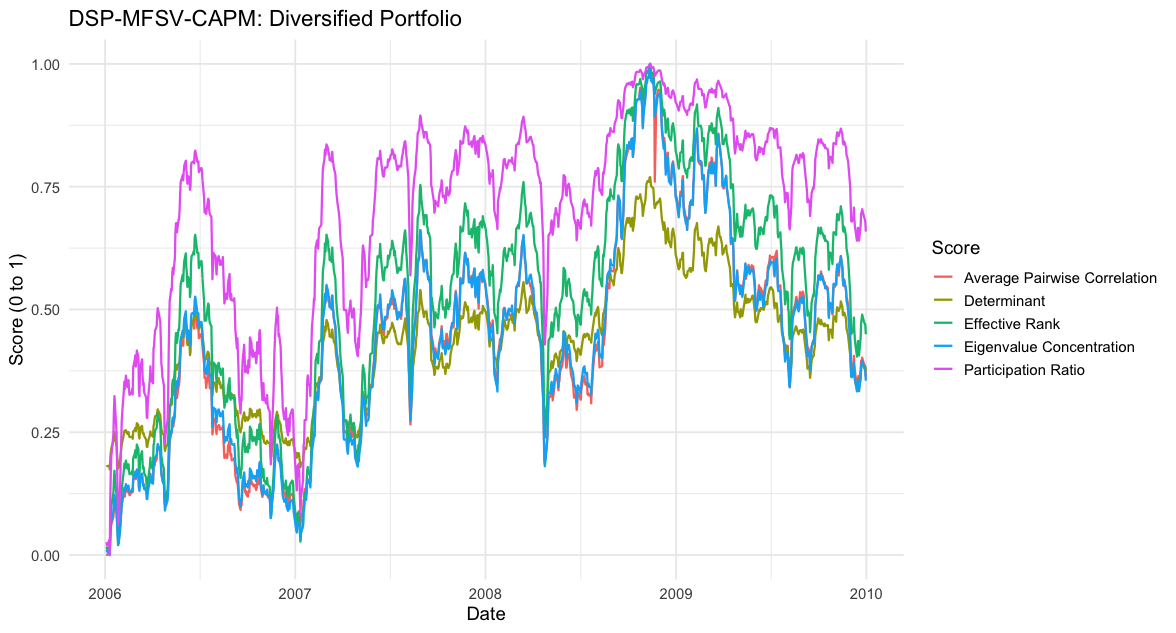}
    \caption{Posterior mean score estimates for the diversified portfolio}
    \label{fig:2008_diversify_different_scores}
\end{figure}

The two plots show a strong common dependence signal across all measures but also systematic differences due to the different notions of dependence they capture. For both portfolios, all scores rise noticeably during the period 2007--2009 and peak in late 2008. This demonstrates that the proposed determinant score identifies the same major periods of elevated dependence as established alternatives. The determinant score tracks the average pairwise correlation and eigenvalue concentration particularly closely, indicating that it preserves the broad chronology of the crisis evident under these alternative summaries. 

At the same time, the series are not identical. The determinant score frequently separates from average pairwise correlation outside the most severe crisis period. The differences are particularly visible in 2007--2008. Effective rank and participation ratio show larger deviations, reflecting their different weighting of the eigenvalue spectrum. Thus, the figures suggest that the measures broadly agree on when dependence rises and falls, but can differ materially in the relative severity they assign to particular crisis periods. The diversified portfolio is especially informative because its more heterogeneous cross-sectional structure creates greater scope for dependence to be concentrated within subsets of assets rather than uniformly across all pairs. We focus the subsequent analysis on the comparison between the determinant score and average pairwise correlation, as this comparison provides the most transparent interpretation of the sources of disagreement. 

In the technology portfolio, the determinant score is almost indistinguishable from average pairwise correlation in rank and timing, with Pearson and Spearman correlations of 0.997 and 0.996, respectively. In the diversified portfolio, the relationship remains strong but weaker, with corresponding correlations of 0.980 and 0.967.  This suggests that the determinant score preserves the broad dependence signal while exhibiting greater departures from average pairwise correlation in the diversified portfolio. 

We then examine the dates with the largest departures as measured by the largest percentile ranking disagreement between the determinant score and the average pairwise correlation. This disagreement is particularly noticeable in the diversified portfolio, where the maximum rank disagreement is approximately 0.31, compared with 0.11 in the technology portfolio. On the five dates with the largest ranking disagreement, absolute correlations are substantially stronger within clusters than between clusters. In particular, within-cluster absolute correlations are around 0.41--0.43, whereas between-cluster correlations are around 0.18--0.20. The clusters are obtained using average-linkage hierarchical clustering with pairwise dissimilarity $d_{ij} = 1- |R_{ij}|$, resulting in four clusters. This shows that the determinant score responds to structured, block-like dependence that is partially obscured by averaging across pairwise correlations. This is illustrated clearly by the heat map below for 2008-08-01, which is a date with large disagreement between the two measures of overall dependence. 

\begin{figure}[H]
    \centering
    \includegraphics[width = 0.65\linewidth]{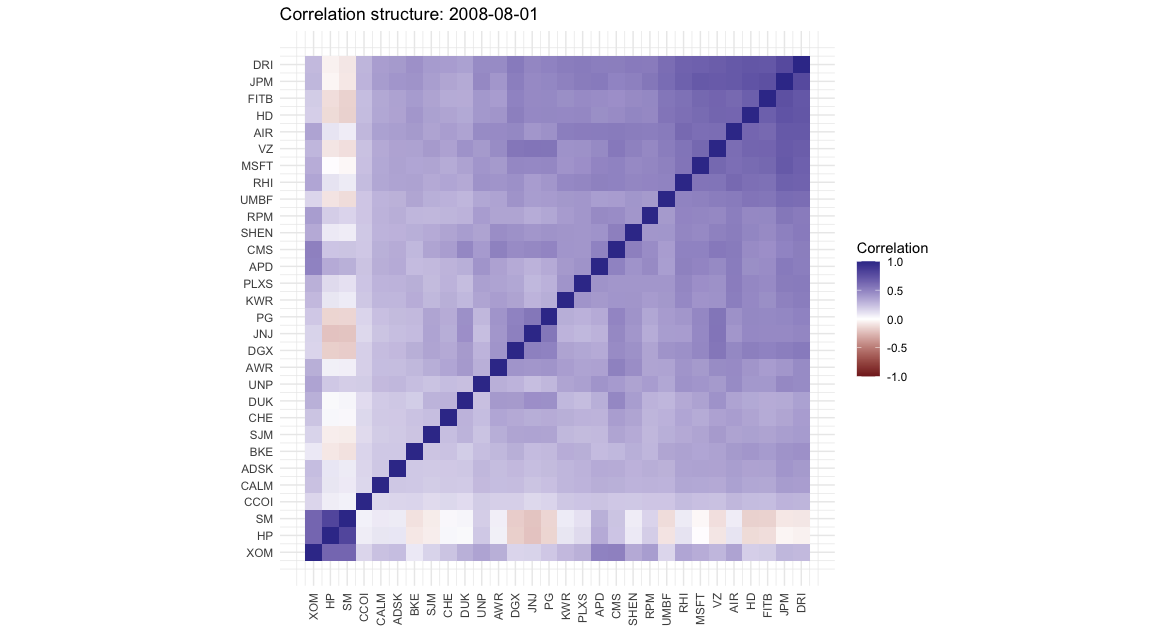}
    \caption{Correlation matrix heat map}
    \label{fig:2008_correlation_heat_map}
\end{figure}

The heat map shows that dependence on 2008-08-01 is heterogeneous. The matrix contains localized regions of stronger positive dependence, most notably the block containing XOM, HP, and SM, whose within-block correlations are substantially stronger than their correlations with much of the remaining portfolio. Many of the correlations linking this block to other assets are considerably weaker, with some close to zero or negative. Average pairwise correlation can attenuate the contribution of such a block because the strongly correlated pairs are averaged together with many weakly correlated pairs. By contrast, the determinant score can register the loss of effective multivariate diversification. In this sense, sensitivity to heterogeneous dependence is useful if the target is overall portfolio dependence or diversification loss.

This difference in handling heterogeneity is illustrated mathematically with the following example. Consider a portfolio with $N=4$ assets and $0<\rho<1/3$, with equicorrelation matrix 
\begin{equation*}
    R_{\text{eq}} = \begin{pmatrix}
    1 & \rho & \rho & \rho\\
    \rho & 1&\rho&\rho\\
    \rho & \rho & 1 & \rho\\
    \rho & \rho & \rho & 1
\end{pmatrix}.
\end{equation*}
Its average pairwise correlation is 
\begin{equation*}
    S^{\text{AP}} (R_{\text{eq}}) = \frac{1}{6}(6\rho) = \rho.
\end{equation*}
The eigenvalues of $R_{\text{eq}}$ are $1+3\rho$ once and $1-\rho$ three times. Therefore, 
\begin{equation*}
    \det(R_{\text{eq}}) = (1-\rho)^{3}(1+3\rho).
\end{equation*}
Now consider the block correlation matrix 
\begin{equation*}
    R_{\text{block}} =\begin{pmatrix}
    1 & 3\rho & 0 & 0 \\
    3\rho & 1 & 0 & 0 \\
    0 & 0 & 1 & 3\rho\\
    0&0&3\rho &1\\
\end{pmatrix}.
\end{equation*}
Its average pairwise correlation is
\begin{equation*}
    S^{\text{AP}}(R_{\text{block}}) = \frac{1}{6}(3\rho + 3\rho) = \rho. 
\end{equation*}
Thus
\begin{equation*}
    S^{\text{AP}}(R_{\text{eq}}) = S^{\text{AP}}(R_{\text{block}}).
\end{equation*}
However,
\begin{equation*}
    \det(R_{\text{block}}) = \det\begin{pmatrix}
        1 & 3\rho\\
        3\rho &1
    \end{pmatrix}^{2} = (1-9\rho^{2})^{2}.
\end{equation*}
Comparing determinants gives 
\begin{align*}
    \det(R_{\text{eq}})- \det(R_{\text{block}}) &= (1-\rho)^{3}(1+3\rho) - (1-9\rho^{2})^{2}\\
    &=1-6\rho^{2} +8\rho^{3} - 3\rho^{4} -1 + 18\rho^{2}+81 \rho^{4}\\
    &= 12\rho^{2} + 8\rho^{3} -84\rho^{4}\\
    &= 4\rho^{2}(3 + 2\rho - 21\rho^{2}).
\end{align*}
For $0< \rho < \frac{1}{3}$, 
\begin{equation*}
    4\rho^{2}>0 \text{ and } 3+2\rho-21\rho^{2}>0. 
\end{equation*}
Therefore, $\det(R_{\text{block}})< \det(R_{\text{eq}})$. Thus, 
\begin{equation*}
    S^{D}(R_{\text{block}}) > S^{D}(R_{\text{eq}}) \text{ while }S^{\text{AP}}(R_{\text{block}}) = S^{\text{AP}}(R_{\text{eq}}), 
\end{equation*}
since $S^{D}(R)$ is decreasing in $\det(R)$. 
\subsection{Correlation with the VIX }

The contemporaneous correlation between the VIX and the posterior mean dependence scores is 0.864 for the diversified portfolio and 0.876 for the technology portfolio. The strongest posterior mean score--VIX correlation occurs 5 trading days ahead for the diversified portfolio and 8 trading days ahead for the technology portfolio, with correlations of 0.870 and 0.887, respectively.

\subsection{Time-varying shrinkage plots}
Below, we present plots of the posterior mean of the local shrinkage parameter associated with the beta coefficient for each asset, where the spacing on the y-axis is logarithmic to accommodate relatively large values. 

\begin{figure}[H]
    \centering
    \includegraphics[width = 0.65\linewidth]{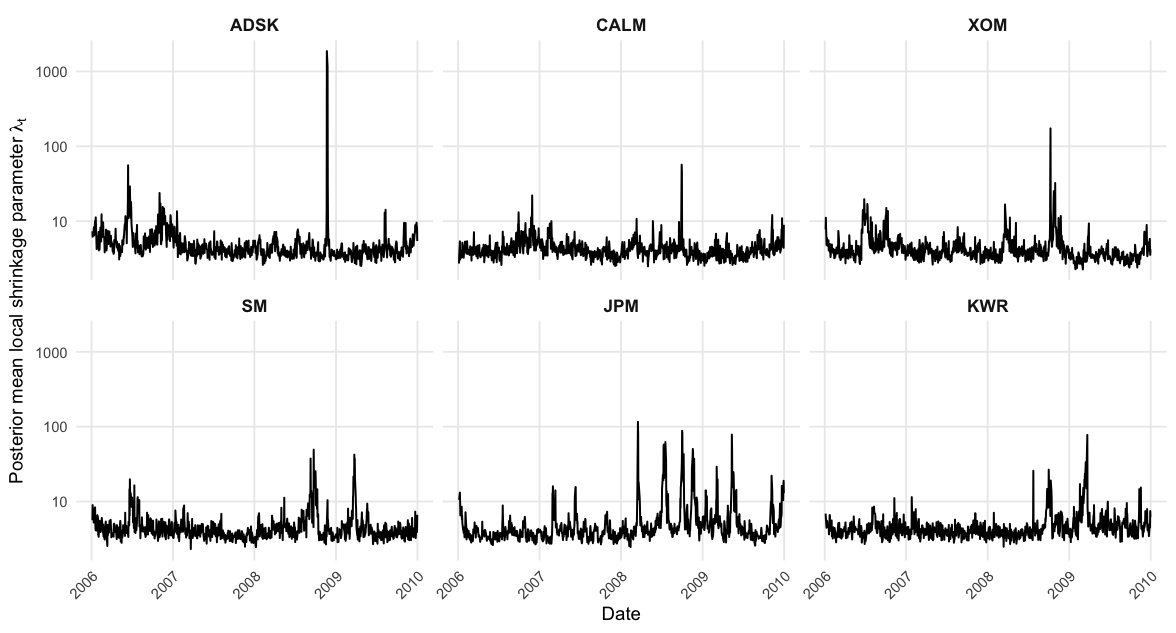}
    \caption{Posterior mean of the local shrinkage parameter for six assets}
    \label{fig:2008_shrinkage_plots}
\end{figure}

Across the assets, the posterior mean of $\lambda_{\beta_{a},t}$ remains comparatively low and stable for much of the sample, but with short-lived spikes. The clustering of several of the spikes around 2008--2009 is consistent with the model detecting periods in which the corresponding beta coefficients require much greater flexibility during the financial crisis. The heterogeneity in both the timing and magnitude of these spikes highlights one of the motivations for the dynamic shrinkage specification. Rather than imposing a fixed common timing of increased coefficient innovation across assets, the dynamic shrinkage prior permits asset-specific and highly localized departures from otherwise smooth coefficient evolution. 

\section{Extended COVID-19 analysis}
This section extends the COVID-19 pandemic analysis of Section~5.2.2 of
the main paper. Section~E.1 presents the estimated determinant score
time series for each method. Section~E.2 assesses the robustness of the
conclusions using the alternative dependence measures of Section~D.2.
Section~E.3 quantifies the relationship between the estimated scores
and the VIX. Section~E.4 displays the estimated time-varying shrinkage,
which adapts between the stable and turbulent periods of the sample.
\subsection{Point estimate score plots for each method}
Below, we present the estimated determinant score time series for each of the methods under consideration. 
\begin{figure}[H]
    \centering
    \includegraphics[width = 0.65\linewidth]{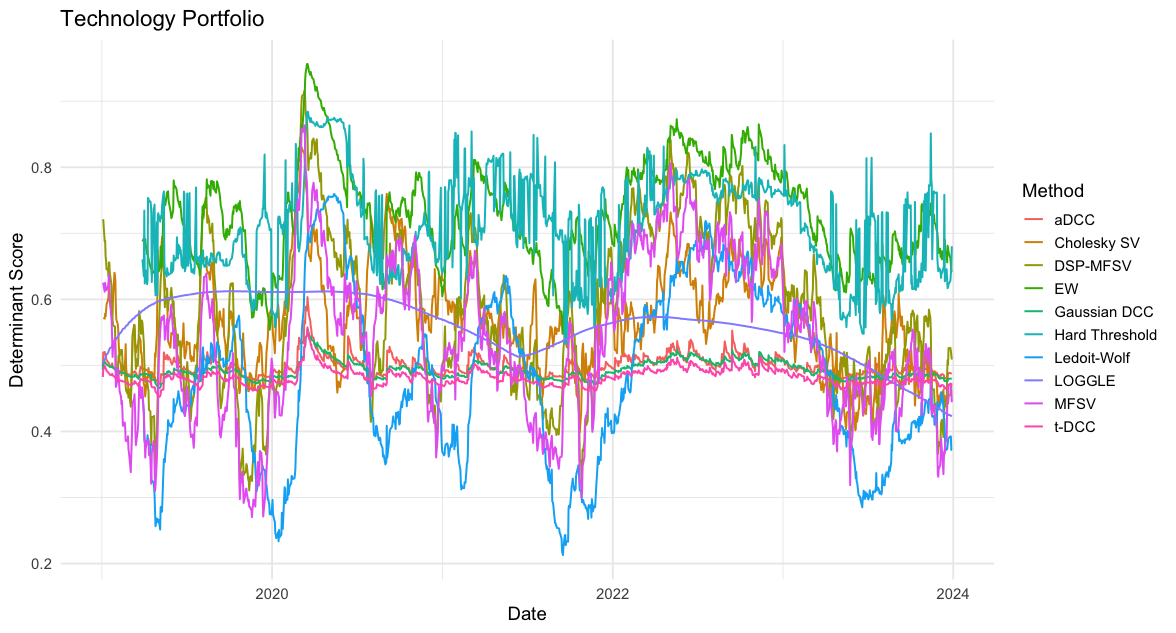}
    \caption{Time series of point estimates of the determinant score for all considered methods}
    \label{fig:2020_tech_extended}
\end{figure}
\begin{figure}[H]
    \centering
    \includegraphics[width = 0.65\linewidth]{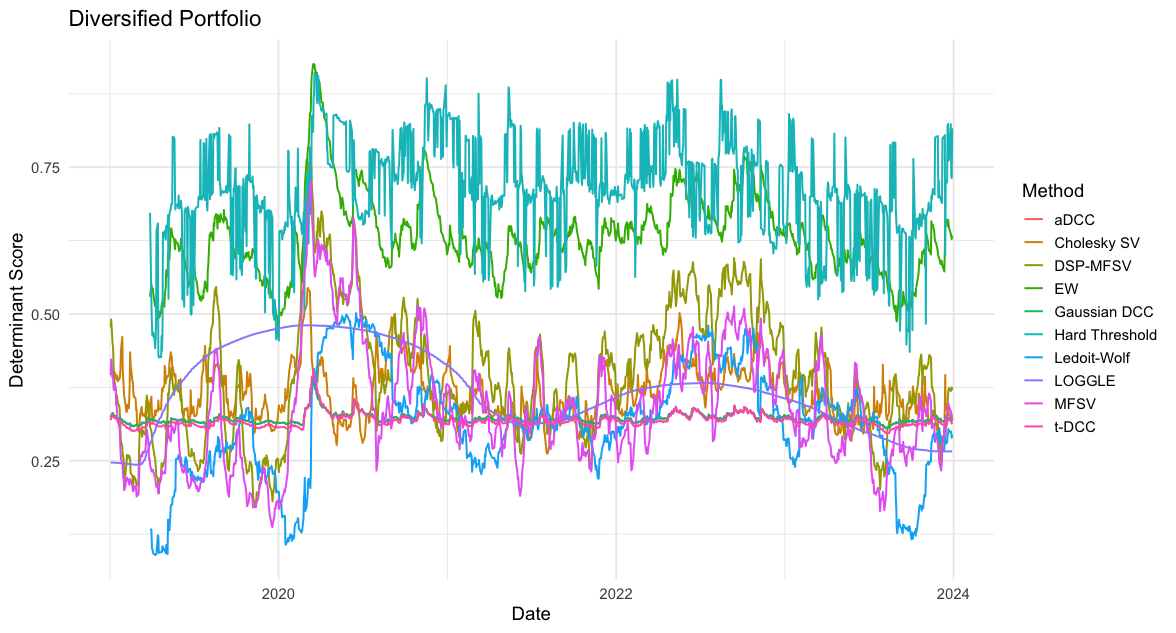}
    \caption{Time series of point estimates of the determinant score for all considered methods}
    \label{fig:2020_diversify_extended}
\end{figure}

The estimated determinant scores again differ considerably across methods in both their baseline levels and the degree of temporal variation. The DCC specifications are relatively stable throughout the sample and identify only a modest increase in dependence around the onset of the COVID-19 pandemic. EW and Hard Threshold estimate persistently greater dependence in both portfolios, and although their scores also increase in early 2020, the contrast with the preceding period is less pronounced given their generally elevated levels. In contrast, Ledoit--Wolf, MFSV, and DSP-MFSV-CAPM exhibit substantially greater temporal variation and identify a much sharper increase in dependence during the initial COVID-19 shock, while Cholesky SV displays a less pronounced increase. Loggle is again considerably smoother than many of the estimators. Its estimated change around 2020 is spread over a wider interval rather than being concentrated in a sharp spike, especially in the diversified portfolio. That is, the increase begins before the most abrupt movements in the other methods and dissipates more gradually. Thus, although the methods broadly associate early 2020 with stronger dependence, they differ materially in how exceptional the COVID-19 periods appears relative to their estimated pre-pandemic levels and in the magnitude and persistence of the subsequent adjustments. 

\subsection{Time series of different scores}
In the following plots, all scores are scaled to lie on $[0,1]$ and oriented so that larger values correspond to stronger dependence.

\begin{figure}[H]
    \centering
    \includegraphics[width = 0.65\linewidth]{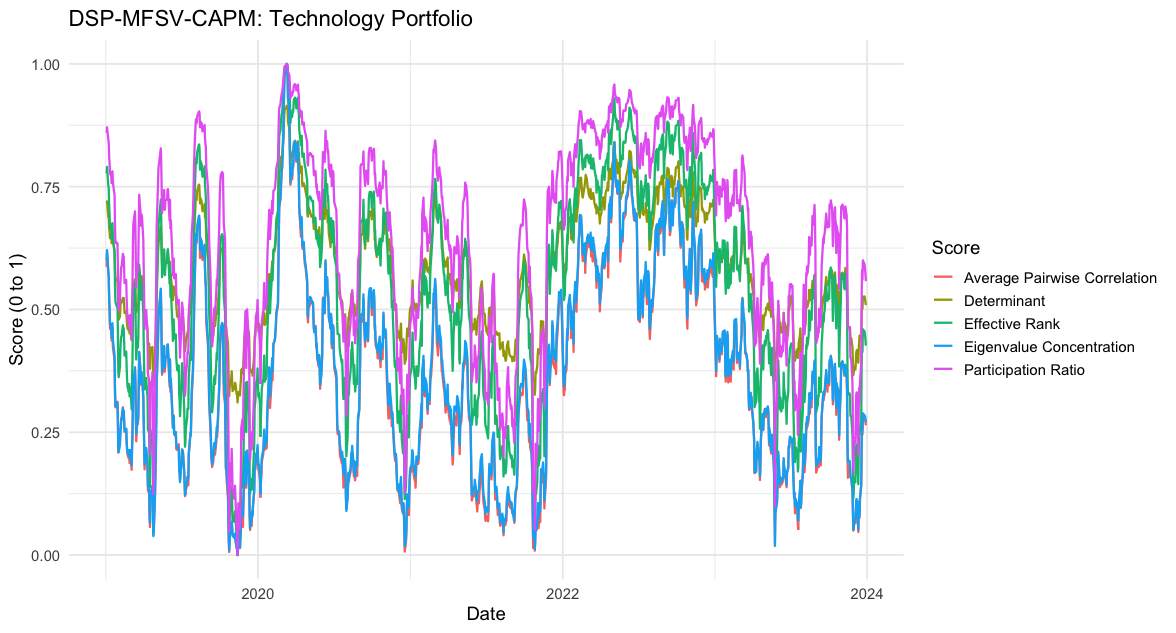}
    \caption{Posterior mean score estimates for the technology portfolio}
    \label{fig:2020_tech_different_scores}
\end{figure}

\begin{figure}[H]
    \centering
    \includegraphics[width = 0.65\linewidth]{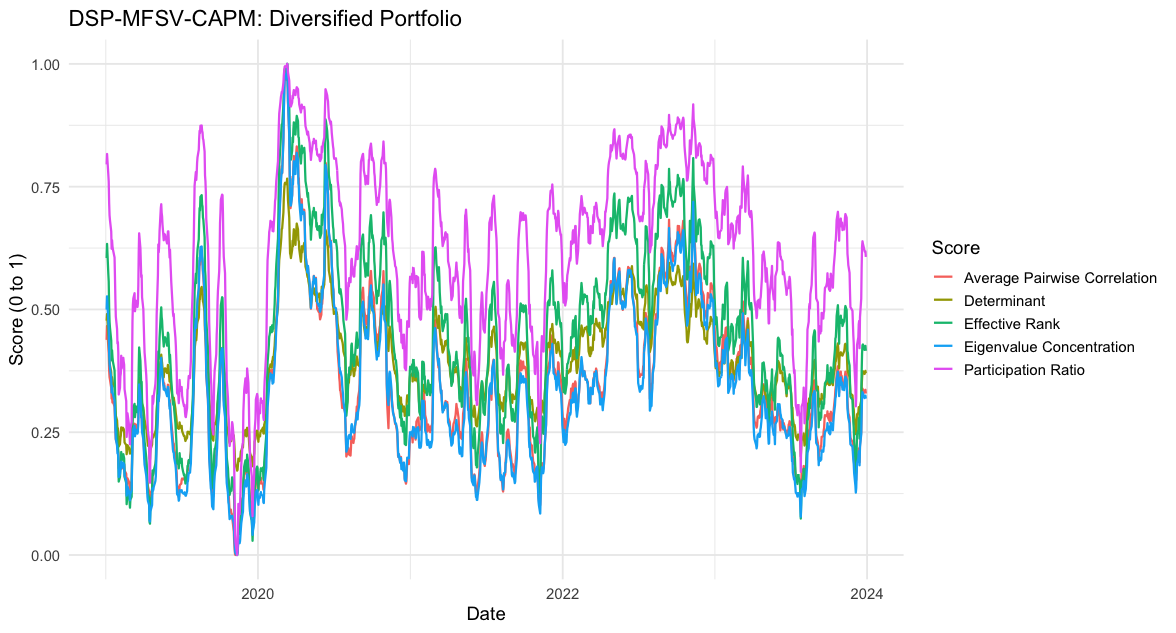}
    \caption{Posterior mean score estimates for the diversified portfolio}
    \label{fig:2020_diversify_different_scores}
\end{figure}

The two plots show a strong common dependence signal across all measures but also systematic differences due to the different notions of dependence they capture. For both portfolios, all scores rise noticeably during several periods, especially in early 2020. This demonstrates that the proposed determinant score identifies the same major periods of elevated dependence as established alternatives. The determinant score closely tracks average pairwise correlation and the scores produced by the other spectral methods, indicating that it preserves the broad chronology of market dependence rather than producing a qualitatively different account of the crisis period. 

At the same time, the series are not identical. The determinant score diverges from average pairwise correlation at a number of points outside the most pronounced periods of dependence. These differences are visible in both portfolios and arise when dependence is distributed unevenly across the correlation matrix rather than being uniformly strong across all asset pairs. As in the 2008 analysis, we focus the subsequent comparison on the determinant score and average pairwise correlation, since this comparison provides the most transparent interpretation of the sources of disagreement. 

In the technology portfolio, the determinant score remains strongly aligned with average pairwise correlation, with both Pearson and Spearman correlations equal to 0.979. In the diversified portfolio, the relationship is weaker, although still strong, with Pearson and Spearman correlations of 0.957 and 0.951, respectively. 

We then examine the dates with the largest departures, as measured by disagreement in percentile rankings between the determinant score and average pairwise correlation. Unlike in the 2008 example, substantial disagreement occurs in both portfolios: the maximum rank disagreement is approximately 0.24 in the technology portfolio and 0.27 in the diversified portfolio. On these dates, absolute correlations are considerably stronger within clusters than between clusters. For example, on 2021-06-17, the technology portfolio has an average within-cluster absolute correlation of approximately 0.41, compared with only 0.17 between clusters. In the diversified portfolio, the contrast is even sharper on the most divergent dates, with within-cluster absolute correlations of approximately 0.32--0.40 compared with only 0.10--0.11 between clusters. This again shows that the determinant score responds to structured, block-like dependence that is partially obscured by averaging across pairwise correlations.

The eigenvalue spectra provide further evidence that these disagreements arise from structured rather than uniform dependence. On the most divergent technology dates, the largest eigenvalue accounts for approximately 37--46$\%$ of the total eigenvalue mass, whereas in the diversified portfolio it accounts for only around 27--31$\%$. Moreover, approximately 16--17 eigenvalues are required to explain $80\%$ of the eigenvalue mass in the diversified portfolio, compared with 11--15 in the technology portfolio. Thus, particularly in the diversified portfolio, dependence is distributed across several cross-sectional modes rather than being dominated by a single common component. Averaging across pairwise correlations partially obscures this richer structure, whereas the determinant score responds to changes throughout the eigenvalue spectrum. 

This difference in handling the eigenvalue spectrum is illustrated mathematically with the following example. 

Consider a portfolio with $N=4$ assets, and let 
\begin{equation*}
    0<\rho <1 ,\text{ and } 0 < \delta < \min\{ 2 \rho , 1-\rho \}. 
\end{equation*}
Define the correlation matrix 
\begin{equation*}
    R(\delta) = \begin{pmatrix}
        1 & \rho + \delta & \rho - \frac{\delta}{2} & \rho - \frac{\delta}{2}\\
        \rho + \delta & 1& \rho - \frac{\delta}{2} & \rho - \frac{\delta}{2}\\
        \rho - \frac{\delta}{2} & \rho - \frac{\delta}{2} & 1 & \rho + \delta\\
        \rho - \frac{\delta}{2} & \rho - \frac{\delta}{2} & \rho + \delta & 1
    \end{pmatrix}.
\end{equation*}
The first two assets form one cluster and the final two assets form a second cluster. The within-cluster correlation is $\rho + \delta$, whereas the between-cluster correlation is $\rho - \frac{\delta}{2}$. Thus, increasing $\delta$ strengthens the cross-sectional cluster structure. Then, for any admissible value of $\delta$, 
\begin{align*}
    S^{\text{AP}}(R(\delta)) &= \frac{1}{6}[2(\rho+\delta) + 4(\rho - \frac{\delta}{2})]\\
    &= \frac{1}{6}[2\rho + 2\delta + 4\rho - 2\delta]\\
    &= \rho. 
\end{align*}
Thus, average pairwise correlation is unchanged as dependence becomes increasingly concentrated within the two clusters. When $\delta = 0$, the matrix reduces to the equicorrelation matrix 
\begin{equation*}
    R(0) = \begin{pmatrix}
        1 & \rho & \rho & \rho\\
        \rho & 1 & \rho & \rho\\
        \rho & \rho & 1 & \rho \\
        \rho & \rho & \rho & 1
    \end{pmatrix}. 
\end{equation*}
Thus, $S^{\text{AP}}(R(\delta)) = S^{\text{AP}}(R(0)) = \rho$. 
Now consider the eigenvalues of $R(\delta)$. First, consider 
\begin{equation*}
    \boldsymbol{v}_{1}^\top = \begin{pmatrix}
        1 & 1& 1 & 1
    \end{pmatrix}. 
\end{equation*}
Note that each row of $R(\delta)$ sums to 
\begin{equation*}
    1 + (\rho + \delta) + 2(\rho - \frac{\delta}{2}) = 1+ 3\rho. 
\end{equation*}
Therefore, 
\begin{equation*}
    R(\delta)\boldsymbol{v}_{1} = (1+3 \rho)\boldsymbol{v}_{1}. 
\end{equation*}
Hence, $\lambda_{1} = 1 + 3\rho$ is an eigenvalue and does not depend on $\delta$. Now consider the vector 
\begin{equation*}
    \boldsymbol{v}_{2}^\top = \begin{pmatrix}
        1 & 1 & -1 & -1
    \end{pmatrix}
\end{equation*}
which represents a contrast between the two clusters. For the first and second components, 
\begin{align*}
    [R(\delta)\boldsymbol{v}_{2}]_{1}  &= 1+(\rho + \delta) - (\rho - \frac{\delta}{2}) - (\rho - \frac{\delta}{2})\\
    &= 1- \rho +2 \delta\\
    & = [R(\delta)\boldsymbol{v}_{2}]_{2}. 
\end{align*}
For the third and fourth components, 
\begin{align*}
    [R(\delta)\boldsymbol{v}_{2}]_{3} &= (\rho - \frac{\delta}{2}) + (\rho - \frac{\delta}{2}) - 1-(\rho+\delta)\\
    &= - (1-\rho + 2\delta)\\
    &=[R(\delta)\boldsymbol{v}_{2}]_{4}
\end{align*}
Therefore, 
\begin{equation*}
    R(\delta)\boldsymbol{v}_{2} = (1- \rho+2\delta)\boldsymbol{v}_{2}.
\end{equation*}
Hence, $\lambda_{2} = 1-\rho+2\delta$. Moreover, since $\delta<2\rho$, $\lambda_{1}-\lambda_{2}=4\rho-2\delta>0$, so $\lambda_{1}$ is the largest eigenvalue. Next consider 
\begin{equation*}
    \boldsymbol{v}_{3}^\top = \begin{pmatrix}
        1 & -1&0&0
    \end{pmatrix}. 
\end{equation*}
Then,
\begin{equation*}
    [R(\delta)\boldsymbol{v}_{3}]_{1} = 1-\rho - \delta \text{ and } [R(\delta)\boldsymbol{v}_{3}]_{2} =- (1-\rho-\delta). 
\end{equation*}
For the third and fourth components,
\begin{equation*}
    [R(\delta)\boldsymbol{v}_{3}]_{3} = [R(\delta)\boldsymbol{v}_{3}]_{4} = (\rho - \frac{\delta}{2}) - (\rho - \frac{\delta}{2}) = 0. 
\end{equation*}
Hence, 
\begin{equation*}
    R(\delta) \boldsymbol{v}_{3} = (1- \rho - \delta) \boldsymbol{v}_{3}. 
\end{equation*}
Therefore, $\lambda_{3} = 1-\rho-\delta$. Similarly, for 
\begin{equation*}
    \boldsymbol{v}_{4}^\top = \begin{pmatrix}
        0 & 0 & 1 & -1
    \end{pmatrix}
\end{equation*}
we obtain \begin{equation*}
    R(\delta)\boldsymbol{v}_{4} = (1-\rho - \delta)\boldsymbol{v}_{4}.
\end{equation*}
Hence, $\lambda_{4} = 1- \rho-\delta$. Since $0 < \delta < 1-\rho$, we have 
\begin{equation*}
    1- \rho - \delta>0
\end{equation*}
and, consequently, all four eigenvalues are positive. Hence, $R(\delta)$ is positive definite. The determinant is then given by 
\begin{align*}
    \det(R(\delta)) &= \lambda_{1}\lambda_{2}\lambda_{3}\lambda_{4}\\
    & = (1+3\rho)(1-\rho+2\delta) (1-\rho - \delta)^{2}. 
\end{align*}
We can now compare the two determinants: 
\begin{equation*}
    \det(R(0))- \det(R(\delta)) = (1+3\rho)[(1-\rho)^{3} - (1-\rho + 2\delta)(1-\rho - \delta)^{2}]. 
\end{equation*}
Note that 
\begin{equation*}
    (1-\rho-\delta)^{2} = (1-\rho)^{2} -2(1-\rho)\delta + \delta^{2}. 
\end{equation*}
Therefore, 
\begin{align*}
    (1-\rho+2\delta) (1-\rho - \delta)^{2} & = (1- \rho +  2\delta)[(1-\rho)^{2}-2(1-\rho)\delta + \delta^{2}]\\
    &=(1-\rho)^{3} -2(1-\rho)^{2}\delta + (1-\rho)\delta^{2} + 2(1-\rho)^{2}\delta -4(1-\rho)\delta^{2} + 2\delta^{3}\\
    &= (1-\rho)^{3} -3(1-\rho)\delta^{2} +2\delta^{3}.
\end{align*}
It follows that 
\begin{align*}
    \det(R(0)) - \det(R(\delta)) &= (1+3\rho)[3(1-\rho)\delta^{2} - 2\delta^{3}]\\
    &= (1+3\rho) \delta^{2} [3(1-\rho)-2\delta]. 
\end{align*}
Since 
\begin{equation*}
    1+3\rho >0,\delta^{2}>0, \text{ and } \delta < 1- \rho
\end{equation*}
implies that
\begin{align*}
    3(1-\rho)-2\delta &> 3(1-\rho)-2(1-\rho)\\
    & = 1-\rho\\
    &>0,
\end{align*}
we have 
\begin{equation*}
    \det(R(0)) - \det(R(\delta)) >0. 
\end{equation*}
Therefore, 
\begin{equation*}
    \det(R(\delta)) < \det(R(0)).
\end{equation*}
Since $S^{D}(R)$ is decreasing in $\det(R)$, it follows that
\begin{equation}
    S^{D}(R(\delta)) > S^{D}(R(0)),
\end{equation}
while 
\begin{equation*}
    S^{\text{AP}}(R(\delta)) = S^{\text{AP}}(R(0)) = \rho. 
\end{equation*}
Thus, the determinant score increases as dependence becomes more strongly clustered, even though average pairwise correlation is exactly unchanged. Moreover, as the eigenvalues show, the determinant responds to the redistribution of dependence across the full eigenvalue spectrum.

\subsection{Correlation with the VIX}
The contemporaneous correlation between the VIX and the posterior mean dependence scores is 0.764 for the diversified portfolio and 0.696 for the technology portfolio. The strongest lead-lag correlations occur when the posterior mean scores lead the VIX by 3 trading days for both portfolios, with correlations of 0.772 for the diversified portfolio and 0.703 for the technology portfolio.

\subsection{Time-varying shrinkage plots}
Below, we present plots of the posterior mean of the local shrinkage parameter associated with the beta coefficient for each asset, where the spacing on the y-axis is logarithmic to accommodate relatively large values.

\begin{figure}[H]
    \centering
    \includegraphics[width = 0.65\linewidth]{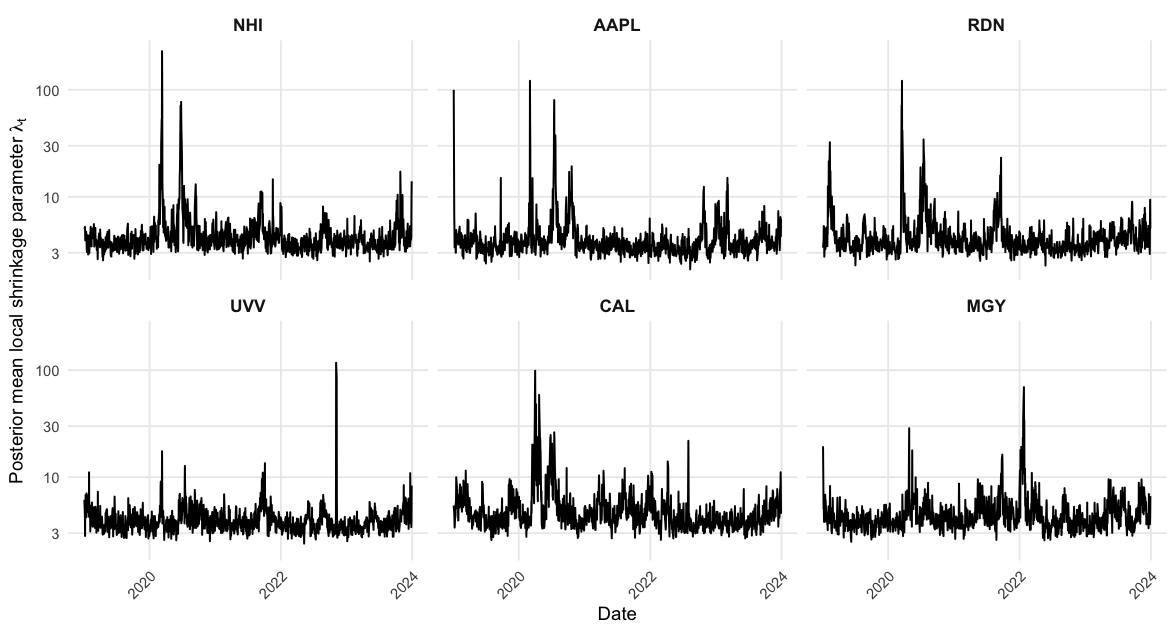}
    \caption{Posterior mean of the local shrinkage parameter for six assets}
    \label{fig:2020_shrinkage_plots}
\end{figure}

Across the assets, the posterior mean of $\lambda_{\beta_{a},t}$ remains comparatively low and stable for much of the sample, with occasional sharp and short-lived spikes. Several of the most pronounced spikes occur around 2020, particularly for NHI, AAPL, RDN, and CAL, consistent with the model allowing substantially greater flexibility in the corresponding beta coefficients during the severe market disruption associated with the COVID-19 crisis. At the same time, the timing and magnitude of the spikes vary considerably across assets. In particular, UVV and MGY exhibit some of their largest increases later in the sample, indicating that periods of unusually rapid coefficient evolution are not confined to the common market-wide disruption in 2020. This heterogeneity again highlights the motivation for the dynamic shrinkage specification. Rather than imposing a common timing or magnitude of coefficient innovation across assets, the dynamic shrinkage prior permits asset-specific and highly localized departures from otherwise smooth coefficient evolution. 

\section{Rolling backtest results}
The out-of-sample exercise is conducted over two separate four-year periods: January 2008–December 2011 and January 2020–December 2023. Each period therefore contains 48 out-of-sample months, giving 96 monthly out-of-sample portfolio decisions in total. The underlying data begin in January 2006 for the first period and January 2018 for the second, providing sufficient observations before the first out-of-sample month. Performance statistics are pooled across the two out-of-sample periods, except for maximum drawdown, which is calculated separately within each period. The backtest uses a rolling 500-trading-day estimation window followed by a one-month out-of-sample holding period. For each calendar month in the out-of-sample period, the estimation sample consists of exactly the 500 trading days ending on the final trading day immediately before that month begins. The subsequent calendar month in its entirety is then treated as genuinely out of sample. At each monthly rebalance, the estimation window is rolled forward, a new 500-day sample is formed, and all models are re-estimated. Thus, this is not a single fixed 500-day training/test split: there is a new 500-day training window and a new one-month out-of-sample test period at each monthly refit. Each refit can be run in parallel, as it uses only its own 500-trading-day history. We use the diversified portfolio with 30 stocks. 

At every monthly refit, each of the 10 covariance models is fitted using only the preceding 500 trading days. The covariance estimate at the end of the 500-day estimation window is then extracted. In the following, $t$ indexes successive out-of-sample holding months, while $d$ indexes individual trading days in the out-of-sample holding periods. 

\subsection{Minimum variance portfolio performance}

For each model and each monthly refit, a fully invested global minimum-variance portfolio (MVP) with unrestricted weights is constructed from that model's estimated covariance matrix. The weights sum to one, but there is no long-only restriction, so short positions and gross exposure greater than one are possible. That is, the MVP associated with holding month $t$ is the portfolio that solves 
\begin{equation*}
    \min\limits_{\boldsymbol{w}_{t}} \boldsymbol{w}^\top_{t}\hat{\Sigma}_{t} \boldsymbol{w}_{t} \text{ subject to } \boldsymbol{1}^\top\boldsymbol{w}_{t} =1, 
\end{equation*}
where $\hat{\Sigma}_{t}$ is the estimated covariance matrix available immediately before holding month $t$ and $\boldsymbol{w}_{t}$ is the vector of portfolio weights used during holding month $t$. The solution is 
\begin{equation*}
    \boldsymbol{w}_{t} = \frac{\hat{\Sigma}_{t}^{-1}\boldsymbol{1}}{\boldsymbol{1}^\top\hat{\Sigma}_{t}^{-1}\boldsymbol{1}}.
\end{equation*}

Crucially, the covariance model is fitted only once at the beginning of each out-of-sample month and the resulting MVP is held fixed throughout that entire month. The same vector of weights is applied to the realized 30-stock return vector on each day of the holding month; there is no daily re-estimation or within-month rebalancing. At the beginning of the next month, the 500-day window is rolled forward, each covariance model is refitted, a new covariance estimate and hence a new MVP is obtained for each model, and the new weights are used for every trading day of that month. This monthly refit-and-hold procedure is repeated for all 48 months of each period.

We report the annualized return, defined as 
\begin{equation*}
    \text{Annualized Return} := (\prod\limits_{d=1}^{D}(1+R_{p,d}))^{\frac{252}{D}}-1,
\end{equation*}
where $R_{p,d}$ is the realized simple return of the portfolio on trading day $d$, $D$ is the number of trading-day observations, and $252$ is the assumed number of trading days per year. 

We define annualized volatility as 
\begin{equation*}
    \text{Annualized Volatility}:= \sqrt{252}\text{sd}(R_{p,d}).
\end{equation*}

The annualized Sharpe ratio is defined as 
\begin{equation*}
    \text{Sharpe}:= \sqrt{252}\frac{\frac{1}{D}\sum\limits_{d=1}^{D}(R_{p,d}-R_{d}^{f})}{\text{sd}(R_{p,d}-R_{d}^{f})},
\end{equation*}
where $R_{d}^{f}$ is the daily risk-free return. 

To define the maximum drawdown, let 
\begin{equation*}
    W_{e,d} = \prod\limits_{s=1}^{d} (1+R_{p,e,s})
\end{equation*}
denote the cumulative wealth within backtest episode $e$, with $W_{e,0}=1$. The maximum drawdown for episode $e$ is 
\begin{equation*}
    \text{MDD}_{e} = \min\limits_{d}\{ \frac{W_{e,d}}{\max\limits_{0\leq s \leq d}W_{e,s}}-1  \}.
\end{equation*}
We calculate the maximum drawdown separately for the 2008--2011 and 2020--2023 periods and report the more severe of the two: 
\begin{equation*}
    \text{MDD} := \min\{ \text{MDD}_\text{2008--2011}, \text{MDD}_\text{2020--2023}\}. 
\end{equation*}

We also report average monthly turnover. Because the portfolio is rebalanced monthly, its weights passively drift during each holding month as the constituent stocks earn different returns. In particular, define 
\begin{equation*}
    \tilde{w}_{a,t-1} := \frac{w_{a,t-1} (1+R_{a,t-1})}{\sum\limits_{j} w_{j,t-1}(1+R_{j,t-1})}, 
\end{equation*}
where $R_{a,t-1}$ is the cumulative realized simple return of stock $a$ over the preceding holding month. When the new covariance estimate produces target weights $w_{a,t}$, turnover is defined as 
\begin{equation*}
    \text{TO}_{t} := \frac{1}{2}\sum\limits_{a}|w_{a,t} - \tilde{w}_{a,t-1}|. 
\end{equation*}
Average monthly turnover is the mean of the available monthly turnover observations. Annualized turnover is the average monthly turnover multiplied by 12. 

We also define the gross exposure of the portfolio as 
\begin{equation*}
    \text{Gross Exposure}_{t} = \sum\limits_{a}|w_{a,t}|. 
\end{equation*}
 
\begin{table}[H]
\centering
\caption{Combined Out-of-Sample Minimum-Variance Portfolio Performance}
\label{tab:combined_minvar}
\resizebox{\textwidth}{!}{
\begin{tabular}{lrrrrrrr}
\hline
Method
& Annualized Return
& Annualized Volatility
& Sharpe
& Max. Drawdown
& Avg. Monthly Turnover
& Annualized Turnover
& Avg. Gross Exposure \\
\hline

MFSV
& 3.68\%
& 19.45\%
& 0.229
& -44.25\%
& 47.48\%
& 569.80\%
& 1.724 \\

Cholesky SV
& 4.39\%
& 18.87\%
& 0.266
& -37.92\%
& 53.55\%
& 642.60\%
& 1.651 \\

\textbf{DSP-MFSV-CAPM}
& \textbf{7.91\%}
& \textbf{18.16\%}
& \textbf{0.451}
& -34.06\%
& 56.68\%
& 680.13\%
& \textbf{1.642} \\

EW
& 3.10\%
& 24.23\%
& 0.203
& -55.21\%
& 151.05\%
& 1812.59\%
& 3.080 \\

Ledoit--Wolf
& 4.85\%
& 19.26\%
& 0.288
& -38.23\%
& \textbf{12.97\%}
& \textbf{155.62\%}
& 1.751 \\

Hard Threshold
& 4.51\%
& 19.55\%
& 0.270
& -39.32\%
& 18.07\%
& 216.82\%
& 1.987 \\

Loggle
& 2.97\%
& 18.71\%
&0.193
&\textbf{-31.27\%}
&37.85\%
&454.17\%
&1.809 \\

Gaussian DCC
& 3.63\%
& 18.24\%
& 0.228
& -35.37\%
& 68.68\%
& 824.14\%
& 1.709 \\

$t$-DCC
& 4.25\%
& 18.40\%
& 0.261
& -36.30\%
& 67.92\%
& 815.09\%
& 1.709 \\

aDCC
& 3.67\%
& 18.21\%
& 0.231
& -35.38\%
& 68.91\%
& 826.87\%
& 1.712 \\
\hline
\end{tabular}
}
\end{table}

The minimum-variance portfolio induced by DSP-MFSV-CAPM achieves the largest annualized return, highest Sharpe ratio, lowest annualized volatility, and lowest average gross exposure compared with the competing methods. 

We also report performance net of proportional transaction costs, following \cite{conlon2025asset}. We set the transaction cost to $c=20$ basis points per dollar traded. Since $TO_{t}$ denotes the one-way turnover, total traded notional is $2TO_{t}$, and the proportional cost deducted from portfolio wealth at rebalance $t$ is $2cTO_{t}$. The initial portfolio formation in each backtest period is treated as the starting allocation, so transaction costs are applied from the first subsequent rebalance onward. 

\begin{table}[H]
\centering
\caption{Combined Out-of-Sample Minimum-Variance Portfolio Performance Net of 20-bp Transaction Costs}
\label{tab:combined_minvar_tc20}
\resizebox{\textwidth}{!}{
\begin{tabular}{lrrrrrrr}
\hline
Method
& Annualized Return
& Annualized Volatility
& Sharpe
& Max. Drawdown
& Avg. Monthly Turnover
& Annualized Turnover
& Avg. Gross Exposure \\
\hline

MFSV
& 1.39\%
& 19.47\%
& 0.115
& -44.49\%
& 47.48\%
& 569.80\%
& 1.724 \\

Cholesky SV
& 1.79\%
& 18.88\%
& 0.133
& -39.02\%
& 53.55\%
& 642.60\%
& 1.651 \\

\textbf{DSP-MFSV-CAPM}
& \textbf{5.07\%}
& \textbf{18.16\%}
& \textbf{0.304}
& -35.45\%
& 56.68\%
& 680.13\%
& \textbf{1.642} \\

EW
& -3.99\%
& 24.27\%
& -0.090
& -58.99\%
& 151.05\%
& 1812.59\%
& 3.080 \\

Ledoit--Wolf
& 4.21\%
& 19.27\%
& 0.256
& -38.26\%
& \textbf{12.97\%}
& \textbf{155.62\%}
& 1.751 \\

Hard Threshold
& 3.63\%
& 19.56\%
& 0.227
& -39.42\%
& 18.07\%
& 216.82\%
& 1.987 \\

Loggle 
& 1.15\%
& 18.71\%
&0.098
&\textbf{-31.36\%}
&37.85\%
&454.18\%
&1.809\\

Gaussian DCC
& 0.33\%
& 18.24\%
& 0.051
& -36.69\%
& 68.68\%
& 824.14\%
& 1.709 \\

$t$-DCC
& 0.97\%
& 18.40\%
& 0.087
& -37.55\%
& 67.92\%
& 815.09\%
& 1.709 \\

aDCC
& 0.36\%
& 18.21\%
& 0.053
& -36.70\%
& 68.91\%
& 826.87\%
& 1.712 \\
\hline
\end{tabular}
}
\end{table}

After accounting for transaction costs, DSP-MFSV-CAPM continues to achieve the highest annualized return and Sharpe ratio, and the lowest annualized volatility. Its annualized return declines from $7.91\%$ to $5.07\%$, while its Sharpe ratio declines from $0.451$ to $0.304$. Thus, its relative performance is preserved after accounting for trading costs. Loggle continues to exhibit the least severe maximum drawdown. 

\subsection{Using a scalar multivariate dependence measure for leverage decisions }

For this exercise, all five dependence measures are oriented and normalized to lie in \([0,1]\), with larger values indicating greater dependence. Thus, \(S_t=0\) corresponds to low dependence and full risky exposure, while \(S_t=1\) corresponds to the maximum-dependence boundary and zero risky exposure, for holding month $t$.

The leverage exercise is based specifically on the DSP-MFSV-CAPM covariance matrix estimated before each holding month. The DSP-MFSV-CAPM covariance matrix is converted to a correlation matrix, from which the five measures of cross-sectional dependence are calculated. The risky component of the leverage exercise is the equally weighted portfolio of the 30 stocks, rather than the minimum-variance portfolio. For a chosen monthly exposure $L_{t}$, the daily strategy return for trading day $d$ during that month is $R^{\mathrm{strategy}}_{d}=R^{f}_{d}+L_{t}\left(R^{\mathrm{EW}}_{d}-R^{f}_{d}\right).$

Thus, $L_{t}$ of portfolio wealth is allocated to the risky equal-weighted portfolio and $1-L_{t}$ to the risk-free asset. As with the MVP weights, the dependence score and the resulting leverage decision are calculated once before the month begins and then held constant for every trading day in that month. 

The monthly risky exposure is set according to
\begin{equation*}
    L_{t} = L_{\max} (1-S_{t}),
\end{equation*}
with $L_{\max} =1$. Thus, when estimated dependence is low, the portfolio remains close to fully invested; when estimated dependence rises, risky exposure is reduced and the remainder is placed in the risk-free asset. 

Let $R_{t-1}^{\text{EW}}$ denote the cumulative realized simple return of the equally weighted risky portfolio during the preceding holding month, and let $R_{t-1}^{f}$ denote the corresponding cumulative simple return on the risk-free asset. Since the portfolio begins holding month $t-1$ with risky exposure $L_{t-1}$, its risky exposure immediately before the rebalance at the beginning of month $t$ has drifted to 
\begin{equation*}
    \tilde{L}_{t-1} := \frac{L_{t-1}(1+R_{t-1}^{\text{EW}})}{L_{t-1}(1+R_{t-1}^{\text{EW}}) + (1-L_{t-1})(1+R_{t-1}^{f})}.
\end{equation*}
The one-way turnover required to move to the new target exposure $L_{t}$ is then 
\begin{align*}
    \text{TO}_{t} &:=\frac{1}{2}(|L_{t}-\tilde{L}_{t-1}| + |(1-L_{t}) - (1-\tilde{L}_{t-1})|)\\
    &= |L_{t}- \tilde{L}_{t-1}|.
\end{align*}

\begin{table}[htbp]
\centering
\caption{Combined Out-of-Sample Dependence-Based Leverage Performance}
\label{tab:combined_leverage}
\resizebox{\textwidth}{!}{
\begin{tabular}{lrrrrrrr}
\hline
Rule
& Annualized Return
& Annualized Volatility
& Sharpe
& Max. Drawdown
& Avg. Leverage
& Avg. Monthly Turnover
& Annualized Turnover \\
\hline

\textbf{Determinant}
& \textbf{9.70\%}
& 16.62\%
& \textbf{0.576}
& -28.36\%
& 0.633
& 8.26\%
& 99.14\% \\

Average Pairwise
& 9.65\%
& 17.83\%
& 0.547
& -31.66\%
& 0.657
& 7.88\%
& 94.51\% \\

Effective Rank
& 7.26\%
& 11.98\%
& 0.556
& -21.71\%
& 0.477
& 10.58\%
& 126.90\% \\

Eigenvalue Concentration
& 9.63\%
& 17.14\%
& 0.560
& -29.83\%
& 0.636
& \textbf{7.57\%}
& \textbf{90.81\%} \\

Participation Ratio
& 3.38\%
& \textbf{5.12\%}
& 0.467
& \textbf{-9.71\%}
& 0.211
& 8.16\%
& 97.89\% \\
\hline
\end{tabular}
}
\end{table}

Our proposed determinant score achieves the largest annualized return and Sharpe ratio compared with the methods exhibiting similar levels of annualized volatility. The participation ratio score achieves the smallest annualized volatility and least severe maximum drawdown, but this reflects its risk aversion as indicated by its low average leverage, which comes at the cost of a lower annualized return. 

We apply the same 20-basis-point transaction-cost convention to the drift-adjusted exposure turnover defined above. 

\begin{table}[htbp]
\centering
\caption{Combined Out-of-Sample Dependence-Based Leverage Performance Net of 20-bp Transaction Costs}
\label{tab:combined_leverage_tc}
\resizebox{\textwidth}{!}{
\begin{tabular}{lrrrrrrr}
\hline
Rule
& Annualized Return
& Annualized Volatility
& Sharpe
& Max. Drawdown
& Avg. Leverage
& Avg. Monthly Turnover
& Annualized Turnover \\
\hline

\textbf{Determinant}
& \textbf{9.27\%}
& 16.61\%
& \textbf{0.553}
& -28.44\%
& 0.633
& 8.26\%
& 99.14\% \\

Average Pairwise
& 9.25\%
& 17.83\%
& 0.526
& -31.72\%
& 0.657
& 7.88\%
& 94.51\% \\

Effective Rank
& 6.72\%
& 11.97\%
& 0.515
& -21.83\%
& 0.477
& 10.58\%
& 126.90\% \\

Eigenvalue Concentration
& 9.24\%
& 17.14\%
& 0.540
& -29.89\%
& 0.636
& \textbf{7.57\%}
& \textbf{90.81\%} \\

Participation Ratio
& 2.98\%
& \textbf{5.11\%}
& 0.393
& \textbf{-9.84\%}
& 0.211
& 8.16\%
& 97.89\% \\
\hline
\end{tabular}
}
\end{table}

The main conclusions are also unchanged after accounting for transaction costs. The determinant-based rule retains the highest annualized return, $9.27\%$, and the highest Sharpe ratio, 0.553. Transaction costs have little effect on annualized volatility or maximum drawdown, while reducing returns and Sharpe ratios as expected. The participation ratio continues to produce the lowest volatility and least severe drawdown, but does so with substantially lower average risky exposure. 

\section{Large cross-section example}
Below we provide one proposal for how to extend our methodology to a larger portfolio of stocks. 

First, fit static CAPMs to every stock in the portfolio: 
\begin{equation*}
    r_{a,t} = \alpha_{a} + \beta_{a}r_{M,t} + u_{a,t}. 
\end{equation*}
Construct the matrix of residuals $U$ and compute its sample correlation matrix. Next, compute the associated distance matrix using the metric 
\begin{equation*}
    d_{ij} = \sqrt{2(1-\rho_{i,j})},
\end{equation*}
where $\rho_{i,j}$ is the $(i,j)$th entry of the sample correlation matrix. Apply the $K$-medoids algorithm to the distance matrix for multiple candidate numbers of clusters $K$. In our examples, we consider candidate values of $K$ from 30 to 50. Select the value of $K$ with the largest average silhouette width. Designate each cluster’s medoid as its representative stock. Next, fit the DSP-MFSV-CAPM to the \(K\) representative stocks. This gives the time-varying covariance matrix of the representative stocks, denoted by $\Sigma_{\text{rep},t}$. Regress every stock on all \(K\) representative stocks, i.e.,
\begin{equation*}
    \boldsymbol{r}_{t} = \boldsymbol{\alpha} + A\boldsymbol{r}_{\text{rep},t} + \boldsymbol{e}_{t}. 
\end{equation*}
The matrix \(A\) is ordered so that the rows corresponding to the representative stocks appear first and form the identity matrix, with their residuals set to zero, while the remaining rows correspond to the non-representative stocks. Allocate the non-representative regression residuals to their respective clusters. Fit an MFSV model separately to the residuals in each cluster. For each residual block \(g\), the model uses $\min(3,n_{g}-1)$ latent factors, where $n_{g}$ is the number of non-representative stocks in the block. Construct the residual covariance matrix $\Sigma_{\epsilon,t}$ from the cluster-specific MFSV covariance matrices. This matrix is block diagonal, with stocks ordered by cluster. Cross-cluster residual covariances are set to zero, as are the residual rows and columns corresponding to the representative stocks. Finally, reconstruct the full $N\times N$ covariance matrix 
\begin{equation*}
    \Sigma_{t} = A\Sigma_{\text{rep},t}A^\top + \Sigma_{\epsilon,t}. 
\end{equation*}
Let $A_{-\text{rep}}$ denote the rows of $A$ corresponding to the non-representative stocks. Thus 
\begin{equation*}
    A = \begin{pmatrix}
        I_{K}\\
        A_{-\text{rep}}
    \end{pmatrix}, \qquad \Sigma_{\epsilon,t} = \begin{pmatrix}
        0 & 0 \\
        0 & D_{t}\\
    \end{pmatrix},
\end{equation*}
where $D_{t}$ is the block-diagonal covariance matrix of the non-representative residuals. Therefore, 
\begin{equation*}
    \Sigma_{t} = \begin{pmatrix}
        \Sigma_{\text{rep},t} & \Sigma_{\text{rep},t}A^\top_{-\text{rep}}\\
        A_{-\text{rep}}\Sigma_{\text{rep},t}& A_{-\text{rep}}\Sigma_{\text{rep},t}A^\top_{-\text{rep}}+D_{t}
    \end{pmatrix}.
\end{equation*}
The Schur complement of $\Sigma_{\text{rep},t}$ is $D_{t}$. Therefore, 
\begin{equation*}
    \det(\Sigma_{t}) = \det(\Sigma_{\text{rep},t}) \det(D_{t}) = \det(\Sigma_{\text{rep},t})\prod\limits_{g=1}^{K}\det(\Sigma_{\epsilon,g,t}), 
\end{equation*}
where $\Sigma_{\epsilon,g,t}$ is the covariance matrix for the non-representative residuals in cluster $g$. Consequently, for the correlation matrix $R_{t}$ derived from $\Sigma_{t}$,
\begin{equation*}
    \log(\det(R_{t})) = \log(\det(\Sigma_{\text{rep},t})) + \sum\limits_{g=1}^{K} \log(\det(\Sigma_{\epsilon,g,t})) - \sum\limits_{i=1}^{N}\log(\Sigma_{ii,t}). 
\end{equation*}
This allows us to compute the determinant score time series 
\begin{equation*}
    S_{t} = 1- \det(R_{t})^{1/N} = 1- \exp\{ \frac{1}{N}\log(\det(R_{t}))\}. 
\end{equation*} 

In the subsections below, we apply this procedure to two portfolios constructed around the 2008 and 2020 periods and present the resulting posterior-mean dependence-score series. Each portfolio contains \(N=1000\) U.S. stocks and is constructed using the same selection and data-processing procedure.

Daily stock returns and security characteristics were obtained from the Center for Research in Security Prices (CRSP) through Wharton Research Data Services (WRDS). The initial universe was restricted to ordinary common equity issued by U.S.-incorporated companies and primarily listed on the New York Stock Exchange (NYSE), NYSE American/AMEX, or NASDAQ. Securities were required to be actively traded under regular-way settlement, to have a positive and finite market capitalization, and to have a complete series of finite daily returns throughout the relevant estimation period. This complete-return requirement produces a balanced panel and ensures that every covariance matrix is estimated using the same observations for all stocks.
Portfolio membership was determined before the beginning of each return sample. For the 2008 application, eligible securities were identified using their December 2005 CRSP characteristics and ranked by their December 2005 market capitalizations. The 1,000 largest eligible stocks were retained, yielding a balanced panel of 1,000 stocks and 1,006 daily observations over the 2006--2009 period. The 2020 portfolio was constructed analogously using December 2018 security characteristics and market capitalizations. The 1,000 largest eligible stocks were retained, producing a balanced panel of 1,000 stocks and 1,257 daily observations over the 2019--2023 period. In each application, the selected portfolio was held fixed throughout the corresponding estimation period; therefore, stocks were not added, removed, or reweighted as part of the sample-selection procedure. Daily market excess returns and risk-free rates were obtained from Kenneth French's Data Library \cite{fama2023production}. The dependence scores are computed from the posterior-mean covariance matrices.

\subsection{Subprime mortgage crisis}
Using the procedure described above, we obtain $K=50$ representative stocks. This results in the following plot of dependence dynamics. 

\begin{figure}[H]
    \centering
    \includegraphics[width = 0.65\linewidth]{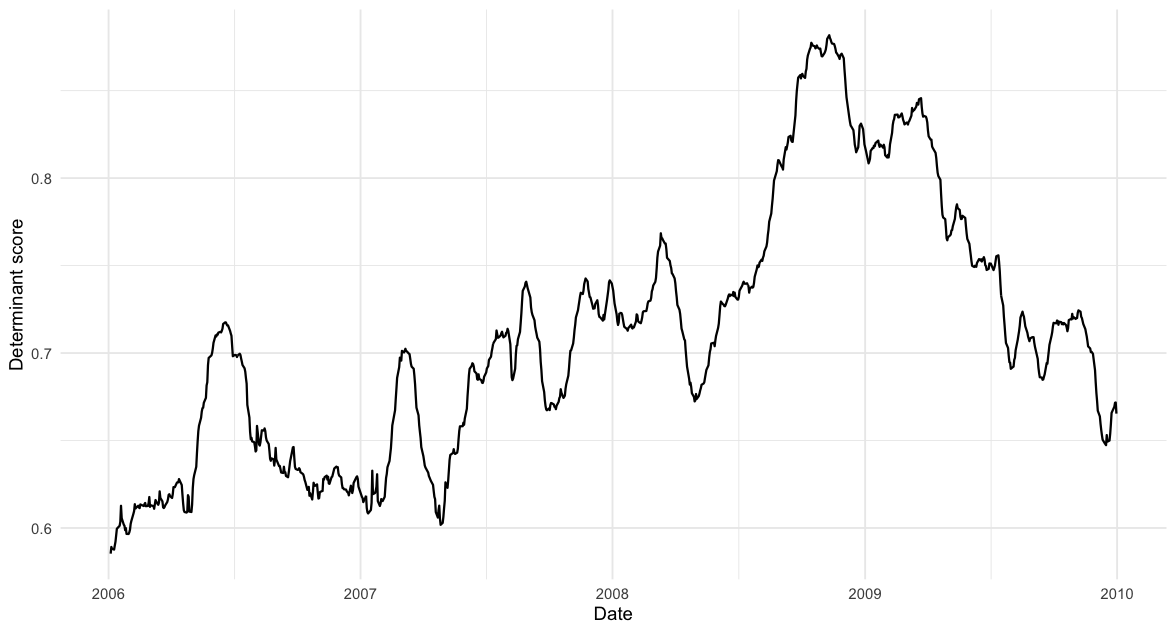}
    \caption{Determinant-score series for the 1,000-stock portfolio}
    \label{fig:2008_large_cross_section_plot}
\end{figure}
The pairwise sample correlations among the posterior-mean score series for the 1,000-stock portfolio and the two 30-stock portfolios are provided below. 
\begin{table}[H]
\centering
\caption{Correlation comparison across portfolio specifications}
\label{tab:portfolio_correlations}
\begin{tabular}{lcc}
\hline
Comparison & Pearson & Spearman \\
\hline
1,000 vs.\ Technology 30
& 0.971 & 0.957 \\

1,000 vs.\ Diversified 30
& 0.976 & 0.962 \\

Technology 30 vs.\ Diversified 30
& 0.961 & 0.940 \\
\hline
\end{tabular}
\end{table}

All three time series are highly correlated. The high Pearson correlations indicate that all three series have similar dynamics, rising and falling at similar times. The high Spearman correlations indicate that the series classify dates similarly in terms of low and high dependence. This is illustrated graphically below. 

\begin{figure}[H]
    \centering
    \includegraphics[width = 0.65\linewidth]{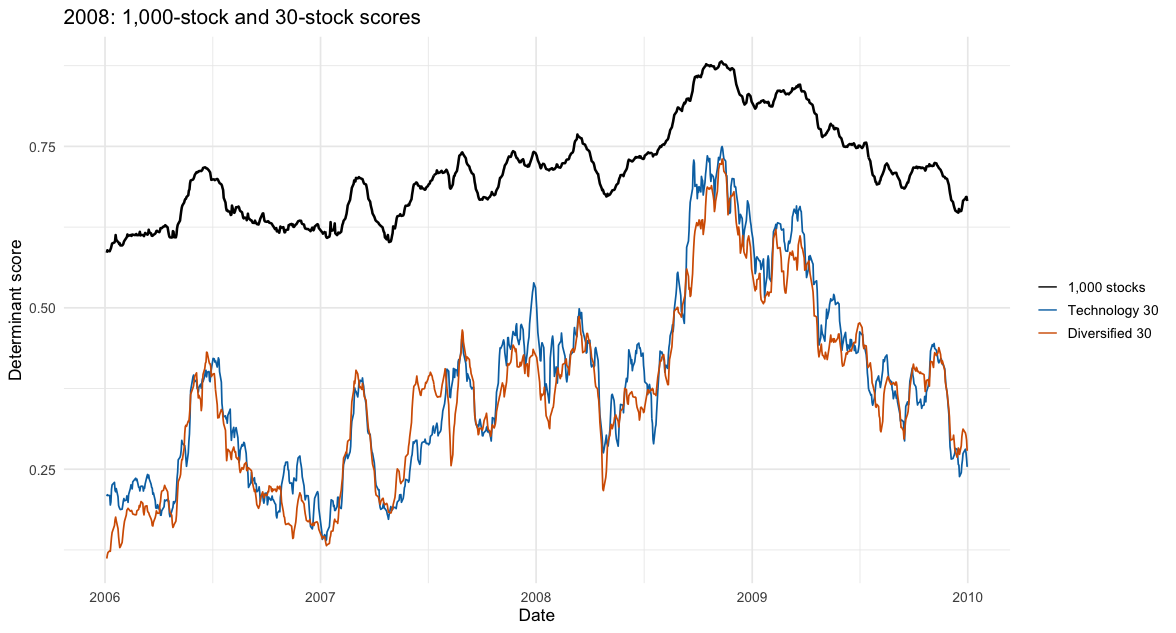}
    \caption{Posterior-mean score series for the three portfolio specifications}
    \label{fig:2008_large_cross_section_plot_comparison_1}
\end{figure}

All three portfolio scores increase and decrease at similar points in time. However, the level of the 1,000-stock portfolio score is consistently greater than those of the 30-stock portfolio scores. All three portfolio scores peak around the same time. The 1,000-stock portfolio reaches its maximum dependence score of 0.882 on 10 November 2008. The Technology 30 and Diversified 30 portfolios reach their maximum scores on 11 November 2008, with values of 0.750 and 0.730, respectively. 

Furthermore, considering the dates corresponding to the top 10\% of dependence scores for all three series, there are 89 common dates out of 101 between the 1,000-stock portfolio and each of the two 30-stock portfolios. There are 92 common dates out of 101 between the two 30-stock portfolios.

We also consider how much a 30-stock score depends on which 30 stocks are selected. In particular, on each date, the score is calculated for 50 different random 30-stock subsets drawn from the same 1,000-stock portfolio. For each date, we then plot the median score and the 5th-to-95th percentile range of the scores across the 50 random subsets. We provide summary statistics for the 1,000-Stock Portfolio score and the pointwise median below. Note that the same 50 fixed subsets are used on every date. Their \(30 \times 30\) covariance matrices are principal submatrices of the 1,000-stock covariance matrix. 

\begin{table}[H]
\centering
\caption{Summary statistics across the 1,000-stock and 30-stock portfolios}
\label{tab:portfolio_summary_statistics}
\begin{tabular}{lcc}
\hline
Statistic & 1,000-Stock Portfolio Score & Pointwise Median of Random 30-Stock Subset Scores \\
\hline
Minimum
& 0.585 & 0.109 \\

First quartile
& 0.648 & 0.218 \\

Median
& 0.708 & 0.331 \\

Mean
& 0.710 & 0.333 \\

Third quartile
& 0.749 & 0.398 \\

Maximum
& 0.882 & 0.677 \\
\hline
\end{tabular}
\end{table}

Across the sample, the median levels of the pointwise 5th-percentile, median, and 95th-percentile series across the random subsets are 0.285, 0.331, and 0.378, respectively. The median daily width between the two percentile boundaries is 0.091. By comparison, the difference between the median 1,000-stock score and the median random 30-stock score is 0.377. Thus, the variation attributable to the particular selection of 30 stocks is considerably smaller than the difference between the 30-stock and 1,000-stock score levels. This can be seen graphically below. 

\begin{figure}[H]
    \centering
    \includegraphics[width = 0.65\linewidth]{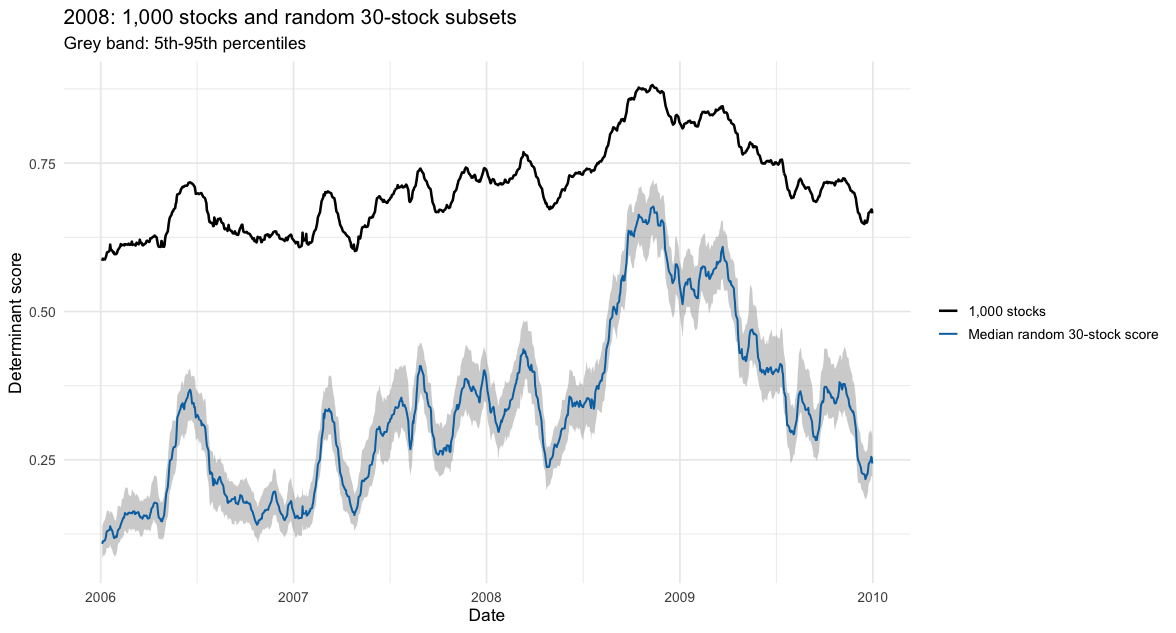}
    \caption{Random 30-stock subset scores and the 1,000-stock portfolio score}
    \label{fig:2008_large_cross_section_plot_comparison_2}
\end{figure}

The pointwise median reproduces the same major movements in dependence, including the increase during 2008. The gray band widens and narrows on some dates. This indicates that the particular stocks selected have a greater effect on the score on those dates. The 1,000-stock line remains well above the entire gray band. Therefore, the difference between the 1,000-stock and 30-stock scores cannot be explained simply by the selection of an unusual set of 30 stocks. Instead, it reflects a systematic difference between evaluating the broad cross-section and the smaller submatrices. The higher 1,000-stock score should not be interpreted as evidence of greater accuracy. These results show that broadening the cross-section materially changes the score level while largely preserving its temporal dynamics.

\subsection{COVID-19 pandemic}
Using the procedure described above, we obtain $K=39$ representative stocks. This results in the following plot of dependence dynamics.  
\begin{figure}[H]
    \centering
    \includegraphics[width = 0.65\linewidth]{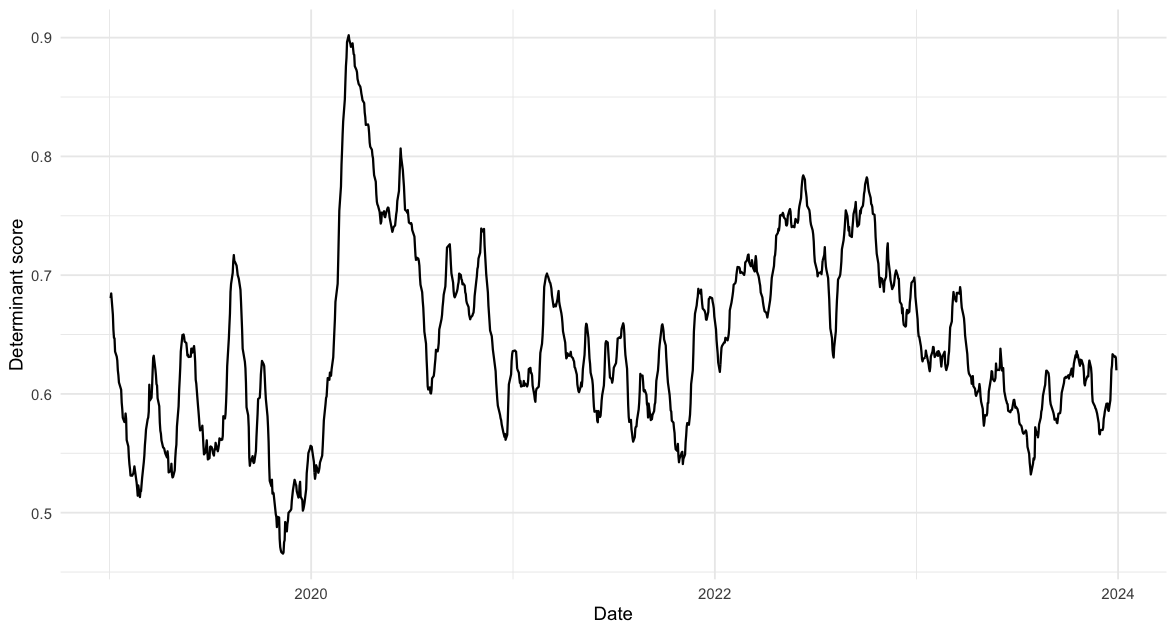}
    \caption{Determinant-score series for the 1,000-stock portfolio}
    \label{fig:2020_large_cross_section_plot}
\end{figure}

The pairwise sample correlations among the posterior-mean score series for the 1,000-stock portfolio and the two 30-stock portfolios are provided below. 

\begin{table}[H]
\centering
\caption{Correlation comparison across portfolio specifications}
\label{tab:portfolio_correlations_2}
\begin{tabular}{lcc}
\hline
Comparison & Pearson & Spearman \\
\hline
1,000 vs.\ Technology 30
& 0.892 & 0.889 \\

1,000 vs.\ Diversified 30
& 0.964 & 0.962 \\

Technology 30 vs.\ Diversified 30
& 0.891 & 0.886 \\
\hline
\end{tabular}
\end{table}

All the correlations remain high but are weaker than their 2008 counterparts. For example, the Pearson and Spearman correlations between the 1,000-stock portfolio and the Technology 30 portfolio are 0.892 and 0.889, respectively. This suggests that technology-specific dependence movements may be more important during the 2019--2023 period. The global semiconductor shortage also occurred during this period. 

\begin{figure}[H]
    \centering
    \includegraphics[width = 0.65\linewidth]{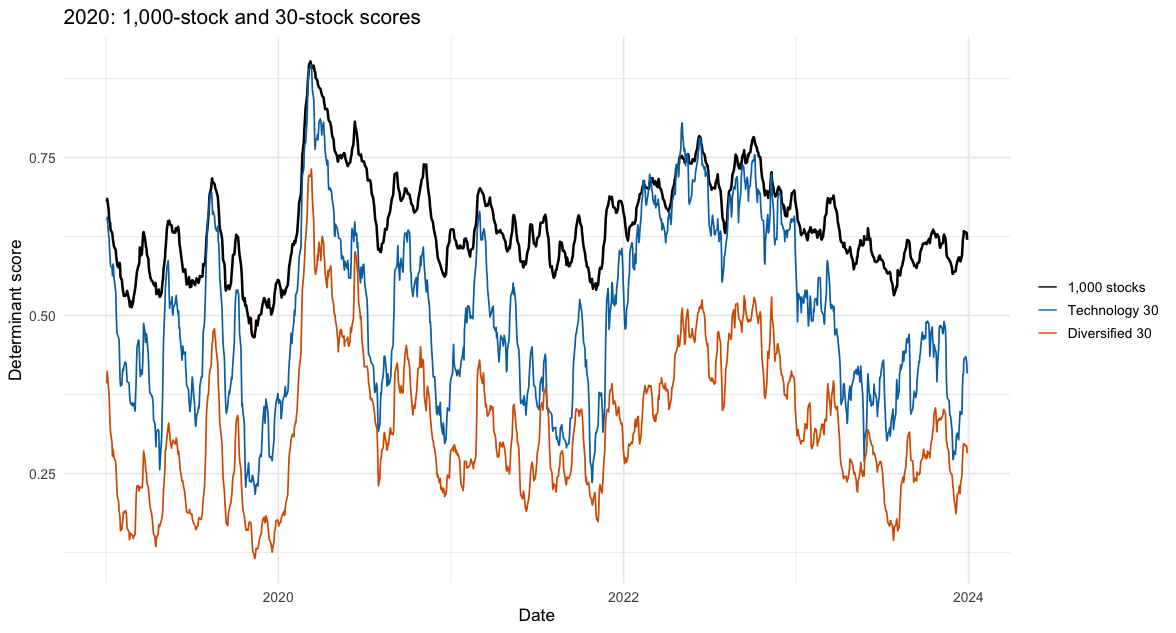}
    \caption{Posterior-mean score series for the three portfolio specifications}
    \label{fig:2020_large_cross_section_plot_comparison_1}
\end{figure}

The plot also shows that the level of the Technology 30 series is closer to that of the 1,000-stock series than is the level of the Diversified 30 series. 

All three portfolio scores also peak at similar times. In particular, the 1,000-stock portfolio reaches its maximum score of 0.902 on 9 March 2020. The Technology 30 portfolio reaches its maximum score of 0.896 on 10 March 2020, while the Diversified 30 portfolio reaches its maximum score of 0.732 on 11 March 2020. The Technology 30 peak is very close in level to the 1,000-stock portfolio peak, whereas the Diversified 30 portfolio has a substantially lower score. Nevertheless, the Diversified 30 portfolio tracks the broader series more consistently over the entire sample, as shown by its higher correlations. 

Again, considering the dates corresponding to the top 10\% of scores in each portfolio, there is still substantial overlap, but less than in the subprime mortgage example. In particular, the 1,000-stock and Technology 30 portfolios have 84 common dates out of 126, the 1,000-stock and Diversified 30 portfolios have 106 common dates out of 126, and the two 30-stock portfolios have 92 common dates out of 126. This is consistent with the findings above. The Diversified 30 portfolio identifies almost the same set of high-dependence dates as the 1,000-stock portfolio. However, the Technology 30 portfolio agrees less often, which may reflect sector-specific dependence movements. Nevertheless, its 66.7\% overlap remains much higher than the roughly 10\% overlap expected from two unrelated rankings. 

Below are summary statistics for the 1,000-stock portfolio and the random 30-stock subsets. 

\begin{table}[H]
\centering
\caption{Summary statistics for the 1,000-stock portfolio and random 30-stock subsets}
\label{tab:portfolio_summary_statistics_random30}
\begin{tabular}{lcc}
\hline
Statistic & 1,000-Stock Portfolio Score & Pointwise Median of Random 30-Stock Subset Scores \\
\hline
Minimum
& 0.466 & 0.127 \\

First quartile
& 0.589 & 0.252 \\

Median
& 0.632 & 0.318 \\

Mean
& 0.644 & 0.342 \\

Third quartile
& 0.695 & 0.414 \\

Maximum
& 0.902 & 0.792 \\
\hline
\end{tabular}
\end{table}

Across the sample, the median levels of the pointwise 5th-percentile, median, and 95th-percentile series across the random subsets are 0.274, 0.318, and 0.383, respectively. The median daily width between the two percentile boundaries is 0.104. By comparison, the difference between the median 1,000-stock score and the median random 30-stock score is 0.314. Thus, the variation attributable to the particular selection of 30 stocks is considerably smaller than the difference between the 30-stock and 1,000-stock score levels. The median daily width is larger in the COVID-19 example than in the subprime mortgage example, at 0.104 compared with 0.091, indicating slightly greater sensitivity to the selection of the 30 stocks. 

\begin{figure}[H]
    \centering
    \includegraphics[width = 0.65\linewidth]{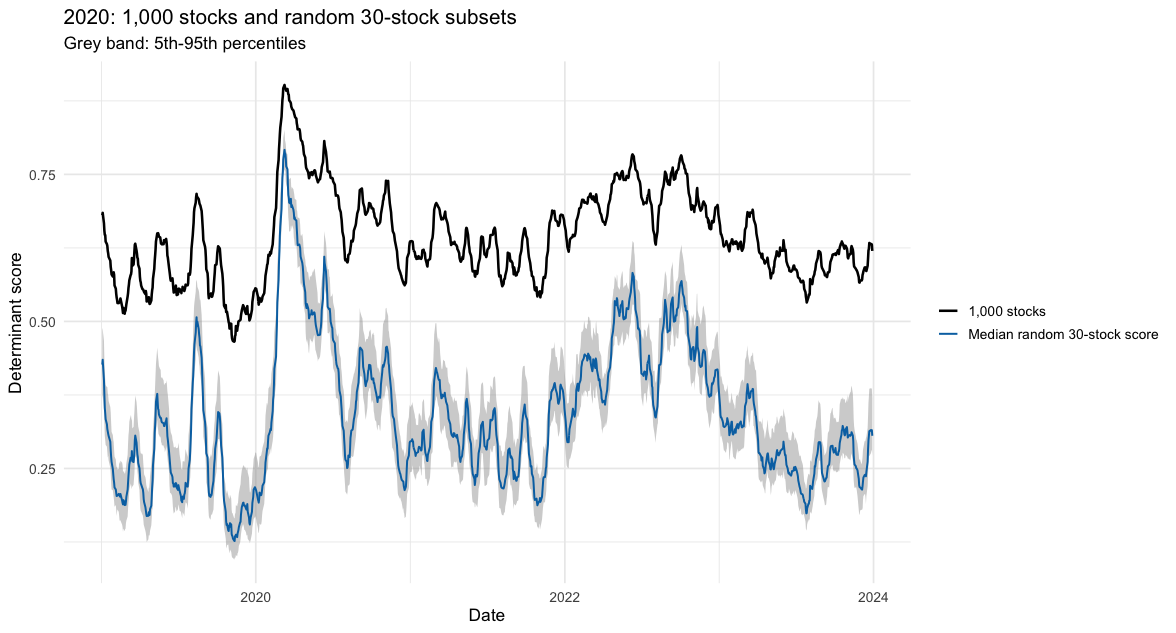}
    \caption{Random 30-stock subset scores and the 1,000-stock portfolio score}
    \label{fig:2020_large_cross_section_plot_comparison_2}
\end{figure}

The median line captures the March 2020 spike and much of the subsequent variation. The gray band widens substantially during some turbulent periods, showing that stock selection matters more during those periods. As in the subprime mortgage example, the 1,000-stock portfolio line is consistently above the gray band. The 1,000-stock score therefore exhibits a systematically higher level of measured dependence than scores calculated from the 30-stock submatrices. The COVID-19 plot is more irregular than the subprime mortgage plot, and the width of the gray band varies more noticeably. This is consistent with the weaker Technology 30 correlations and lower date overlaps. The higher 1,000-stock score should not be interpreted as evidence of greater accuracy. These results show that broadening the cross-section materially changes the score level while largely preserving its temporal dynamics.

\section{Point cloud contraction interpretation of proposed score}
Let $\boldsymbol{x} \in \mathbb{R}^{N}$ be the vector of standardized returns with $R \in \mathbb{S}_{++}^{N}$, the set of symmetric $N\times N$ positive definite matrices. Consider the quadratic form 
\begin{equation*}
    Q_{R}(\boldsymbol{x}) = \boldsymbol{x}^\top R^{-1} \boldsymbol{x} \text{ and the set } E_{R}(c) := \{ \boldsymbol{x} \in \mathbb{R}^{N} : \boldsymbol{x}^\top R^{-1} \boldsymbol{x} \leq c\}.
\end{equation*}
Notice that on $E_{R}(c)$, 
\begin{equation*}
    \boldsymbol{x}^\top R^{-1}\boldsymbol{x} = \boldsymbol{x}^\top R^{-\frac{1}{2}}R^{-\frac{1}{2}}\boldsymbol{x} = ||R^{-\frac{1}{2}}\boldsymbol{x}||_{2}^{2} \leq c,
\end{equation*}
where $||\cdot||_{2}$ denotes the Euclidean norm. Let $\boldsymbol{u} = c^{-\frac{1}{2}} R^{-\frac{1}{2}} \boldsymbol{x}$ which yields the transformed constraint $||\boldsymbol{u}||_{2}^{2} \leq 1$. Clearly, $E_{R}(c)$ is the image of the unit ball $B_{N} = \{ \boldsymbol{u}: ||u||_{2} \leq 1 \}$ under the linear map $\boldsymbol{u} \rightarrow \sqrt{c} R^{\frac{1}{2}}\boldsymbol{u}$. Then the Jacobian determinant is given by
\begin{equation*}
    |\det(\sqrt{c}R^{\frac{1}{2}})|= c^{\frac{p}{2}}\det(R)^{\frac{1}{2}}.
\end{equation*}
Therefore, 
\begin{equation*}
    \text{Vol}(E_{R}(c)) = \text{Vol}(\sqrt{c}R^{\frac{1}{2}}B_{N}) = |\det(\sqrt{c}R^{\frac{1}{2}})| \text{Vol}(B_{N}) = c^{\frac{p}{2}} \det(R)^{\frac{1}{2}} \text{Vol}(B_{N}). 
\end{equation*}
Then, 
\begin{equation*}
    \frac{\text{Vol}(E_{R}(c))}{\text{Vol}(E_{I_{p}}(c))} = \frac{c^{\frac{p}{2}}\det(R)^{\frac{1}{2}}\text{Vol}(B_{N})}{c^{\frac{p}{2}}\det(I_{N})^{\frac{1}{2}}\text{Vol}(B_{N})} = \det(R)^{\frac{1}{2}}. 
\end{equation*}
By Hadamard's inequality, $\det(R) \leq \prod\limits_{i=1}^{N}R_{ii} = 1$. Note that equality occurs iff $R = I_{N}$. So, we have $\det(R)^{\frac{1}{2}} \leq 1$. This means $E_{R}(c)$ has the largest volume when assets are uncorrelated. Observe that
\begin{equation*}
    (\frac{\text{Vol}(E_{R}(c))}{\text{Vol}(E_{I_{N}}(c))})^{\frac{2}{N}} = \det(R)^{\frac{1}{N}} \implies \text{score}(R) = 1 - (\frac{\text{Vol}(E_{R}(c))}{Vol(E_{I_{N}}(c))})^{\frac{2}{N}}. 
\end{equation*}

\section{The prior distribution: structure}
\label{supp:prior}

This section supports Section~2  of the main paper. Section~\ref{supp:structure} describes how the two components of the prior are coupled, Section~\ref{supp:roles} derives the marginal behavior of the dynamic shrinkage prior, Section~\ref{supp:staticdyn}
presents the comparison between static and dynamic shrinkage, and
Section~\ref{supp:why} collects the justification of each component of
the prior.

\subsection{Structure and interaction of the prior components}
\label{supp:structure}

The DSP and MFSV components of the prior meet in the full conditional
distribution of the state variables, displayed in Section~K, whose
posterior precision is
$Q_\beta=X'\Sigma_\epsilon^{-1}X+(L'\otimes I_2)\Sigma_\omega^{-1}(L\otimes I_2)$.
The prior innovation variances in $\Sigma_\omega$ are supplied by the
DSP component, and the observation error variances in $\Sigma_\epsilon$
are supplied by the MFSV component through
\begin{equation}
\sigma_{a,t}^{2}
=\big(\Lambda\tilde\Sigma_t\Lambda'+\bar\Sigma_t\big)_{aa}.
\label{eq:supp-coupling}
\end{equation}

Equation~\eqref{eq:supp-coupling} is the basis for the precision-ratio interpretation of
shrinkage given in Section~2.1.3 of the main paper: the noise scale in
the state update of asset $a$ is estimated, time-varying, and shared
across assets. The same coupling implies that the two components compete
to explain the observed comovement, an allocation that Section~B
discusses qualitatively and that the entrywise decomposition of
Section~2.1.3, reproduced here for reference, makes quantitative. For
$i\ne j$,
\begin{equation}
\operatorname{cov}(\boldsymbol r_t)_{ij}
=\operatorname{var}(r_{M,t})\,
\mathrm{E}[\beta_{i,t}]\,\mathrm{E}[\beta_{j,t}]
+\sum_{k=1}^{r}\Lambda_{ik}\Lambda_{jk}\,e^{\tilde h_{t,k}}.
\label{eq:supp-offdiag}
\end{equation}
Tightening the DSP shrinkage routes comovement into the residual
component, and tightening the normal--gamma shrinkage routes it into
the systematic component.

The normal--gamma prior shrinks the column norms of $\Lambda$
continuously, so the implied eigenvalue profile of the residual
covariance decays gradually, in line with the empirical spectrum of
equity panels, rather than exhibiting the step profile of a hard
rank-$r$ model; the fixed value $r=3$ therefore acts as an upper bound
while the prior performs the effective selection (Section~A). In
addition, the loadings enter (3) through the products
$\Lambda_{ik}\Lambda_{jk}$, which have prior mean zero, so the residual
component contributes dispersion but no prior mean to the off-diagonal
entries.

\subsection{Marginal behavior of the dynamic shrinkage prior}
\label{supp:roles}

The stationary distribution of $h_{\beta_a,t}$ is the distribution of
$\mu_{\beta_a}+\sum_{j\ge0}\phi_{\beta_a}^{j}\eta_{t-j}$ with
$\eta\sim Z(\tfrac12,\tfrac12,0,1)$. The density of $\eta$ is
$f_\eta(x)=\tfrac{1}{2\pi}\operatorname{sech}(x/2)$, which equals the
density of $\log\lambda^{2}$ when $\lambda\sim C^{+}(0,1)$. The DSP therefore generalizes the horseshoe prior, which it recovers at
$\phi_{\beta_a}=0$, by introducing temporal persistence. As shown in Section~J.4, the
stationary characteristic function is
$\varphi_h(u)=\prod_{k\ge0}\operatorname{sech}(\pi\phi^{k}u)$, and the
stationary density is the location mixture
\begin{equation}
f_h(x)=\int_{-\infty}^{\infty}\frac{1}{2\pi}
\operatorname{sech}\!\Big(\frac{x-\mu-y}{2}\Big)\,H(dy),
\label{eq:supp-hmarg}
\end{equation}
where $H$ has characteristic function
$\prod_{k\ge1}\operatorname{sech}(\pi\phi^{k}u)$. From the moment generating function $M_\eta(t)=\sec(\pi t)$, we
have $\operatorname{var}(\eta)=\pi^{2}$, so
$\operatorname{var}(h)=\pi^{2}/(1-\phi^{2})$. At the prior mean
$\phi_{\beta_a}=2/3$, the stationary standard deviation of the
log-variance process is approximately $4.2$, so the central $95\%$
interval of the stationary distribution spans more than seven orders of
magnitude in the innovation variance. This range accommodates strong
shrinkage in stable periods and essentially unpenalized innovations
during periods of change. The shrinkage weight $\kappa_{a,t}$ introduced in Section~2.1.3 admits a sharper description: the conditional posterior mean of an innovation
is approximately $(1-\kappa_{a,t})$ times its unshrunk value, and the
marginal distribution of $\kappa_{a,t}$ is of
$\mathrm{Beta}(\tfrac12,\tfrac12)$ type, so a genuine change is
transmitted without being smeared over neighboring periods. Integrating the local
scale out of the stationary distribution, the marginal prior on a single
innovation is the normal scale mixture
$\omega\sim\int N(0,e^{h})f_h(h)\,dh$. The following result shows which
features of this prior depend on the persistence parameter.

\newtheorem{propI}{Proposition}
\renewcommand{\thepropI}{I.\arabic{propI}}

\begin{propI} \label{prop:supp-invariance}
For every $\phi\in(-1,1)$, the marginal prior on $\omega_{\beta_a,t}$
has (i) polynomial tails of the same index as the horseshoe prior,
$p(\omega)\asymp|\omega|^{-2}$ as $|\omega|\to\infty$, and (ii) a
logarithmic pole at the origin, $p(\omega)\asymp\log(1/|\omega|)$ as
$\omega\to0$.
\end{propI}

\begin{proof}
Since $f_\eta(x)\sim\pi^{-1}e^{-|x|/2}$, the summands
$\phi^{j}\eta_{t-j}$ of $h-\mu$ have exponential tails with rates
$1/(2\phi^{j})$, which are increasing in $j$. The slowest rate, that of
the $j=0$ term, dominates the tail of the sum, so
$P(h-\mu>x)\asymp e^{-x/2}$, and symmetrically for the lower tail.
Consequently, $\lambda^{2}=e^{h-\mu}$ satisfies
$P(\lambda^{2}>u)\asymp u^{-1/2}$ and has density $\asymp u^{-1/2}$ near
zero, which is the behavior of the square of a half-Cauchy random
variable. Substituting these bounds into the normal scale mixture and
applying the computation of Carvalho et al.\ [2010] for the horseshoe
prior yields (i) and (ii).
\end{proof}

Proposition~\ref{prop:supp-invariance} separates the roles of the local and persistence
parameters: the local scales determine which innovations survive,
through a shrinkage profile whose tail and origin behavior do not
depend on $\phi_{\beta_a}$, while the persistence parameter determines
how much temporal information is pooled into this decision and affects
only the intermediate part of the profile. This pooling underlies the calibration of the credible intervals reported in
Tables~1--5, while the heavy tails transmit large innovations without
attenuation. The implied temporal dependence has a closed form.
Writing
$\log\omega_{\beta_a,t}^{2}
=\log(\tau_{a,0}^{2}\tau_{\beta_a}^{2})+h_{\beta_a,t}+\log z_{t}^{2}$
with $z_t\sim N(0,1)$ and using
$\operatorname{var}(\log\chi^{2}_{1})=\pi^{2}/2$,
\begin{equation}
\operatorname{corr}\!\big(\log\omega_{\beta_a,t}^{2},\,
\log\omega_{\beta_a,t-k}^{2}\big)
=\phi_{\beta_a}^{\,k}\cdot\frac{2}{3-\phi_{\beta_a}^{2}},
\end{equation}
in which the nuisance scale $\pi^{2}$ cancels; the lag-one value at the
prior mean equals $0.52$. This identity expresses the persistence
parameter in terms of an interpretable quantity, the autocorrelation of
the magnitudes of the coefficient innovations.

\subsection{Static versus dynamic shrinkage}
\label{supp:staticdyn}
To isolate the contribution of temporal dependence in the shrinkage process, we conduct an additional simulation comparing the proposed DSP-MFSV-CAPM specification, in which $(\phi_{\alpha_{a}},\phi_{\beta_{a}})$ are estimated separately for each asset, with an otherwise identical specification in which $\phi_{\alpha_{a}},\phi_{\beta_{a}} = 0$ are fixed for every asset. In particular, we build on the simulation design in Section 4.1 of  \cite{kowal2019dynamic} and adapt it to our multivariate DSP-MFSV-CAPM setting. We simulate \(N=30\) asset-return series over \(T=200\) observations from
\begin{equation*}
    y_{i,t}=\alpha_i+\beta_{i,t}r_{m,t}+\epsilon_{i,t}, \qquad r_{m,t}\sim N(0,1)
\end{equation*}
and divide the assets equally among three coefficient regimes. All betas are centered at 2. In the first group, \(\beta_{i,t}=2\) throughout. In the second group, the beta is piecewise constant, taking values \(4\) for \(t=41,\ldots,80\), \(0\) for \(t=121,\ldots,160\), and 2 otherwise. In the third group, the beta coefficient follows a random-walk bridge during the first half of the sample: for $t=2,\dots,100$, innovations are generated from $N(0,0.35^{2})$ and adjusted so that the coefficient returns exactly to its baseline at $t=100$. The coefficient is then held exactly constant at $2$ for $t=101,\dots,200$. The regression errors $\epsilon_{i,t}$ are generated from a three-factor MFSV model in which the factor and idiosyncratic log-volatility innovation standard deviations are set to zero.

Across the 100 simulation replications, the proposed dynamic specification achieves an RMSE of 0.00254 when estimating the determinant score time series, compared with an RMSE of 0.00313 for the static model. To explore where these differences in score recovery arise, we examine how the relative performance of the two specifications evolves over time in the figure below. The upper panel reports cumulative excess absolute score error for the static specification relative to the dynamic model. The lower panel reports the mean log ratio of the posterior beta-innovation scales under dynamic and static shrinkage across assets and replications, smoothed using an 11-period moving average. 

\begin{figure}[H]
    \centering
    \includegraphics[width = 0.65\linewidth]{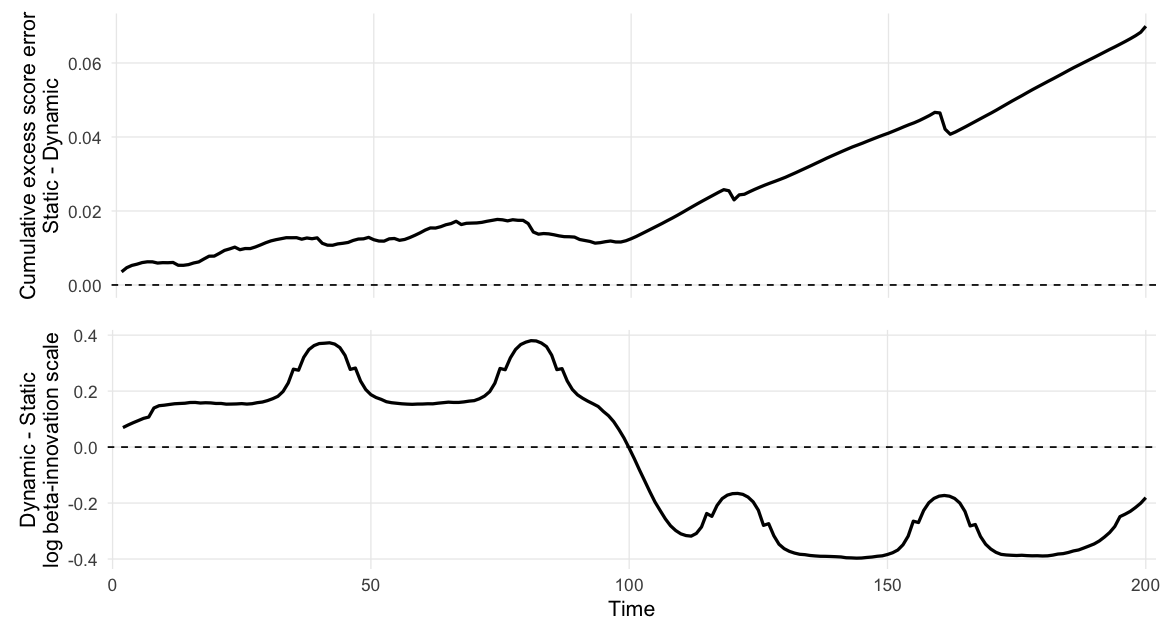}
    \caption{Dynamic vs static shrinkage}
    \label{fig:dynamic_v_static_shrinkage}
\end{figure}

The static specification accumulates more error over the sample. During
the period of persistent coefficient variation, the dynamic
specification permits innovation scales approximately $22.3\%$ larger
than under static shrinkage, accommodating sustained coefficient
movement. Once the process becomes stable, the ordering reverses: the
dynamic specification contracts the innovation scale to approximately
$26.1\%$ below that of the static specification, applying stronger
shrinkage to subsequent innovations.

This asymmetry is the mechanism identified by Proposition~I.1: since
the tail and origin behavior of the marginal shrinkage profile do not
depend on $\phi$, any performance difference between the dynamic
specification and the static restriction $\phi=0$ must arise from the
pooling of the shrinkage decision across time, which releases shrinkage
where movement persists and strengthens it where stability persists.

\subsection{Justification of the prior components}
\label{supp:why}

The prior specification can be motivated from the structure of the
estimation problem, in addition to the references in Section~2.1.2 of
the main paper. First, the coefficient innovations are approximately
sparse rather than exactly sparse. A firm's market exposure changes
through discrete events, such as acquisitions, debt issuance, or changes
in business mix, and does not evolve continuously at the daily
frequency. Between such events the exposure does not remain exactly
constant, since leverage, liquidity, and index membership drift over
time. The appropriate prior belief about the innovations is therefore
that most are negligible, a small number are large, and none are exactly
zero. This belief corresponds to a heavy-tailed normal scale mixture
with substantial mass near zero, that is, to a global--local shrinkage
prior. An exact-sparsity prior would impose literal constancy over
stretches of time, in addition to the computational considerations
discussed in Section~2.1.4 of the main paper.

Second, Section~2.1.2 of the main paper motivates dynamic innovation
variances by the observation that parameter instability is unlikely to
be constant over time; the temporal dependence of the shrinkage process
is additionally a statistical-power argument. A single observation
provides little information about whether an innovation reflects signal
or noise, and an independent horseshoe prior would treat all time points
as exchangeable. Autocorrelating $\log\lambda_{\beta_a,t}^{2}$ borrows
information from neighboring periods, and since the events that change
market exposures cluster in time, this borrowing is well founded.

Third, the AR(1) specification with $Z$-distributed innovations serves
both computational and theoretical purposes. The $Z$ distribution admits
a normal mean--scale mixture representation, which permits the
P\'olya--Gamma augmentation of \cite{kowal2019dynamic} and keeps every step
of the Gibbs sampler in closed form at $O(T)$ cost per asset. The
stationary linear recursion for the log-variance process, together with
the tail bounds obtained from the moment generating function
$\sec(\pi t)$, underlies the latent-state localization and the sieve
construction in the proof of Theorem~5. A change-point or
spike-and-slab specification for the innovations would provide neither
property.

Fourth, the horseshoe-type shape of the shrinkage addresses the two main
shortcomings of existing methods. The polynomial tails leave large
innovations essentially unshrunk, so a change in exposure at the onset
of a crisis is transmitted without attenuation. The unbounded density at
the origin shrinks the innovations almost completely during stable
periods, which produces the narrow credible intervals in Tables~1--5,
and the near-nominal empirical coverage indicates that this reduction in
interval width is not obtained at the cost of bias.

Finally, the remaining fixed hyperparameters play limited roles. The
normal priors on the log-volatility means are effectively flat at the
scale of the data, the Gamma priors on the volatility-of-volatility
parameters are standard weakly informative choices, and the scale of
$1/\sqrt{T}$ on the global parameter is robust because $\tau_{a,0}$ is
informed by $2T$ innovations under a heavy-tailed prior.

\section{Proof of Posterior Contraction result}
This section contains the complete proof of Theorem~5; Section~4 of the
main paper states the intermediate results and outlines the arguments.
The proof verifies the conditions of the posterior contraction theorem
for non-i.i.d. observations of \cite{ghosal2007convergence},
applied conditionally on the latent volatility path $\boldsymbol{s}_n$ and uniformly
over the sieve $H_n$; Proposition~4 then removes the conditioning.
Section~J.1 collects the definitions and notation, including the metric
$d_{n,\boldsymbol{s}_n}$, the sieves $\Theta_n$ and $H_n$, and the neighborhoods
$B_n(\theta_0,\epsilon_n;k)$. Sections~J.2 and~J.3 prove the metric
entropy bound of Proposition~1 and construct the consistent tests of
Proposition~2 from the resulting covers. Section~J.4 proves the prior
mass condition of Proposition~3;the strict positivity of the
stationary density of the dynamic shrinkage process, derived there, is
also used in Section~I.2. Section~J.5 proves Proposition~4, and
Section~J.6 assembles the proof of Theorem~5.

\subsection{Definitions and notation used in theoretical results}
Define $\bar{n}:= \max\{ n, 16\}$.Using the notions in Section 4, define the following expanding parameter sieve :\\
$\Theta_{n} = \{\boldsymbol{\theta}: |\mu_{M}|\leq \log(\bar{n}),\, |\phi_{M}| \leq \sqrt{1-\frac{1}{\log(\log(\bar{n}))}},\sigma_{M}^{2} \leq \log(\bar{n})^{0.5},||\Lambda||_{F} \leq \log(\bar{n}),\\ \max\limits_{1\leq k \leq r} |\tilde{\psi}_{k}| \leq \sqrt{1-\frac{1}{\log(\log(\bar{n}))}},  \max\limits_{1 \leq k \leq r}\tilde{\sigma}_{k}^{2} \leq \log(\bar{n})^{0.5},\max\limits_{1\leq i \leq N} |\bar{\mu}_{i}| \leq \log(\bar{n}),\\
\max\limits_{1 \leq i \leq N}|\bar{\psi}_{i}| \leq \sqrt{1-\frac{1}{\log(\log(\bar{n}))}}, \max\limits_{1 \leq i \leq N}\bar{\sigma}^{2}_{i} \leq \log(\bar{n})^{0.5}, \max\limits_{1\leq i \leq N}|\mu_{\alpha_{i}}| \leq \log(\bar{n}^{4}), \max\limits_{1 \leq i \leq N}|\mu_{\beta_{i}}| \leq \log(\bar{n}^{4}),\\
\max\limits_{1 \leq i \leq N} |\phi_{\beta_{i}}| \leq \sqrt{1-\frac{1}{\log(\log(\bar{n}))}}, \max\limits_{1 \leq i \leq N} |\phi_{\alpha_{i}}| \leq \sqrt{1-\frac{1}{\log(\log(\bar{n}))}} \}$. \vspace{.1cm}\\ 

We also define the following latent-state localization set : \\
$H_{n} = \{ \boldsymbol{s}_{n}: \max\limits_{1 \leq t \leq n}|h_{M,t-1}| \leq |\mu_{M}|+p\log(\bar{n}),\max\limits_{1\leq i \leq N,1\leq t\leq n-1}|h_{\alpha_{i,t-1}}| \leq |\mu_{\alpha_{i}}| + \frac{c\log(\bar{n}) + \log(2k(s))}{s},\\
\max\limits_{1\leq i\leq N, 1\leq t\leq n}|\beta_{i,t-1}|\leq \sqrt{(\kappa+1)C_{4}\bar{n}^{\alpha}\log(\bar{n})},\max\limits_{1\leq k \leq r, 1\leq t \leq n}|\tilde{h}_{k,t-1}| \leq p\log(\bar{n}),\\ \max\limits_{1 \leq i \leq N, 1\leq t \leq n}|\bar{h}_{i,t-1}| \leq |\bar{\mu}_{i}|+p\log(\bar{n})\}$ 
for constants $p,\alpha,s,C_{4},c,\kappa>0.$ \vspace{.4cm}

For $\boldsymbol{\theta}\in \Theta$ and $\boldsymbol{s} \in H $ define the conditional experiment, from (3), as follows:
\begin{equation*}
    Q_{\boldsymbol{\theta},\boldsymbol{s}_{n}}^{(n)} (\cdot) := \mathbb{P}_{\boldsymbol{\theta}}(\boldsymbol{r}_{1:n}\in \cdot|\boldsymbol{s}_{n}, \boldsymbol{r}_{M, 1:n}), \boldsymbol{r}_{1:n} = (\boldsymbol{r}_{1}, \dots, \boldsymbol{r}_{n}) \text{ and } \boldsymbol{r}_{M,1:n} = (r_{M,1},\dots, r_{M,n}) . 
\end{equation*}
Now note that $\boldsymbol{s}_{n}$ denotes the latent states up to time $n$. 
By conditional independence 
\begin{equation*}
    Q_{\boldsymbol{\theta}, \boldsymbol{s}_{n}}^{(n)} = \bigotimes_{t=1}^{n} Q_{\boldsymbol{\theta}, \boldsymbol{s}_{n,t}}, q_{\boldsymbol{\theta},\boldsymbol{s}_{n}}^{(n)}(\boldsymbol{x}^{(n)}) = \prod\limits_{t=1}^{n} q_{\boldsymbol{\theta}, \boldsymbol{s}_{n,t}}(\boldsymbol{x}_{t}), \text{ and } q_{\boldsymbol{\theta},\boldsymbol{s}_{n,t}} (\cdot) = \phi_{N}(\cdot; \mu_{t}(\boldsymbol{\theta}), \Sigma_{t}(\boldsymbol{\theta}))
\end{equation*}
is the Gaussian probability density function with mean vector $\mu_{t}(\boldsymbol{\theta})$ and covariance matrix $\Sigma_{t}(\boldsymbol{\theta})$.
Finally, define 
\begin{equation*}
    d_{n,\boldsymbol{s}_{n}}^{2}(\boldsymbol{\theta}, \boldsymbol{\theta}') := \frac{1}{n}\sum\limits_{t=1}^{n}h_{t}^{2}(\boldsymbol{\theta}, \boldsymbol{\theta}'), h_{t}^{2}(\boldsymbol{\theta}, \boldsymbol{\theta}') = \int(\sqrt{q_{\boldsymbol{\theta}, \boldsymbol{s}_{n,t}}}-\sqrt{q_{\boldsymbol{\theta}', \boldsymbol{s}_{n,t}}})^{2} d\lambda_{N}(\boldsymbol{x}), \boldsymbol{x}\in \mathbb{R}^{N}. 
\end{equation*}
\subsection{Proof of Proposition 1}
\begin{proof}
First we establish Lipschitz-type bounds on the mean vector and covariance matrix on $\Theta_{n}$. 

Consider a fixed $\boldsymbol{s}_{n} \in H_{n}$. Recall for fixed $a \in \{ 1, \dots, N \}$, $\mu_{a,t} = \alpha_{a,t} + r_{M,t}\beta_{a,t}$. Then $\partial_{\boldsymbol{\theta}}\mu(\boldsymbol{\theta}) =0$ since for a fixed latent state vector, the mean does not depend on the parameters. 
Recall, 
\begin{align*}
cov({\bf{r}}_{t})_{ij}|\boldsymbol{\theta},\boldsymbol{s}_{n}   &= \sum\limits_{k=1}^{r} \Lambda_{ik}\Lambda_{jk}\exp\{   \tilde{\psi}_{k}\tilde{h}_{k,t-1} + \frac{1}{2}\tilde{\sigma}_{k}^{2} \} \\ &+ \delta_{ij}\exp\{ \bar{\mu}_{i} + \bar{\psi}_{i}(\bar{h}_{i,t-1}-\bar{\mu}_{i}) + \frac{1}{2}\bar{\sigma}_{i}^{2} \}.
    \end{align*}
Now we compute the relevant partial derivatives. 
For $\Lambda$, suppressing the conditioning set $(\boldsymbol{\theta},\boldsymbol{s}_{n})$ for clarity, we have:
\begin{align*}
    |\frac{\partial cov({\bf{r}}_{t})_{ij}}{\partial \Lambda_{ik}}| &= \exp\{\tilde{\psi}_{k}\tilde{h}_{k,t-1} + \frac{1}{2}\tilde{\sigma}_{k}^{2}\}|\Lambda_{jk}|\\& = \exp\{ \tilde{\psi}_{k}\tilde{h}_{k,t-1} \} \exp \{ \frac{1}{2}\tilde{\sigma}_{k}^{2}\}|\Lambda_{jk}| \\
    &\leq \exp \{ |\tilde{\psi}_{k}||\tilde{h}_{k,t-1}|\}\exp\{\frac{1}{2}\tilde{\sigma}_{k}^{2} \}|\Lambda_{jk}|\\
    &\leq  \exp\{p\log(\bar{n})\} \exp \{ \frac{1}{2}\log(\bar{n})^{0.5} \} \log(\bar{n})\\
    &\leq \bar{n}^{p} \bar{n}\log(\bar{n}) = \bar{n}^{1+p} \log(\bar{n}), \text{ since } \bar{n} \geq 16,
\end{align*}
with an additional factor of 2 if $i = j. $
In addition, 
\begin{align*}
    |\frac{\partial cov({\bf{r}}_{t})_{ij}}{\partial \tilde{\psi}_{k}}| & = |\Lambda_{ik}\Lambda_{jk}| \exp \{ \frac{1}{2}\tilde{\sigma}_{k}^{2} \}|\tilde{h}_{k,t-1}|\exp\{ \tilde{\psi}_{k}\tilde{h}_{k,t-1} \} \\ & \leq  (\log(\bar{n}))^{2}\exp \{  \frac{1}{2}\log(\bar{n})^{0.5} \}p\log(\bar{n}) \exp \{p\log(\bar{n}) \} \\ &\leq\log(\bar{n})^{2}\bar{n}p\log(\bar{n})\bar{n}^{p}=\log(\bar{n})^{3}p\bar{n}^{1+p}, \text{ since } \bar{n} \geq 16.
\end{align*}
Furthermore, 
\begin{align*}
    |\frac{\partial cov({\bf{r}}_{t})_{ij}}{\partial \tilde{\sigma}_{k}^{2}}| &=  |\Lambda_{ik}\Lambda_{jk}|\exp\{ \tilde{\psi}_{k}\tilde{h}_{k,t-1} \}\frac{1}{2}\exp\{ \frac{1}{2}\tilde{\sigma}_{k}^{2} \} \\ &\leq \frac{1}{2}\log(\bar{n})^{2}\exp \{ p\log(\bar{n})  \} \exp\{ \frac{1}{2}\log(\bar{n})^{0.5}\}   \\ &  \leq \frac{1}{2}\log(\bar{n})^{2}\bar{n}^{p}\bar{n}=\frac{1}{2}\log(\bar{n})^{2} \bar{n}^{1+p}, \text{ since } \bar{n} \geq 16 .
\end{align*}
Antepenultimately, 
\begin{align*}
    |\frac{\partial cov({\bf{r}}_{t})_{ij}}{\partial \bar{\mu}_{i}}| &= \delta_{ij}\exp\{ \bar{\psi}_{i}\bar{h}_{i,t-1} + \frac{1}{2}\bar{\sigma}_{i}^{2} \}| (1-\bar{\psi}_{i})|\exp\{ \bar{\mu}_{i}(1-\bar{\psi}_{i}) \}\\ &  \leq \exp \{ (\log(\bar{n}) + p\log(\bar{n})) + \frac{1}{2}\log(\bar{n})^{0.5}\} \\ &\times (1+1)\exp \{ \log(\bar{n})(1+1)\} \\ &  \leq 2\exp\{\log(\bar{n}^{1+p}) \}\exp\{ \frac{1}{2}\log(\bar{n})^{0.5} \}\exp \{2\log(\bar{n}) \},\\ & \leq 2\bar{n}^{1+p}\bar{n}\bar{n}^{2} = 2\bar{n}^{4+p}, \text{ since } \bar{n} \geq 16. 
\end{align*}
Penultimately, 
\begin{align*}
    |\frac{\partial cov({\bf{r}}_{t})_{ij}}{\partial \bar{\psi}_{i}}| &= \delta_{ij}\exp\{ \bar{\mu}_{i}+\frac{1}{2}\bar{\sigma}_{i}^{2} \}|(\bar{h}_{i,t-1}-\bar{\mu}_{i})|\exp\{ \bar{\psi}_{i}(\bar{h}_{i,t-1}-\bar{\mu}_{i}) \} \\ & \leq \exp \{ \log(\bar{n}) + \frac{1}{2}\log(\bar{n})^{0.5}\}(p\log(\bar{n}) + 2\log(\bar{n}))\\ & \times\exp \{ (p\log(\bar{n})+2\log(\bar{n}))\} \\ & \leq \exp \{ \log(\bar{n}) \} \exp\{ \frac{1}{2}\log(\bar{n})^{0.5}\}(2+p)\log(\bar{n})\exp \{ \log(\bar{n}^{p+2}) \} \\ & \leq (2+p)\log(\bar{n})\bar{n}\bar{n}\bar{n}^{p+2} = (2+p)\log(\bar{n})\bar{n}^{4+p}, \text{ since } \bar{n} \geq 16.
\end{align*}
Finally, 
\begin{align*}
   |\frac{\partial cov({\bf{r}}_{t})_{ij}}{\partial \bar{\sigma}_{i}^{2}}| & = \delta_{ij}\exp\{ \bar{\mu}_{i} +\bar{\psi}_{i}(\bar{h}_{i,t-1}-\bar{\mu}_{i})  \}\frac{1}{2}\exp\{\frac{1}{2}\bar{\sigma}_{i}^{2} \} \\ & \leq \exp \{ \log(\bar{n}) + (p\log(\bar{n})+2\log(\bar{n}))\} \frac{1}{2}\exp \{ \frac{1}{2}\log(\bar{n})^{0.5} \} \\ &\leq \frac{1}{2}\exp \{ \log(\bar{n}^{3+p}) \}\exp \{\frac{1}{2}\log(\bar{n})^{\frac{1}{2}} \} \\ & \leq \frac{1}{2}\bar{n}^{3+p}\bar{n} = \frac{1}{2}\bar{n}^{4+p}, \text{ since } \bar{n} \geq 16.
\end{align*} 
Then, 
\begin{align*}
    |\Sigma_{t,i,j}(\boldsymbol{\theta}) - \Sigma_{t,i,j}(\boldsymbol{\theta}')| & = |\nabla_{\boldsymbol{\theta}}\Sigma_{t,i,j}(\eta_{i,j})^\top (\boldsymbol{\theta}- \boldsymbol{\theta}')|, \eta_{i,j} = \boldsymbol{\theta}' + s_{ij}(\boldsymbol{\theta} - \boldsymbol{\theta}'), s_{ij} \in (0,1)\\ 
    &\leq ||\nabla_{\boldsymbol{\theta}}\Sigma_{t,i,j}(\eta_{i,j})||_{2}||\boldsymbol{\theta}-\boldsymbol{\theta}'||_{2}\\
    &\leq L_{n}||\boldsymbol{\theta} - \boldsymbol{\theta}'||_{2},L_{n} = c_{1}\bar{n}^{c_{2}}\log(\bar{n})^{c_{3}}, c_{1},c_{2}, c_{3}>0 , 
\end{align*}
since the restrictions defining $\Theta_{n}$ are convex, $\Theta_{n}$ is convex; hence $\eta_{ij} \in \Theta_{n}$. Combining these terms to calculate the Frobenius norm, we have
\begin{align*}
    ||\Sigma_{t}(\boldsymbol{\theta}) - \Sigma_{t}(\boldsymbol{\theta}')||_{F}^{2} &= \sum\limits_{i,j= 1}^{N} |\Sigma_{t,i,j}(\boldsymbol{\theta}) - \Sigma_{t,i,j}(\boldsymbol{\theta}')|^{2}\\ 
    &\leq \sum\limits_{i,j=1}^{N} L_{n}^{2}||\boldsymbol{\theta}- \boldsymbol{\theta}'||_{2}^{2}\\ 
    &= N^{2} L_{n}^{2}||\boldsymbol{\theta}- \boldsymbol{\theta}'||_{2}^{2}. 
\end{align*}
That is,  
\begin{equation*}
    ||\Sigma_{t}(\boldsymbol{\theta}) - \Sigma_{t}(\boldsymbol{\theta}')||_{F} \leq NL_{n}||\boldsymbol{\theta}-\boldsymbol{\theta}'||_{2}. 
\end{equation*}
Recall that $h^{2}(P,Q) \leq 2KL (P,Q)$ and 
\begin{equation*}
    KL(N(\boldsymbol{\mu}_{1}, \Sigma_{1}), N(\boldsymbol{\mu}_{2}, \Sigma_{2})) = \frac{1}{2}[\text{tr}(\Sigma_{2}^{-1}\Sigma_{1}-I_{N}) - \log(\det(\Sigma_{2}^{-1}\Sigma_{1})) + (\boldsymbol{\mu}_{1} - \boldsymbol{\mu}_{2})^\top\Sigma_{2}^{-1}(\boldsymbol{\mu}_{1}-\boldsymbol{\mu}_{2})]. 
\end{equation*}
Note that since $\boldsymbol{s}_{n}$ is fixed $\boldsymbol{\mu}_{t}(\boldsymbol{\theta}_{1}) - \boldsymbol{\mu}_{t}(\boldsymbol{\theta}_{2}) = \boldsymbol{0}_{N}. $
Define 
\begin{equation*}
    X := \Sigma_{t}(\boldsymbol{\theta}_{2})^{-\frac{1}{2}} \Sigma_{t}(\boldsymbol{\theta}_{1}) \Sigma_{t}(\boldsymbol{\theta}_{2})^{-\frac{1}{2}}.
\end{equation*}
Then 
\begin{equation*}
    KL_{t}(\boldsymbol{\theta}_{1}, \boldsymbol{\theta}_{2})  = \frac{1}{2}[ tr(X) - \log(\det(X)) - N]. 
\end{equation*}
Since $X$ is symmetric, it is diagonalizable, $X = SDS^{-1}$ where $S$ is some invertible matrix and $D$ is a diagonal matrix consisting of the eigenvalues of $X$. Then it follows that
\begin{equation*}
    \text{tr}(X) =\text{tr}(SDS^{-1}) = tr(DSS^{-1}) = tr(D) = \sum\limits_{i=1}^{N} \lambda_{i}. 
\end{equation*}
Similarly, 
\begin{equation*}
    \log(\det(X)) = \log(\det(SDS^{-1})) = \log(\det(D)) = \log(\prod\limits_{i=1}^{N} \lambda_{i}) = \sum\limits_{i=1}^{N}\log(\lambda_{i}). 
\end{equation*}
Then the difference of these terms in the KL distance is
\begin{align*}
    \text{tr}(X) - \log(\det(X)) -N & = \sum\limits_{i = 1 }^{N} \lambda_{i} - \sum\limits_{i=1}^{N} \log(\lambda_{i}) - \sum\limits_{i=1}^{N}1 \\ 
    &= \sum\limits_{i=1}^{N}(\lambda_{i} - \log(\lambda_{i}) -1)\\ 
    &= \sum\limits_{i=1}^{N}f(\lambda_{i}), f(x) = x-\log(x) -1. 
\end{align*}
Since $ \Lambda \tilde{\Sigma}_{t}\Lambda^\top + \bar{\Sigma}_{t} \succeq 0$, $\Sigma_{t}(\boldsymbol{\theta}) \succeq \bar{\Sigma}_{t}$. Then,  
\begin{equation*}
    \lambda_{\min}(\Sigma_{t}(\boldsymbol{\theta})) \geq \lambda_{\min}(\bar{\Sigma}_{t}) = \min\limits_{1\leq i \leq N}( \exp \{ \bar{\mu}_{i}+ \bar{\psi}_{i}(\bar{h}_{i,t-1} -\bar{\mu}_{i}) + \frac{1}{2}\bar{\sigma}_{i}^{2} \}). 
\end{equation*}
Note that 
\begin{align*}
    \exp\{ \bar{\mu}_{i} + \bar{\psi}_{i}(\bar{h}_{i,t-1}-\bar{\mu}_{i}) + \frac{1}{2}\bar{\sigma}_{i}^{2} \}  &\geq \exp\{ - \log(\bar{n}) - \log(\bar{n}) -p\log(\bar{n}) - \log(\bar{n}) \}\\
    &= \exp \{ - (3+p)\log(\bar{n}) \}\\
    &= \bar{n}^{-(3+p)}. 
\end{align*}
Thus
\begin{equation*}
    \lambda_{\min}(\Sigma_{t}(\boldsymbol{\theta})) \geq \bar{n}^{-(3+p)}.
\end{equation*}
Now we need to upper bound. 
\begin{equation*}
    \lambda_{\max} (\Sigma_{t}(\boldsymbol{\theta})) \leq  \lambda_{\max}(\bar{\Sigma}_{t}) + \lambda_{\max}(\Lambda \tilde{\Sigma}_{t}\Lambda^\top). 
\end{equation*}
Note 
\begin{align*}
    \lambda_{\max}(\bar{\Sigma}_{t}) &\leq \exp\{ \log(\bar{n}) + (2+p)\log(\bar{n}) + \frac{1}{2}\log(\bar{n})^{\frac{1}{2}} \}\\
    & = \exp \{ (3+p) \log(\bar{n}) + \frac{1}{2}\log(\bar{n})^{\frac{1}{2}} \}\\
    &= \bar{n}^{3+p} e^{\frac{1}{2}(\log(\bar{n}))^{\frac{1}{2}}}\\
    &\leq \bar{n}^{4 + p}. 
\end{align*}
Similarly, 
\begin{align*}
    \lambda_{\max} (\Lambda \tilde{\Sigma}_{t}\Lambda^\top) &\leq ||\Lambda||_{F}^{2} \lambda_{\max}(\tilde{\Sigma}_{t})\\
    &\leq (\log(\bar{n}))^{2}\exp\{  p\log(\bar{n}) + \frac{1}{2}\log(\bar{n})^{\frac{1}{2}} \}\\
    &\leq \log(\bar{n})^{2}\bar{n}^{1+p}. 
\end{align*}
So that, 
\begin{equation*}
    \lambda_{\max}(\Sigma_{t}(\boldsymbol{\theta})) \leq \bar{n}^{4+p} + \log(\bar{n})^{2}\bar{n}^{1+p}. 
\end{equation*}
That is 
\begin{equation*}
    \lambda_{\min}(\Sigma_{t}(\boldsymbol{\theta}))\geq m_{n} \text{ and } \lambda_{\max}(\Sigma_{t}(\boldsymbol{\theta})) \leq M_{n}, 
\end{equation*}
where 
\begin{align*}
    m_{n}&:=  \bar{n}^{-(3+p)}\\
    M_{n}&:= \bar{n}^{4+p}  + (\log(\bar{n}))^{2}\bar{n}^{1+p}. 
\end{align*}
Then 
\begin{equation*}
    m_{n}I_{N}\preceq \Sigma_{t}(\boldsymbol{\theta}_{1}) \preceq M_{n}I_{N}  \text{ and } m_{n}I_{N} \preceq \Sigma_{t}(\boldsymbol{\theta}_{2}) \preceq M_{n}I_{N}. 
\end{equation*}
Then, 
\begin{equation*}
    X = \Sigma_{t}(\boldsymbol{\theta}_{2})^{-\frac{1}{2}}\Sigma_{t}(\boldsymbol{\theta}_{1})\Sigma_{t}(\boldsymbol{\theta}_{2})^{-\frac{1}{2}} \succeq \Sigma_{t}(\boldsymbol{\theta}_{2})^{-\frac{1}{2}} m_{n}I_{N} \Sigma_{t}(\boldsymbol{\theta}_{2})^{-\frac{1}{2}} = m_{n}\Sigma_{t}(\boldsymbol{\theta}_{2})^{-1}. 
\end{equation*}
We know 
\begin{equation*}
    \Sigma_{t}(\boldsymbol{\theta}_{2}) \preceq M_{n}I_{N} \iff \Sigma_{t}^{-1}(\boldsymbol{\theta}_{2}) \succeq \frac{1}{M_{n}}I_{N}. 
\end{equation*}
Therefore, we have
\begin{equation*}
    X \succeq m_{n}\Sigma_{t}^{-1}(\boldsymbol{\theta}_{2}) \succeq \frac{m_{n}}{M_{n}}I_{N}. 
\end{equation*}
Similarly
\begin{equation*}
    X = \Sigma_{t}(\boldsymbol{\theta}_{2})^{-\frac{1}{2}} \Sigma_{t}(\boldsymbol{\theta}_{1})\Sigma_{t}(\boldsymbol{\theta}_{2})^{-\frac{1}{2}} \preceq \Sigma_{t}(\boldsymbol{\theta}_{2})^{-\frac{1}{2}} (M_{n}I_{N})\Sigma_{t}(\boldsymbol{\theta}_{2})^{-\frac{1}{2}} = M_{n}\Sigma_{t}^{-1}(\boldsymbol{\theta}_{2}). 
\end{equation*}
with 
\begin{equation*}
    \Sigma_{t}(\boldsymbol{\theta}_{2}) \succeq m_{n}I_{N} \iff \Sigma_{t}^{-1}(\boldsymbol{\theta}_{2}) \preceq \frac{1}{m_{n}}I_{N}. 
\end{equation*}
This implies 
\begin{equation*}
    X \preceq M_{n}\Sigma_{t}(\boldsymbol{\theta}_{2})^{-1} \preceq \frac{M_{n}}{m_{n}}I_{N}. 
\end{equation*}
Hence
\begin{equation*}
    \frac{m_{n}}{M_{n}}I_{N} \preceq \Sigma_{t}^{-\frac{1}{2}}(\boldsymbol{\theta}_{2})\Sigma_{t}(\boldsymbol{\theta}_{1}) \Sigma_{t}(\boldsymbol{\theta}_{2})^{-\frac{1}{2}} \preceq \frac{M_{n}}{m_{n}} I_{N}. 
\end{equation*}
Thus 
\begin{equation*}
    \frac{m_{n}}{M_{n}}\leq \lambda_{i}(X) \leq \frac{M_{n}}{m_{n}}, i= 1, \dots, N. 
\end{equation*}
By Taylor's theorem, $\exists c_{i}$ between $1$ and $\lambda_{i}(X)$ such that 
\begin{align*}
    f(\lambda_{i}(X)) &= f(1) + f'(1) (\lambda_{i}(X)-1) + \frac{1}{2}f''(c_{i})(\lambda_{i}(X)-1)^{2}\\
    &= \frac{1}{2}f''(c_{i})(\lambda_{i}(X)-1)^{2}\\
    &= \frac{1}{2c_{i}^{2}} (\lambda_{i}(X)-1)^{2}. 
\end{align*}
Observe, 
\begin{equation*}
    \frac{1}{c_{i}^{2}} \leq \sup\limits_{u \in [\frac{m_{n}}{M_{n}},\frac{M_{n}}{m_{n}}]} \frac{1}{u^{2}} = (\frac{M_{n}}{m_{n}})^{2},
\end{equation*}
so that
\begin{equation*}
    |f(\lambda_{i}(X))| \leq \frac{1}{2}(\frac{M_{n}}{m_{n}})^{2}(\lambda_{i}(X)-1)^{2}. 
\end{equation*}
Then, 
\begin{equation*}
    \text{tr}(X) - \log(\det(X))-N \leq \frac{1}{2}(\frac{M_{n}}{m_{n}})^{2}\sum\limits_{i=1}^{N}(\lambda_{i}(X)-1)^{2}.
\end{equation*}
Since $X$ is symmetric, we have for orthogonal $Q$
\begin{equation*}
    X-I_{N} = Q\text{diag}(\lambda_{1}(X)-1, \dots, \lambda_{N}(X)-1)Q^\top. 
\end{equation*}
Define 
\begin{equation*}
    D_{0}:= \text{diag}(\lambda_{1}(X)-1,\dots, \lambda_{N}(X)-1).
\end{equation*}
Then, 
\begin{align*}
    ||X-I_{N}||_{F}^{2} &= ||QD_{0}Q^\top||_{F}^{2}\\
    &= \text{tr}(QD_{0}Q^\top QD_{0}Q^\top)\\
    &= \text{tr}(QD_{0}^{2}Q^\top)\\
    &= \text{tr}(D_{0}^{2} Q^\top Q)\\
    &= \text{tr}(D_{0}^{2})\\
    &= \sum\limits_{i=1}^{N} (\lambda_{i}(X)-1)^{2}. 
\end{align*}
So that, 
\begin{equation*}
    \text{tr}(X) - \log(\det(X)) - N \leq \frac{1}{2}(\frac{M_{n}}{m_{n}})^{2}||X-I_{N}||_{F}^{2}. 
\end{equation*}
Then, 
\begin{align*}
    ||X-I_{N}||_{F} & = ||\Sigma_{t}(\boldsymbol{\theta}_{2})^{-\frac{1}{2}}\Sigma_{t}(\boldsymbol{\theta}_{1}) \Sigma_{t}(\boldsymbol{\theta}_{2})^{-\frac{1}{2}} -I_{N}||_{F}\\
    &= ||\Sigma_{t}(\boldsymbol{\theta}_{2})^{-\frac{1}{2}} (\Sigma_{t}(\boldsymbol{\theta}_{1})-\Sigma_{t}(\boldsymbol{\theta}_{2}))\Sigma_{t}(\boldsymbol{\theta}_{2})^{-\frac{1}{2}}||_{F}\\
    &\leq ||\Sigma_{t}(\boldsymbol{\theta}_{2})^{-\frac{1}{2}}||_{OP}^{2} ||\Sigma_{t}(\boldsymbol{\theta}_{1})-\Sigma_{t}(\boldsymbol{\theta}_{2})||_{F}\\
    &= ||\Sigma_{t}(\boldsymbol{\theta}_{2})^{-1}||_{OP} ||\Sigma_{t}(\boldsymbol{\theta}_{1})-\Sigma_{t}(\boldsymbol{\theta}_{2})||_{F}\\
    &= \frac{1}{\lambda_{\min}(\Sigma_{t}(\boldsymbol{\theta}_{2}))}||\Sigma_{t}(\boldsymbol{\theta}_{1})-\Sigma_{t}(\boldsymbol{\theta}_{2})||_{F}\\
    &\leq \frac{1}{m_{n}}||\Sigma_{t}(\boldsymbol{\theta}_{1})- \Sigma_{t}(\boldsymbol{\theta}_{2})||_{F}. 
\end{align*}
So 
\begin{align*}
    \text{tr}(X) - \log(\det(X)) - N &\leq \frac{1}{2}(\frac{M_{n}}{m_{n}})^{2}||\Sigma_{t}(\boldsymbol{\theta}_{1}) - \Sigma_{t}(\boldsymbol{\theta}_{2})||_{F}^{2}\frac{1}{m_{n}^{2}}\\
    & = C_{n}'||\Sigma_{t}(\boldsymbol{\theta}_{1})- \Sigma_{t}(\boldsymbol{\theta}_{2})||_{F}^{2}, C_{n}' = \frac{1}{2}(\frac{M_{n}}{m_{n}})^{2}\frac{1}{m_{n}^{2}}\\
    &\leq C_{n}' N^{2}L_{n}^{2}||\boldsymbol{\theta}_{1}-\boldsymbol{\theta}_{2}||_{2}^{2}, B_{n} = C_{n}'N^{2}L_{n}^{2}. 
\end{align*}
Then
\begin{equation*}
    KL_{t}(\boldsymbol{\theta}_{1}, \boldsymbol{\theta}_{2}) \leq \frac{1}{2}B_{n}||\boldsymbol{\theta}_{1}-\boldsymbol{\theta}_{2}||_{2}^{2}. 
\end{equation*}
So 
\begin{align*}
    {d}_{n,\boldsymbol{s}_{n}}^{2}(\boldsymbol{\theta}_{1}, \boldsymbol{\theta}_{2}) & = \frac{1}{n}\sum\limits_{t=1}^{n}h_{t}^{2}(\boldsymbol{\theta}_{1}, \boldsymbol{\theta}_{2})\\
    &\leq \frac{2}{n} \sum\limits_{t=1}^{n}KL_{t}(\boldsymbol{\theta}_{1}, \boldsymbol{\theta}_{2})\\
    & \leq \frac{2}{n}\sum\limits_{t=1}^{n}\frac{1}{2}(B_{n})||\boldsymbol{\theta}_{1}-\boldsymbol{\theta}_{2}||_{2}^{2}\\
    &= B_{n}||\boldsymbol{\theta}_{1}-\boldsymbol{\theta}_{2}||_{2}^{2}. 
\end{align*}
That is, 
\begin{align*}
   {d}_{n,\boldsymbol{s}_{n}}(\boldsymbol{\theta}_{1}, \boldsymbol{\theta}_{2}) &\leq \Gamma_{n}||\boldsymbol{\theta}_{1}-\boldsymbol{\theta}_{2}||_{2}, \Gamma_{n}:= \sqrt{B_{n}}\\
    & \leq c_{4}\bar{n}^{c_{5}} \log(\bar{n})^{c_{6}}||\boldsymbol{\theta}_{1}- \boldsymbol{\theta}_{2}||_{2}, c_{4}, c_{5}, c_{6}>0\\
    &= \tilde{L}_{n}||\boldsymbol{\theta}_{1}- \boldsymbol{\theta}_{2}||_{2}, \tilde{L}_{n} = c_{4}\bar{n}^{c_{5}} \log(\bar{n})^{c_{6}}. 
\end{align*}
Let $\epsilon > \epsilon_{n}$ and set $\delta:= \epsilon/36\tilde{L}_{n}$. If $||\boldsymbol{\theta}-\boldsymbol{\theta}'||_{2} <  \delta$, then by above 
\begin{equation*}
    \sup\limits_{\boldsymbol{s}_{n}\in H_{n}}d_{n,\boldsymbol{s}_{n}}(\boldsymbol{\theta}, \boldsymbol{\theta}') < \frac{\tilde{L}_{n}\epsilon}{36\tilde{L}_{n}} = \frac{\epsilon}{36}.
\end{equation*}
Therefore, for every $\boldsymbol{s}_{n} \in H_{n}$, 
\begin{equation*}
    N(\frac{\epsilon}{36}, \{ \boldsymbol{\theta} \in \Theta_{n}:d_{n,\boldsymbol{s}_{n}}(\boldsymbol{\theta}, \boldsymbol{\theta}_{0})\leq \epsilon \}, d_{n,\boldsymbol{s}_{n}}) \leq N(\delta, \Theta_{n}, ||\cdot||_{2}).
\end{equation*}
Now we just need to bound $N(\delta, \Theta_{n}, ||\cdot||_{2}).$
From the definition of $\Theta_{n}$,
\begin{equation*}
    \Theta_{n} \subset \prod\limits_{j=1}^{c_{\boldsymbol{\theta}}} [a_{j}, b_{j}],L_{j} := b_{j} - a_{j} \leq 8\log(\bar{n})
\end{equation*} 
Split each coordinate interval into pieces of length $h = \delta/ \sqrt{c_{\boldsymbol{\theta}}}$. Then the number of subintervals needed in coordinate $j$ is at most 
\begin{equation*}
    \lceil \frac{L_{j}}{h}\rceil \leq \frac{L_{j}}{h}+1 \leq \frac{8\log(\bar{n})}{\frac{\delta}{\sqrt{c_{\boldsymbol{\theta}}}}}+1 = \frac{8\sqrt{c_{\boldsymbol{\theta}}}\log(\bar{n})}{\delta}+1. 
\end{equation*}
So the whole rectangle is partitioned into at most 
\begin{equation*}
    \prod\limits_{j=1}^{c_{\boldsymbol{\theta}}}(\frac{8\sqrt{c_{\boldsymbol{\theta}}}\log(\bar{n})}{\delta} +1) = (\frac{8\sqrt{c_{\boldsymbol{\theta}}}\log(\bar{n})}{\delta} + 1)^{c_{\boldsymbol{\theta}}} \text{ small boxes}. 
\end{equation*}
A box with side length $h$ has Euclidean diameter at most $\sqrt{c_{\boldsymbol{\theta}}}h =  \sqrt{c_{\boldsymbol{\theta}}} \cdot \delta / \sqrt{c_{\boldsymbol{\theta}}} = \delta.$ So each small box can be covered by one $\delta$-ball in $||\cdot||_{2}$. Therefore 
\begin{align*}
    N(\delta, \Theta_{n}, ||\cdot||_{2})&\leq (\frac{8\sqrt{c_{\boldsymbol{\theta}}} \log(\bar{n})}{\delta} + 1)^{c_{\boldsymbol{\theta}}}\\
    &\leq (\frac{c_{7}\log(\bar{n})}{\delta})^{c_{\boldsymbol{\theta}}}, c_{7}>0. 
\end{align*}
Taking logs gives 
\begin{equation*}
    \log(N(\delta, \Theta_{n}, ||\cdot||_{2})) \leq c_{\boldsymbol{\theta}}[\log(c_{7}\log(\bar{n})) + \log(\frac{1}{\delta})]. 
\end{equation*}
Substituting $\delta = \epsilon/ 36 \tilde{L}_{n}$ gives 
\begin{equation*}
    \log(N(\delta, \Theta_{n}, ||\cdot||_{2})) \leq c_{\boldsymbol{\theta}}[\log(c_{7} \log(\bar{n})) + \log(36) + \log(\tilde{L}_{n}) + \log(\frac{1}{\epsilon})]. 
\end{equation*}
Note 
\begin{align*}
    \log(\tilde{L}_{n}) &= \log(c_{4}\bar{n}^{c_{5}}\log(\bar{n})^{c_{6}})\\
    &= \log(c_{4}) + c_{5}\log(\bar{n}) + c_{6}\log(\log(\bar{n})). 
\end{align*}
Similarly, since $\epsilon > \epsilon_{n}$,
\begin{align*}
    \log(\frac{1}{\epsilon}) &\leq \log(\frac{1}{\epsilon_{n}})\\
    &= \log(\frac{1}{c_{\epsilon}n^{-\frac{1}{2+c_{\boldsymbol{\theta}}}}})\\
    &= \log(1) - \log(c_{\epsilon} n^{-\frac{1}{2+c_{\boldsymbol{\theta}}}})\\
    &= -\log(c_{\epsilon}) - \log(n^{-\frac{1}{2+c_{\boldsymbol{\theta}}}})\\
    &= - \log(c_{\epsilon}) - (-\frac{1}{2+c_{\boldsymbol{\theta}}})\log(n)\\
    &= - \log(c_{\epsilon}) + \frac{1}{2+c_{\boldsymbol{\theta}}} \log(n). 
\end{align*}
Therefore
\begin{equation*}
    \log(N(\delta, \Theta_{n}, ||\cdot||_{2})) \leq c_{\boldsymbol{\theta}} [\log(c_{7}\log(\bar{n}))+ \log(36) + \log(c_{4}) + c_{5}\log(\bar{n})+ c_{6}\log(\log(\bar{n})) - \log(c_{\epsilon}) + \frac{1}{2+c_{\boldsymbol{\theta}}}\log(n)].
\end{equation*}
Since $\log(n) \leq \log(\bar{n})$ and $\log(c_{7}\log(\bar{n})) = \log(c_{7}) + \log(\log(\bar{n}))$, and choosing $c_{\epsilon} \geq 1$ we have that, 
\begin{equation*}
    \log(N(\delta,\Theta_{n}, ||\cdot||_{2})) \leq c_{\boldsymbol{\theta}}[\log(c_{7})+ \log(\log(\bar{n})) + \log(36)+ \log(c_{4}) + c_{5}\log(\bar{n}) + c_{6}\log(\log(\bar{n})) + \frac{1}{2+c_{\boldsymbol{\theta}}}\log(\bar{n})]. 
\end{equation*}
It follows that, for some finite positive constants $A,B,D$, 
\begin{equation*}
    \log(N(\delta,\Theta_{n}, ||\cdot||_{2})) \leq A\log(\bar{n}) + B\log(\log(\bar{n})) +D.
\end{equation*}
Therefore, 
\begin{equation*}
    \sup\limits_{\boldsymbol{s}_{n} \in H_{n}}\sup\limits_{\epsilon > \epsilon_{n}} \log(N(\frac{\epsilon}{36}, \{ \boldsymbol{\theta} \in \Theta_{n}:d_{n,\boldsymbol{s}_{n}}(\boldsymbol{\theta},\boldsymbol{\theta}_{0}) \leq \epsilon \}, d_{n,\boldsymbol{s}_{n}})) \leq A\log(\bar{n}) + B\log(\log(\bar{n})) + D.
\end{equation*}
Define 
\begin{equation*}
    q:= \frac{c_{\boldsymbol{\theta}}}{2+c_{\boldsymbol{\theta}}}>0.
\end{equation*}
Since $\bar{n} = \max \{ n,16 \}$, define 
\begin{equation*}
    C:= \sup\limits_{n \geq 1} \frac{A\log(\bar{n}) + B\log(\log(\bar{n})) + D}{n^{q}}.
\end{equation*}
For $1 \leq n \leq 16, \bar{n} = 16$ and hence 
\begin{equation*}
    \frac{A\log(\bar{n}) + B\log(\log(\bar{n})) + D}{n^{q}} = \frac{A\log(16) + B\log(\log(16)) + D}{n^{q}} \leq A\log(16) + B\log(\log(16)) + D. 
\end{equation*}
For $n\geq16,n = \bar{n}$, 
\begin{equation*}
    \frac{A\log(\bar{n}) + B\log(\log(\bar{n})) + D }{n^{q}} = \frac{A\log(n) + B\log(\log(n)) + D}{n^{q}}.
\end{equation*}
Since $q>0$, 
\begin{equation*}
    \lim\limits_{n \rightarrow \infty} \frac{\log(n)}{n^{q}}, \lim\limits_{n \rightarrow \infty} \frac{\log(\log(n))}{n^{q}} = 0, \text{ and } \lim\limits_{n \rightarrow \infty} \frac{1}{n^{q}} = 0. 
\end{equation*}
Thus, $C < \infty$. Choose $c_{\epsilon}$ sufficiently large such that $c_{\epsilon} \geq \max\{1,\sqrt{C} \}$. Then, 
\begin{equation*}
    A\log(\bar{n}) + B\log(\log(\bar{n})) + D \leq c_{\epsilon}^{2}n^{q}. 
\end{equation*}
Since 
\begin{equation*}
    \epsilon_{n} = c_{\epsilon}n^{-\frac{1}{2+c_{\boldsymbol{\theta}}}},
\end{equation*}
we have, 
\begin{align*}
    n\epsilon_{n}^{2}&= nc_{\epsilon}^{2}n^{-\frac{2}{2+c_{\boldsymbol{\theta}}}}\\
    &=c_{\epsilon}^{2}n^{1-\frac{2}{2+c_{\boldsymbol{\theta}}}}\\
    &=c_{\epsilon}^{2}n^{\frac{c_{\boldsymbol{\theta}}}{2+c_{\boldsymbol{\theta}}}}\\
    &= c_{\epsilon}^{2}n^{q}.
\end{align*}
Therefore, 
\begin{equation*}
    A\log(\bar{n}) + B\log(\log(\bar{n})) + D \leq n\epsilon_{n}^{2}. 
\end{equation*}
Hence, 
\begin{equation*}
    \sup\limits_{\boldsymbol{s}_{n} \in H_{n}} \sup\limits_{\epsilon>\epsilon_{n}} \log(N(\frac{\epsilon}{36},\{ \boldsymbol{\theta} \in \Theta_{n}: d_{n,\boldsymbol{s}_{n}}(\boldsymbol{\theta}, \boldsymbol{\theta}_{0}) \leq \epsilon \},d_{n,\boldsymbol{s}_{n}})) \leq n \epsilon_{n}^{2}. 
\end{equation*}
\end{proof}
\subsection{Proof of Proposition 2}
\begin{proof}
Fix $\boldsymbol{s}_{n} \in H_{n}$ and $j \geq 2$. Since 
    \begin{equation*}
        S_{n,j} \subset \{ \boldsymbol{\theta} \in \Theta_{n}:d_{n,\boldsymbol{s}_{n}}(\boldsymbol{\theta}, \boldsymbol{\theta}_{0}) \leq 2j\epsilon_{n} \},
    \end{equation*}
Proposition 1, applied with $\epsilon = 2j \epsilon_{n}$ gives 
\begin{equation*}
    \log(N(\frac{2j\epsilon_{n}}{36}, S_{n,j}, d_{n,\boldsymbol{s}_{n}})) \leq n\epsilon_{n}^{2}. 
\end{equation*}
That is 
\begin{equation*}
    \log(N(\frac{j \epsilon_{n}}{18}, S_{n,j},d_{n,\boldsymbol{s}_{n}})) \leq n\epsilon_{n}^{2}. 
\end{equation*}
Choose a $d_{n,\boldsymbol{s}_{n}}-$cover of $S_{n,j}$ by balls of radius $j\epsilon_{n}/18$ with centers, 
\begin{equation*}
    \boldsymbol{\theta}_{j,1}, \dots, \boldsymbol{\theta}_{j,M_{n,j}} \in S_{n,j}, M_{n,j}:= N(\frac{j \epsilon_{n}}{18}, S_{n,j}, d_{n,\boldsymbol{s}_{n}}). 
\end{equation*}
Then 
\begin{equation*}
    S_{n,j} \subset \bigcup\limits_{\ell = 1}^{M_{n,j}}B_{d_{n,\boldsymbol{s}_{n}}}(\boldsymbol{\theta}_{j, \ell}, \frac{j\epsilon_{n}}{18}), \log(M_{n,j}) \leq n\epsilon_{n}^{2}. 
\end{equation*}
Now fix one center $\boldsymbol{\theta}_{j,\ell}$. Since $\boldsymbol{\theta}_{j,\ell} \in S_{n,j}$, 
\begin{equation*}
    d_{n,\boldsymbol{s}_{n}}(\boldsymbol{\theta}_{0}, \boldsymbol{\theta}_{j, \ell}) > j \epsilon_{n}. 
\end{equation*}
Because $Q_{\boldsymbol{\theta}, \boldsymbol{s}_{n}}^{(n)}$ is a product measure and $d_{n,\boldsymbol{s}_{n}}$ is the average Hellinger distance, then Lemma 2 of \cite{ghosal2007convergence} applies. Hence there exists a test $\phi_{n,j,\ell, \boldsymbol{s}_{n}}$ such that 
\begin{equation*}
    Q_{\boldsymbol{\theta}_{0}, \boldsymbol{s}_{n}}^{(n)} \phi_{n,j,\ell, \boldsymbol{s}_{n}} \leq \exp \{ - \frac{1}{2}nd_{n,\boldsymbol{s}_{n}}^{2} (\boldsymbol{\theta}_{0}, \boldsymbol{\theta}_{j,\ell}) \} \leq e^{-\frac{1}{2}j^{2} n\epsilon_{n}^{2}}
\end{equation*}
and 
\begin{equation*}
Q_{\boldsymbol{\theta},\boldsymbol{s}_{n}}^{(n)}(1- \phi_{n,j, \ell, \boldsymbol{s}_{n}}) \leq \exp \{ - \frac{1}{2} nd_{n,\boldsymbol{s}_{n}}^{2} (\boldsymbol{\theta}_{0}, \boldsymbol{\theta}_{j,\ell}) \} \leq e^{-\frac{1}{2}j^{2} n\epsilon_{n}^{2}},
\end{equation*}
for every $\boldsymbol{\theta}$ satisfying 
\begin{equation*}
    d_{n,\boldsymbol{s}_{n}}(\boldsymbol{\theta}, \boldsymbol{\theta}_{j, \ell}) \leq \frac{1}{18} d_{n,\boldsymbol{s}_{n}} (\boldsymbol{\theta}_{0}, \boldsymbol{\theta}_{j,\ell}).
\end{equation*}
Therefore 
\begin{equation*}
    \sup\limits_{\boldsymbol{\theta} \in B_{d_{n,\boldsymbol{s}_{n}}}(\boldsymbol{\theta}_{j, \ell}, \frac{j\epsilon_{n}}{18})} Q_{\boldsymbol{\theta},\boldsymbol{s}_{n}}^{(n)}(1-\phi_{n,j,\ell,\boldsymbol{s}_{n}}) \leq \exp \{- \frac{1}{2}j^{2} n\epsilon_{n}^{2} \}.
\end{equation*}
Define 
\begin{equation*}
    \phi_{n,j,\boldsymbol{s}_{n}}:= \max\limits_{1 \leq \ell \leq M_{n,j}} \phi_{n,j,\ell, \boldsymbol{s}_{n}}.
\end{equation*}
Since 
\begin{equation*}
    \max\limits_{1 \leq \ell \leq M_{n,j}} \phi_{n,j,\ell,\boldsymbol{s}_{n}} \leq \sum\limits_{\ell=1}^{M_{n,j}} \phi_{n,j,\ell,\boldsymbol{s}_{n}}. 
\end{equation*}
It follows that 
\begin{align*}
    Q_{\boldsymbol{\theta}_{0},\boldsymbol{s}_{n}}^{(n)}\phi_{n,j,\boldsymbol{s}_{n}} &\leq Q_{\boldsymbol{\theta}_{0},\boldsymbol{s}_{n}}^{(n)}(\sum\limits_{\ell = 1}^{M_{n,j}} \phi_{n,j,\ell,\boldsymbol{s}_{n}})\\
    &= \sum\limits_{\ell=1}^{M_{n,j}} Q_{\boldsymbol{\theta}_{0}, \boldsymbol{s}_{n}}^{(n)} \phi_{n,j,\ell,\boldsymbol{s}_{n}}\\
    &\leq M_{n,j} \exp\{ - \frac{1}{2}j^{2}n\epsilon_{n}^{2} \}\\
    &\leq \exp \{ n \epsilon_{n}^{2} \} \exp\{ - \frac{1}{2}j^{2}n\epsilon_{n}^{2} \}\\
    &= \exp \{ (1-\frac{j^{2}}{2})n\epsilon_{n}^{2} \}\\
    &= \exp\{ -(\frac{j^{2}}{2} -1)n\epsilon_{n}^{2} \}.
\end{align*}
Since $j \geq 2$, 
\begin{equation*}
    \frac{j^{2}}{2} -1 \geq \frac{j^{2}}{4}. 
\end{equation*}
Hence, 
\begin{equation*}
    Q_{\boldsymbol{\theta}_{0},\boldsymbol{s}_{n}}^{(n)}\phi_{n,j,\boldsymbol{s}_{n}} \leq \exp \{ - \frac{j^{2} n\epsilon_{n}^{2}}{4} \}. 
\end{equation*}
Again consider $\boldsymbol{\theta} \in S_{n,j}$. Since 
\begin{equation*}
    S_{n,j} \subset \bigcup\limits_{\ell = 1}^{M_{n,j}} B_{d_{n,\boldsymbol{s}_{n}}}(\boldsymbol{\theta}_{j,\ell}, \frac{j\epsilon_{n}}{18}),
\end{equation*}
$\exists \ell,$ called $\ell(\boldsymbol{\theta})$ such that
\begin{equation*}
    d_{n,\boldsymbol{s}_{n}}(\boldsymbol{\theta}, \boldsymbol{\theta}_{j, \ell(\boldsymbol{\theta})}) < \frac{j\epsilon_{n}}{18}.
\end{equation*}
Then 
\begin{equation*}
Q_{\boldsymbol{\theta},\boldsymbol{s}_{n}}^{(n)}(1-\phi_{n,j, \ell(\boldsymbol{\theta}), \boldsymbol{s}_{n}}) \leq e^{-\frac{1}{2}j^{2}n\epsilon_{n}^{2}}. 
\end{equation*}
Note that $\phi_{n, j,\boldsymbol{s}_{n}} \geq \phi_{n,j, \ell(\boldsymbol{\theta}), \boldsymbol{s}_{n}}$. Therefore, 
\begin{equation*}
    1-\phi_{n,j, \boldsymbol{s}_{n}} \leq 1-\phi_{n,j,\ell(\boldsymbol{\theta}), \boldsymbol{s}_{n}}. 
\end{equation*}
So
\begin{equation*}
    Q_{\boldsymbol{\theta},\boldsymbol{s}_{n}}^{(n)}(1-\phi_{n,j,\boldsymbol{s}_{n}}) \leq Q_{\boldsymbol{\theta}, \boldsymbol{s}_{n}}^{(n)}(1 - \phi_{n,j,\ell(\boldsymbol{\theta}), \boldsymbol{s}_{n}}). 
\end{equation*}
Hence 
\begin{equation*}
Q_{\boldsymbol{\theta},\boldsymbol{s}_{n}}^{(n)}(1-\phi_{n,j,\boldsymbol{s}_{n}}) \leq Q_{\boldsymbol{\theta},\boldsymbol{s}_{n}}^{(n)} (1-\phi_{n,j,\ell(\boldsymbol{\theta}), \boldsymbol{s}_{n}}) \leq e^{-\frac{1}{2}j^{2}n \epsilon_{n}^{2}}. 
\end{equation*}
Then, taking the supremum over $\boldsymbol{\theta} \in S_{n,j}$ gives 
\begin{equation*}
    \sup\limits_{\boldsymbol{\theta} \in S_{n,j}} Q_{\boldsymbol{\theta}, \boldsymbol{s}_{n}} ^{(n)} (1-\phi_{n,j,\boldsymbol{s}_{n}}) \leq e^{-\frac{1}{2}j^{2} n\epsilon_{n}^{2}}. 
\end{equation*}
Finally define 
\begin{equation*}
    \phi_{n,\boldsymbol{s}_{n}} := \sup\limits_{j \geq 2} \phi_{n,j,\boldsymbol{s}_{n}}.
\end{equation*}
Since 
\begin{equation*}
    \sup\limits_{j \geq 2} \phi_{n,j,\boldsymbol{s}_{n}} \leq \sum\limits_{j=2}^{\infty} \phi_{n,j,\boldsymbol{s}_{n}},
\end{equation*}
it follows that 
\begin{align*}
    Q_{\boldsymbol{\theta}_{0}, \boldsymbol{s}_{n}}^{(n)}  \phi_{n,\boldsymbol{s}_{n}}& \leq Q_{\boldsymbol{\theta}_{0}, 
    \boldsymbol{s}_{n}}^{(n)} (\sum\limits_{j = 2}^{\infty} \phi_{n,j,\boldsymbol{s}_{n}})\\
    &=\sum\limits_{j=2}^{\infty}Q_{\boldsymbol{\theta}_{0},\boldsymbol{s}_{n}}^{(n)}\phi_{n,j,\boldsymbol{s}_{n}}\\
    &\leq \sum\limits_{j=2}^{\infty} e^{-\frac{j^{2}n\epsilon_{n}^{2}}{4}}. 
\end{align*}
Observe that
\begin{equation*}
    n\epsilon_{n}^{2} = nc_{\epsilon}^{2}n^{- \frac{2}{2+c_{\boldsymbol{\theta}}}} = c_{\epsilon}^{2} n^{1 - \frac{2}{2+c_{\boldsymbol{\theta}}}} = c_{\epsilon}^{2}n^{\frac{c_{\boldsymbol{\theta}}}{2+c_{\boldsymbol{\theta}}}}. 
\end{equation*}
Then
\begin{equation*}
    a_{n}:= \frac{c_{\epsilon}^{2}}{4}n^{\frac{c_{\boldsymbol{\theta}}}{2+c_{\boldsymbol{\theta}}}} \rightarrow \infty
\end{equation*}
and therefore it follows that 
\begin{align*}
    Q_{\boldsymbol{\theta}_{0}, \boldsymbol{s}_{n}}^{(n)}\phi_{n,\boldsymbol{s}_{n}} &\leq \sum\limits_{ j =2}^{\infty} e^{-a_{n}j^{2}}\\
    &\leq \sum\limits_{j=2}^{\infty} e^{-2a_{n}j}, \text{ since } j^{2}\geq2j \text{ for } j \geq 2, \\
    & = \frac{e^{-4a_{n}}}{1-e^{-2a_{n}}}\\
    &\rightarrow 0. 
\end{align*}
Clearly, $\phi_{n,\boldsymbol{s}_{n}} \geq \phi_{n,j,\boldsymbol{s}_{n}}$. So $1-\phi_{n,\boldsymbol{s}_{n}} \leq 1- \phi_{n,j,\boldsymbol{s}_{n}}$. Fix $j \geq 2, \boldsymbol{s}_{n} \in H_{n}$ and $\boldsymbol{\theta} \in S_{n,j}$. By monotonicity of measure, 
\begin{equation*}
    Q_{\boldsymbol{\theta},\boldsymbol{s}_{n}}^{(n)} (1-\phi_{n,\boldsymbol{s}_{n}}) \leq Q_{\boldsymbol{\theta}, \boldsymbol{s}_{n}}^{(n)}(1-\phi_{n,j,\boldsymbol{s}_{n}}). 
\end{equation*}
Then taking the supremum gives 
\begin{equation*}
    \sup\limits_{\boldsymbol{\theta} \in S_{n,j}} Q_{\boldsymbol{\theta}, \boldsymbol{s}_{n}}^{(n)}(1-\phi_{n,\boldsymbol{s}_{n}}) \leq \sup\limits_{\boldsymbol{\theta} \in S_{n,j}} Q_{\boldsymbol{\theta}, \boldsymbol{s}_{n}}^{(n)}(1-\phi_{n,j,\boldsymbol{s}_{n}}) \leq e^{-\frac{1}{2}j^{2}n\epsilon_{n}^{2}}. 
\end{equation*}
We can also take the supremum over $H_{n}$, 
\begin{equation*}
    \sup\limits_{\boldsymbol{s}_{n}\in H_{n}} Q_{\boldsymbol{\theta}_{0},\boldsymbol{s}_{n}}^{(n)} \phi_{n,\boldsymbol{s}_{n}} \rightarrow 0 \text{ and } \sup\limits_{\boldsymbol{s}_{n} \in H_{n}} \sup\limits_{\boldsymbol{\theta} \in S_{n,j}} Q_{\boldsymbol{\theta},\boldsymbol{s}_{n}}^{(n)} (1-\phi_{n,\boldsymbol{s}_{n}}) \leq e^{-\frac{1}{2}j^{2} n\epsilon_{n}^{2}}. 
\end{equation*}
Note that any $K< \frac{1}{2}$ would work, e.g. $K = \frac{1}{4}$. 
\end{proof}
\subsection{Proof of Proposition 3}
\begin{proof}
We can trivially upper bound the numerator by one. Therefore, it remains to lower bound the denominator. To lower bound the denominator, we will show that the set $B_{n}(\boldsymbol{\theta}_{0}, \epsilon_{n}; 2)$ is a superset of some other set in view of (2.5) from \cite{ghosal2007convergence}, and then lower bound the denominator by the probability of this subset using monotonicity of measure.
Let
\begin{equation*}
    R^{(n)} = ({\bf{r}}_{1},...,{\bf{r}}_{n}) \text{ with }r_{a,t} = \alpha_{a,t} + r_{M,t}\beta_{a,t} + \epsilon_{a,t}.
\end{equation*}
Then
\begin{equation*}
    \boldsymbol{r}_{t}|\boldsymbol{\theta}_{t} \sim N_{N}({\boldsymbol{\alpha}}_{t} + r_{M,t} {\boldsymbol{\beta}}_{t}, \Sigma_{t}),
\end{equation*}
where $\Sigma_{t} = \Lambda\tilde{\Sigma}_{t}\Lambda^{T} + \bar{\Sigma}_{t}.$ Since
\begin{equation*}
    q_{\boldsymbol{\theta},\boldsymbol{s}_{n}}^{(n)}(\boldsymbol{r}^{(n)}) = \prod\limits_{t=1}^{n} q_{\boldsymbol{\theta},\boldsymbol{s}_{n,t}}(\boldsymbol{r}_{t}),
\end{equation*}
the log-likelihood is given by 
\begin{equation*}
    L_{n}(\boldsymbol{\theta}) = \log(q_{\boldsymbol{\theta},\boldsymbol{s}_{n}}^{(n)}(\boldsymbol{r}^{(n)})) = \sum\limits_{t=1}^{n}\log(q_{\boldsymbol{\theta},\boldsymbol{s}_{n,t}}(\boldsymbol{r}_{t})) = \sum\limits_{t=1}^{n}L_{t}(\boldsymbol{\theta})
\end{equation*}
Note, 
\begin{align*}
    L_{t}(\boldsymbol{\theta_{}}) &= \log[(2\pi)^{-\frac{N}{2}} \det(\Sigma_{t})^{-\frac{1}{2}} \exp \{ - \frac{1}{2}({\bf{r}}_{t}-\boldsymbol{\mu}_{t})^{T}\Sigma_{t}^{-1}({\bf{r}}_{t}-\boldsymbol{\mu}_{t}) \}] \\
    & = - \frac{1}{2}(\log((2\pi)^{N}) + \log(\det(\Sigma_{t})) + (\boldsymbol{r}_{t} - \boldsymbol{\mu}_{t})^{T}\Sigma_{t}^{-1}(\boldsymbol{r}_{t} - \boldsymbol{\mu}_{t})).
\end{align*}
We will now bound the KL divergence between $Q_{\boldsymbol{\theta},\boldsymbol{s}_{n}}$ and $Q_{\boldsymbol{\theta}_{0},\boldsymbol{s}_{n}}$. Fix $\boldsymbol{s}_{n} \in H_{n}.$ Since $\boldsymbol{\theta}_{0} \in \Theta$ and $\Theta$ is open, choose $\delta_{0}>0$ sufficiently small such that
\begin{equation*}
    U:= \{ \boldsymbol{\theta} : ||\boldsymbol{\theta} - \boldsymbol{\theta}_{0}|| \leq \delta_{0} \} \subset \Theta. 
\end{equation*}
Then $U$ is compact. Since $\boldsymbol{\theta}_{0}$ is fixed and the bounds defining $\Theta_{n}$ expand to cover the whole parameter space, $\exists n_{0} \in \mathbb{N}$ such that $U\subseteq \Theta_{n}$ for $n \geq n_{0}$. 
From the proof of Proposition 1, 
\begin{equation*}
    KL_{t}(\boldsymbol{\theta}_{0}, \boldsymbol{\theta}) = \frac{1}{2}(\text{tr}(A_{t}) - \log(\det(A_{t}))-N), A_{t} = \Sigma_{\boldsymbol{\theta},t}^{-\frac{1}{2}}\Sigma_{\boldsymbol{\theta}_{0},t} \Sigma_{\boldsymbol{\theta},t}^{-\frac{1}{2}}, \boldsymbol{\theta} \in U.
\end{equation*}
Then from the proof of Proposition 1, 
\begin{equation*}
    KL_{t}(\boldsymbol{\theta}_{0}, \boldsymbol{\theta}) = \frac{1}{2} \sum\limits_{i=1}^{N} (\lambda_{i,t}-1-\log(\lambda_{i,t})),
\end{equation*}
where $\lambda_{1,t},\dots, \lambda_{N,t}$ are the eigenvalues of $A_{t}$. By a similar compactness argument as in the proof of Proposition 1,define 
\begin{equation*}
    m_{t}: = \inf\limits_{\boldsymbol{\theta}\in U} \lambda_{\min}(\Sigma_{\boldsymbol{\theta},t})>0. 
\end{equation*}
Then, 
\begin{equation*}
    ||\Sigma_{\boldsymbol{\theta},t}^{-\frac{1}{2}}||_{OP}^{2} = ||\Sigma_{\boldsymbol{\theta},t}^{-1}||_{OP} \leq m_{t}^{-1}, \boldsymbol{\theta} \in U. 
\end{equation*}
Therefore, 
\begin{equation*}
    ||A_{t}-I||_{F} \leq m_{t}^{-1} ||\Sigma_{\boldsymbol{\theta},t}- \Sigma_{\boldsymbol{\theta}_{0},t}||_{F}.
\end{equation*}
Similarly to the proof of proposition 1, 
\begin{equation*}
    ||\Sigma_{t}(\boldsymbol{\theta})- \Sigma_{t}(\boldsymbol{\theta}_{0})||_{F} \leq G_{t}||\boldsymbol{\theta} - \boldsymbol{\theta}_{0}||_{2}, \boldsymbol{\theta} \in U, 
\end{equation*}
for some $G_{t}$. Therefore 
\begin{equation*}
    ||A_{t}-I||_{F} \leq m_{t}^{-1} G_{t}||\boldsymbol{\theta}- \boldsymbol{\theta}_{0}||_{2}. 
\end{equation*}
Similarly define 
\begin{equation*}
    M_{t} := \sup\limits_{\boldsymbol{\theta} \in U}\lambda_{\max}(\Sigma_{\boldsymbol{\theta},t})< \infty. 
\end{equation*}
So that
\begin{equation*}
    \frac{m_{t}}{M_{t}}I \preceq A_{t} \preceq \frac{M_{t}}{m_{t}}I. 
\end{equation*}
By a similar argument to the proof of proposition 1, 
\begin{align*}
    KL_{t}(\boldsymbol{\theta}_{0}, \boldsymbol{\theta}) &\leq \frac{1}{4}(\frac{M_{t}}{m_{t}})^{2}\sum\limits_{i=1}^{N} (\lambda_{i}(A_{t})-1)^{2}\\
    &= \frac{1}{4}(\frac{M_{t}}{m_{t}})^{2}||A_{t}-I||_{F}^{2}\\
    &\leq \frac{M_{t}^{2}G_{t}^{2}}{4m_{t}^{4}}||\boldsymbol{\theta}- \boldsymbol{\theta}_{0}||_{2}^{2}.\\
\end{align*}
Define 
\begin{equation*}
    c_{1,n} := \max\limits_{1 \leq t \leq n} \frac{M_{t}^{2}G_{t}^{2}}{4m_{t}^{4}}.
\end{equation*}
Then $\forall \boldsymbol{\theta} \in U$, 
\begin{equation*}
    KL_{t}(\boldsymbol{\theta}_{0}, \boldsymbol{\theta}) \leq c_{1,n}||\boldsymbol{\theta} - \boldsymbol{\theta}_{0}||_{2}^{2}. 
\end{equation*}
Hence 
\begin{equation*}
    KL_{n}(\boldsymbol{\theta}_{0}, \boldsymbol{\theta}) = \sum\limits_{t=1}^{n} KL_{t}(\boldsymbol{\theta}_{0}, \boldsymbol{\theta}) \leq c_{1,n}n||\boldsymbol{\theta} - \boldsymbol{\theta}_{0}||_{2}^{2} . 
\end{equation*}
If $\boldsymbol{X}_{t} \sim N(\boldsymbol{\mu}_{t}, \Sigma_{t})$, then 
 \begin{align*}
    \ell_{\boldsymbol{\theta}}(\boldsymbol{x}) &= \log\{ (2 \pi)^{- \frac{N}{2}}\det(\Sigma_{t})^{-\frac{1}{2}} \exp\{ -\frac{1}{2}(\boldsymbol{x}_{t}-\boldsymbol{\mu}_{t})^\top\Sigma_{t}^{-1}(\boldsymbol{x}_{t}-\boldsymbol{\mu}_{t}) \} \} \\ &= -\frac{N}{2}\log(2\pi) - \frac{1}{2}\log(\det(\Sigma_{t})) - \frac{1}{2}(\boldsymbol{x}_{t}-\boldsymbol{\mu}_{t})^\top \Sigma^{-1}_{t}(\boldsymbol{x}_{t}- \boldsymbol{\mu}_{t}). 
    \end{align*}
It follows that
\begin{align*}
\log(\frac{Q_{\boldsymbol{\theta}_{0},\boldsymbol{s}_{n},t}}{Q_{\boldsymbol{\theta},\boldsymbol{s}_{n},t}}) &= \log(Q_{\boldsymbol{\theta}_{0},\boldsymbol{s}_{n},t})-\log(Q_{\boldsymbol{\theta},\boldsymbol{s}_{n},t}) \\ & = -\frac{N}{2}\log(2\pi) - \frac{1}{2}\log(\det(\Sigma_{\boldsymbol{\theta}_{0},t})) - \frac{1}{2}(\boldsymbol{x}_{t} - \boldsymbol{\mu}_{0,t})^\top\Sigma_{\boldsymbol{\theta}_{0},t}^{-1}(\boldsymbol{x}_{t}- \boldsymbol{\mu}_{0,t}) + \frac{N}{2}\log(2\pi) \\ & + \frac{1}{2}\log(\det(\Sigma_{\boldsymbol{\theta},t})) + \frac{1}{2}(\boldsymbol{x}_{t}-\boldsymbol{\mu}_{\boldsymbol{\theta},t})^\top\Sigma_{\boldsymbol{\theta},t}^{-1}(\boldsymbol{x}_{t}- \boldsymbol{\mu}_{\boldsymbol{\theta},t})\\& = \frac{1}{2}(\boldsymbol{x}_{t}-\boldsymbol{\mu}_{\boldsymbol{\theta},t})^\top\Sigma_{\boldsymbol{\theta},t}^{-1}(\boldsymbol{x}_{t}-\boldsymbol{\mu}_{\boldsymbol{\theta},t})  + \frac{1}{2}\log(\det(\Sigma_{\boldsymbol{\theta},t}))\\& - \frac{1}{2}\log(\det(\Sigma_{\boldsymbol{\theta}_{0},t})) - \frac{1}{2}(\boldsymbol{x}_{t}- \boldsymbol{\mu}_{0,t})^\top\Sigma_{\boldsymbol{\theta}_{0},t}^{-1}(\boldsymbol{x}_{t} - \boldsymbol{\mu}_{0,t})\\ &  = \frac{1}{2}\{ (\boldsymbol{x}_{t}-\boldsymbol{\mu}_{\boldsymbol{\theta},t})^\top\Sigma_{\boldsymbol{\theta},t}^{-1}(\boldsymbol{x}_{t}-\boldsymbol{\mu}_{\boldsymbol{\theta},t}) -(\boldsymbol{x}_{t}-\boldsymbol{\mu}_{0,t})^\top\Sigma_{\boldsymbol{\theta}_{0},t}^{-1}(\boldsymbol{x}_{t}-\boldsymbol{\mu}_{0,t}) \} + \frac{1}{2}\log(\frac{|\Sigma_{\boldsymbol{\theta},t}|}{|\Sigma_{\boldsymbol{\theta}_{0},t}|}).  
    \end{align*}
Define 
\begin{equation*}
    \ell_{t}(\boldsymbol{\theta}):= \log(\frac{q_{\boldsymbol{\theta}_{0},\boldsymbol{s}_{n},t}(\boldsymbol{X}_{t})}{q_{\boldsymbol{\theta},\boldsymbol{s}_{n},t}(\boldsymbol{X}_{t})}). 
\end{equation*}
Let $\boldsymbol{\eta}_{t} = \boldsymbol{X}_{t}-\boldsymbol{\mu}_{\boldsymbol{\theta}_{0},t}$. Evaluating the log-likelihood ratio at $\boldsymbol{X}_{t}$ gives 
\begin{equation*}
    \ell_{t}(\boldsymbol{\theta}) = \frac{1}{2}\{ (\boldsymbol{X}_{t}-\boldsymbol{\mu}_{\boldsymbol{\theta},t})^\top\Sigma_{\boldsymbol{\theta},t}^{-1}(\boldsymbol{X}_{t} -\boldsymbol{\mu}_{\boldsymbol{\theta},t}) - \boldsymbol{\eta}_{t}^\top\Sigma_{\boldsymbol{\theta}_{0},t}^{-1}\boldsymbol{\eta}_{t} + \log(\frac{|\Sigma_{\boldsymbol{\theta},t}|}{|\Sigma_{\boldsymbol{\theta}_{0},t}|}) \}.
\end{equation*}
Therefore,
\begin{equation*}
    (\boldsymbol{X}_{t} - \boldsymbol{\mu}_{\boldsymbol{\theta},t}) = \boldsymbol{\eta}_{t} + \boldsymbol{\mu}_{\boldsymbol{\theta}_{0},t} - \boldsymbol{\mu}_{\boldsymbol{\theta},t} = \boldsymbol{\eta}_{t} - \Delta \boldsymbol{\mu}_{t},
\end{equation*}
where 
\begin{equation*}
    \Delta \boldsymbol{\mu}_{t} = \boldsymbol{\mu}_{\boldsymbol{\theta},t} - \boldsymbol{\mu}_{\boldsymbol{\theta}_{0},t}=\boldsymbol{0},
\end{equation*}
for fixed $\boldsymbol{s}_{n} \in H_{n}$ so that 
 \begin{align*}
     \ell_{t}(\boldsymbol{\theta}) &= \frac{1}{2} \{ (\boldsymbol{\eta}_{t} - \Delta \boldsymbol{\mu}_{t})^\top\Sigma_{\boldsymbol{\theta},t}^{-1}(\boldsymbol{\eta}_{t} - \Delta \boldsymbol{\mu}_{t}) - \boldsymbol{\eta}_{t}^\top\Sigma_{\boldsymbol{\theta}_{0},t}^{-1}\boldsymbol{\eta}_{t} + \log(\frac{|\Sigma_{\boldsymbol{\theta},t}|}{|\Sigma_{\boldsymbol{\theta}_{0},t}|})\}.
     \end{align*}
Noting that 
\begin{equation*}
    (\boldsymbol{\eta}_{t} - \Delta \boldsymbol{\mu}_{t})^\top\Sigma_{\boldsymbol{\theta},t}^{-1}(\boldsymbol{\eta}_{t} - \Delta\boldsymbol{\mu}_{t}) = \boldsymbol{\eta}_{t}^\top\Sigma_{\boldsymbol{\theta},t}^{-1}\boldsymbol{\eta}_{t}. 
\end{equation*}
It follows that 
     \begin{align*} \ell_{t}(\boldsymbol{\theta}) &= \frac{1}{2} \{ \log(\frac{|\Sigma_{\boldsymbol{\theta},t}|}{|\Sigma_{\boldsymbol{\theta}_{0},t}|})  \} + \frac{1}{2} \{ \boldsymbol{\eta}_{t}^\top\Sigma_{\boldsymbol{\theta},t}^{-1}\boldsymbol{\eta}_{t}  - \boldsymbol{\eta}_{t}^\top\Sigma_{\boldsymbol{\theta}_{0},t}^{-1}\boldsymbol{\eta}_{t} \}\\
     & = C_{\boldsymbol{\theta},t} + \frac{1}{2}\boldsymbol{\eta}_{t}^\top\Sigma_{\boldsymbol{\theta},t}^{-1}\boldsymbol{\eta}_{t}  - \frac{1}{2}\boldsymbol{\eta}_{t}^{\top}\Sigma_{\boldsymbol{\theta}_{0},t}^{-1}\boldsymbol{\eta}_{t} \\ 
     & = C_{\boldsymbol{\theta},t} + \frac{1}{2}\boldsymbol{\eta}_{t}^\top(\Sigma_{\boldsymbol{\theta},t}^{-1} - \Sigma_{\boldsymbol{\theta}_{0},t}^{-1})\boldsymbol{\eta}_{t} ,C_{\boldsymbol{\theta},t} = \frac{1}{2} \{ \log(\frac{|\Sigma_{\boldsymbol{\theta},t}|}{|\Sigma_{\boldsymbol{\theta}_{0},t}|})  \}.
 \end{align*}
 Since $\boldsymbol{\eta}_{t} \sim N(\boldsymbol{0}, \Sigma_{\boldsymbol{\theta}_{0},t}), $ then 
 \begin{align*}
     \mathbb{E}_{\boldsymbol{\theta}_{0}}[\boldsymbol{\eta}^\top_{t}(\Sigma_{\boldsymbol{\theta},t}^{-1}-\Sigma_{\boldsymbol{\theta}_{0},t}^{-1})\boldsymbol{\eta}_{t}] &= \mathbb{E}_{\boldsymbol{\theta}_{0}}[\boldsymbol{\eta}_{t}^\top\Sigma_{\boldsymbol{\theta},t}^{-1}\boldsymbol{\eta}_{t}] - \mathbb{E}_{\boldsymbol{\theta}_{0}}[\boldsymbol{\eta}_{t}^\top\Sigma_{\boldsymbol{\theta}_{0},t}^{-1}\boldsymbol{\eta}_{t}]\\& = tr(\Sigma_{\boldsymbol{\theta},t}^{-1}\Sigma_{\boldsymbol{\theta}_{0},t})-tr(\Sigma_{\boldsymbol{\theta}_{0},t}^{-1}\Sigma_{\boldsymbol{\theta}_{0},t}) \\&=tr(\Sigma_{\boldsymbol{\theta},t}^{-1}\Sigma_{\boldsymbol{\theta}_{0},t}) - N.
 \end{align*}
This implies that 
\begin{equation*}
    \mathbb{E}_{\boldsymbol{\theta}_{0}}[\ell_{t}(\boldsymbol{\theta})] = C_{\boldsymbol{\theta},t} + \frac{1}{2}tr(\Sigma_{\boldsymbol{\theta},t}^{-1}\Sigma_{\boldsymbol{\theta}_{0},t})-\frac{N}{2}.
\end{equation*}
Let $Y_{t} = \ell_{t}(\boldsymbol{\theta}) - \mathbb{E}_{\boldsymbol{\theta}_{0}}[\ell_{t}(\boldsymbol{\theta})]$. Then $\mathbb{E}_{\boldsymbol{\theta}_{0}}[Y_{t}]=0$. Now we have that
    \begin{align*}
    \text{Var}(Y_{t}) &= \text{Var}(\frac{1}{2}\{\boldsymbol{\eta}_{t}^\top(\Sigma_{\boldsymbol{\theta},t}^{-1}-\Sigma_{\boldsymbol{\theta}_{0},t}^{-1})\boldsymbol{\eta}_{t} -tr((\Sigma_{\boldsymbol{\theta},t}^{-1}-\Sigma_{\boldsymbol{\theta}_{0},t}^{-1})\Sigma_{\boldsymbol{\theta}_{0},t}) \}) \\ &= \text{Var}_{\boldsymbol{\theta}_{0}}(\frac{1}{2}\{ \boldsymbol{\eta}_{t}^\top(\Sigma_{\boldsymbol{\theta},t}^{-1}-\Sigma_{\boldsymbol{\theta}_{0},t}^{-1})\boldsymbol{\eta}_{t} -tr((\Sigma_{\boldsymbol{\theta},t}^{-1}-\Sigma_{\boldsymbol{\theta}_{0},t}^{-1})\Sigma_{\boldsymbol{\theta}_{0},t}) \})\\ 
     & = \frac{1}{4}\text{Var}_{\boldsymbol{\theta}_{0}}(\boldsymbol{\eta}_{t}^\top(\Sigma_{\boldsymbol{\theta},t}^{-1}-\Sigma_{\boldsymbol{\theta}_{0},t}^{-1})\boldsymbol{\eta}_{t}) \\
     & = \frac{1}{2}tr([(\Sigma_{\boldsymbol{\theta},t}^{-1}-\Sigma_{\boldsymbol{\theta}_{0},t}^{-1})\Sigma_{\boldsymbol{\theta}_{0},t}]^{2}). 
    \end{align*}
Define 
\begin{equation*}
    D_{t}:= (\Sigma_{\boldsymbol{\theta},t}^{-1} - \Sigma_{\boldsymbol{\theta}_{0},t}^{-1})\Sigma_{\boldsymbol{\theta}_{0},t} = \Sigma_{\boldsymbol{\theta},t}^{-1}\Sigma_{\boldsymbol{\theta}_{0},t} - I_{N}. 
\end{equation*}
Then, 
\begin{align*}
    \Sigma_{\boldsymbol{\theta},t}^{\frac{1}{2}} D_{t} \Sigma_{\boldsymbol{\theta},t}^{-\frac{1}{2}} &= \Sigma_{\boldsymbol{\theta},t}^{\frac{1}{2}} (\Sigma_{\boldsymbol{\theta},t}^{-1}\Sigma_{\boldsymbol{\theta}_{0},t}-I_{N})\Sigma_{\boldsymbol{\theta},t}^{-\frac{1}{2}}\\
    &= \Sigma_{\boldsymbol{\theta},t}^{-\frac{1}{2}} \Sigma_{\boldsymbol{\theta}_{0},t}\Sigma_{\boldsymbol{\theta},t}^{-\frac{1}{2}} - I_{N}\\
    &= A_{t}- I_{N}. 
\end{align*}
Therefore, $D_{t}$ and $A_{t}-I_{N}$ are similar. Since $A_{t} - I_{N}$ is symmetric,   
\begin{equation*}
    \text{tr}(D_{t}^{2}) = \text{tr}((A_{t} - I_{N})^{2}) = ||A_{t} - I_{N}||_{F}^{2}. 
\end{equation*}
Therefore, 
\begin{equation*}
    \text{Var}(Y_{t}) = \frac{1}{2}||A_{t} - I_{N}||_{F}^{2} \leq \frac{G_{t}^{2}}{2m_{t}^{2}} ||\boldsymbol{\theta} - \boldsymbol{\theta}_{0}||_{2}^{2}, c_{2,n}:= \max\limits_{1 \leq t \leq n} \frac{G_{t}^{2}}{2m_{t}^{2}}< \infty. 
\end{equation*}
Since conditional on $\boldsymbol{s}_{n}$, $Q_{\boldsymbol{\theta}_{0},\boldsymbol{s}_{n}}^{(n)} = \bigotimes\limits_{t=1}^{n} Q_{\boldsymbol{\theta}_{0}, \boldsymbol{s}_{n},t}$, the variables $Y_{1}, \dots Y_{n}$ are independent under $Q_{\boldsymbol{\theta}_{0}, \boldsymbol{s}_{n}}^{(n)}$. Therefore, 
\begin{equation*}
    V_{2,0}(Q_{\boldsymbol{\theta}_{0}, \boldsymbol{s}_{n}}^{(n)}, Q_{\boldsymbol{\theta}, \boldsymbol{s}_{n}}^{(n)}) = \text{Var}_{\boldsymbol{\theta}_{0}}(\sum\limits_{t=1}^{n}Y_{t}) = \sum\limits_{t=1}^{n}\text{Var}_{\boldsymbol{\theta}_{0}}(Y_{t}) \leq c_{2,n}n||\boldsymbol{\theta} - \boldsymbol{\theta}_{0}||_{2}^{2}. 
\end{equation*}

For $n \geq \max\{n_{0},16\}, U \subseteq \Theta_{n}$. Then the results from the proof of proposition 1 hold uniformly for $\boldsymbol{\theta} \in U, 1 \leq t \leq n$, and $\boldsymbol{s}_{n} \in H_{n}$. In particular, 
\begin{align*}
    \lambda_{\min}(\Sigma_{\boldsymbol{\theta},t}) &\geq m_{n} =: n^{-(3+p)},\\ \lambda_{\max}(\Sigma_{\boldsymbol{\theta},t}) &\leq n^{4+p}\{ 1+ (\log(n))^{2} \} =: M_{n}, \text{ and}\\
    ||\Sigma_{\boldsymbol{\theta},t} - \Sigma_{\boldsymbol{\theta}_{0},t}||_{F}&\leq NL_{n} ||\boldsymbol{\theta} - \boldsymbol{\theta}_{0}||_{2}, L_{n}:=c_{1}n^{c_{2}}\log(n)^{c_{3}}, c_{1}, c_{2},c_{3}>0. 
\end{align*}
Clearly, $m_{t} \geq m_{n}, M_{t} \leq M_{n}$. Then we can take $G_{t} = NL_{n}$ for $1 \leq t \leq n$. Thus, 
\begin{equation*}
    c_{1,n} \leq \frac{M_{n}^{2}N^{2}L_{n}^{2}}{4m_{n}^{4}} \text{ and } c_{2,n} \leq \frac{N^{2}L_{n}^{2}}{2m_{n}^{2}}.
\end{equation*}
Note that 
\begin{align*}
    \sqrt{c_{1,n}} &\leq \sqrt{\frac{M_{n}^{2}N^{2}L_{n}^{2}}{4m_{n}^{4}}} \\
    &= \frac{NM_{n}L_{n}}{2m_{n}^{2}}\\
    &= \frac{Nn^{4+p}\{1+(\log(n))^{2} \} c_{1}n^{c_{2}}\log(n)^{c_{3}}}{2n^{-2(3+p)}}\\
    &= \frac{Nc_{1}}{2}n^{4+p+6+2p+c_{2}}\log(n)^{c_{3}}\{ 1+(\log(n))^{2} \}. 
\end{align*}
Observe $\log(n) \leq n$ and $1+\log(n)^{2} \leq 2n^{2}$. Hence 
\begin{equation*}
    \sqrt{c_{1,n}} \leq Nc_{1}n^{12+3p+c_{2}+c_{3}} \text{ and } \sqrt{c_{2,n}} \leq Nc_{1}n^{3+p+c_{2}+c_{3}} \leq Nc_{1}n^{12+3p+c_{2}+c_{3}}.
\end{equation*}
Define $a:= 12 +3p+c_{2}+c_{3}, C:=Nc_{1}$. Then for $n\geq \max\{n_{0},16\}$, 
\begin{equation*}
    \max\{ c_{1,n}, c_{2,n}\} \leq C^{2}n^{2a}. 
\end{equation*}
Since $\epsilon_{n}(Cn^{a})^{-1} \rightarrow 0, \exists n_{1} \in \mathbb{N}$ such that $\epsilon_{n}(Cn^{a})^{-1} \leq \delta_{0}, \forall n \geq n_{1}$. Let $n_{0}' = \max \{ n_{0}, n_{1},16\}$. Then, for every $n\geq n_{0}'$, $U_{n} \subseteq \Theta_{n}$ and $\epsilon_{n}(Cn^{a})^{-1} \leq \delta_{0}$, where
\begin{equation*}
    \delta_{n}:= \epsilon_{n}(Cn^{a})^{-1} \text{ and } U_{n}:=\{ \boldsymbol{\theta} \in \Theta : ||\boldsymbol{\theta} - \boldsymbol{\theta}_{0}||_{2} \leq \delta_{n}\}.
\end{equation*}
Since $\delta_{n}\leq \delta_{0}, U_{n} \subseteq U$ for $n\geq n_{0}'$. Hence for every $\boldsymbol{\theta} \in U_{n}$, 
\begin{align*}
    KL_{n}(\boldsymbol{\theta}_{0}, \boldsymbol{\theta}) &\leq nc_{1,n}||\boldsymbol{\theta} - \boldsymbol{\theta}_{0}||_{2}^{2} \leq C^{2}n^{2a}n\frac{\epsilon_{n}^{2}}{C^{2}n^{2a}} = n \epsilon_{n}^{2}\\
    &\text{and}\\
    V_{2,0}(Q_{\boldsymbol{\theta}_{0},\boldsymbol{s}_{n}}^{{(n)}}, Q_{\boldsymbol{\theta},\boldsymbol{s}_{n}}^{(n)}) &\leq c_{2,n} n||\boldsymbol{\theta} - \boldsymbol{\theta}_{0}||_{2}^{2} \leq n\epsilon_{n}^{2}
\end{align*}
Therefore, $U_{n} \subseteq B_{n}(\boldsymbol{\theta}_{0}, \epsilon_{n};2)$ for $n \geq n_{0}'$ which implies $\Pi_{n}(B_{n}(\boldsymbol{\theta}_{0},\epsilon_{n};2)) \geq \Pi_{n}(U_{n})$. Now define $c_{4}:= C\sqrt{c_{\boldsymbol{\theta}}}$ and $r_{n} := \epsilon_{n}(c_{4}n^{a})^{-1}$. Let 
\begin{equation*}
    A_{n}^{\Lambda}:= \bigcap\limits_{i=1}^{N} \bigcap\limits_{j=1}^{r} \{ |\Lambda_{ij} - \Lambda_{0,ij}| \leq r_{n}  \}. 
\end{equation*}
Then, 
\begin{equation*}
    A_{n}^{\Lambda}\cap \bigcap\limits_{j=1}^{c_{\boldsymbol{\theta}}-Nr} \{ |\theta^{(j)} - \theta_{0}^{(j)}| \leq r_{n} \} \subseteq U_{n} 
\end{equation*}
because on this event 
\begin{equation*}
    ||\boldsymbol{\theta}-\boldsymbol{\theta}_{0}||_{2}^{2} \leq c_{\boldsymbol{\theta}}r_{n}^{2} = \frac{\epsilon_{n}^{2}}{C^{2}n^{2a}} = \delta_{n}^{2}. 
\end{equation*}
To lower bound the denominator, we can lower bound the probability of $U_{n}$. 
Let ${c}_{\boldsymbol{\theta}}$ denote the number of parameters in the model (see the definition of $\Theta_{n}$). The required probability is given by
\begin{equation*}
    \Pi_{n}(U_{n}) = \mathbb{P}(\boldsymbol{\theta} \in \Theta: ||\boldsymbol{\theta} - \boldsymbol{\theta}_{0}||_{2}\leq \delta_{n}). 
\end{equation*}
For the mean parameters in the Z-distributed AR(1) processes  
\begin{equation*}
    \mu = \log(\tau_{0}^{2}\tau_{1}^{2}), \tau_{0} \sim C^{+}(0, \frac{1}{\sqrt{n}}), \tau_{1}\sim C^{+}(0,1).
\end{equation*}
The probability density function of a half Cauchy distribution with scale $\sigma$ is $$\frac{2}{\pi \sigma} \frac{1}{1+\frac{y^{2}}{\sigma^{2}}} 
= \frac{2\sigma}{\pi(\sigma^{2} + y^{2})}.$$
Let $X = \tau^{2},$ then,
\begin{equation*}
    f_{X}(x) = f_{\tau}(\sqrt{x})\frac{d}{dx}(\sqrt{x}) =\frac{2\sigma}{\pi(\sigma^{2}+x)}\frac{1}{2\sqrt{x}} = \frac{\sigma}{\sqrt{x}(\sigma^{2}+x)\pi}.
\end{equation*}
Now, let $U = \log(X)$ so that 
\begin{equation*}
    f_{U}(u) = f_{X}(e^{u})e^{u} =\frac{\sigma e^{u}}{\sqrt{e^{u}}(\sigma^{2}+e^{u})\pi} = \frac{\sigma e^{\frac{1}{2}u}}{(\sigma^{2}+e^{u})\pi}.
\end{equation*}
For $i = 1, \dots, N,$ define $U_{i,0} := \log(\tau_{i,0}^{2}), U_{\alpha_{i}} := \log(\tau_{\alpha_{i}}^{2})$, and $U_{\beta_{i}}:= \log(\tau_{\beta_{i}}^{2})$. Then \\ $\mu_{\alpha_{i}}:= U_{i,0} + U_{\alpha_{i}}$ and $\mu_{\beta_{i}} = U_{i,0} + U_{\beta_{i}}$. Then we have that, 
\begin{equation*}
    f_{U_{i,0}}(u) = \frac{n^{-\frac{1}{2}}e^{\frac{u}{2}}}{\pi(n^{-1}+e^{u})} =: f_{0,n}(u). 
\end{equation*}
Since $\tau_{\alpha_{i}}, \tau_{\beta_{i}} \sim C^{+}(0,1)$, the common density of $U_{\alpha_{i}}$ and $U_{\beta_{i}}$ is 
\begin{equation*}
    f_{1}(u) = \frac{e^{\frac{u}{2}}}{\pi(1+e^{u})}. 
\end{equation*}
Thus, the joint density is 
\begin{equation*}
    f_{\mu_{\alpha_{i}}, \mu_{\beta_{i}},n} (x,y) = \int\limits_{\mathbb{R}}f_{0,n}(u)f_{1}(x-u)f_{1}(y-u)du. 
\end{equation*}
Define $I_{\alpha_{i}}:= [\mu_{\alpha_{i},0}-1, \mu_{\alpha_{i},0}+1]$ and $I_{\beta_{i}} := [\mu_{\beta_{i,0}}-1 , \mu_{\beta_{i,0}}+1]$. Note that $f_{1}$ is continuous and strictly positive. Now, consider
\begin{equation*}
    b_{\alpha_{i}} := \min\limits_{x \in I_{\alpha_{i}}, 0 \leq u \leq 1}f_{1}(x-u)>0 \text{ and } b_{\beta_{i}} := \min\limits_{y \in I_{\beta_{i}}, 0 \leq u \leq 1} f_{1}(y-u) >0.
\end{equation*}
For $0 \leq u \leq1$, 
\begin{align*}
    f_{0,n}(u) &= \frac{n^{-\frac{1}{2}}e^{\frac{u}{2}}}{\pi(n^{-1} + e^{u})}\\
    &\geq \frac{n^{-\frac{1}{2}}}{\pi(1+e)}. 
\end{align*}
Therefore, 
\begin{align*}
    f_{\mu_{\alpha_{i}}, \mu_{\beta_{i}},n}(x,y) &\geq \int\limits_{0}^{1} f_{0,n}(u)f_{1}(x-u)f_{1}(y-u)du\\
    &\geq \frac{b_{\alpha_{i}}b_{\beta_{i}}n^{-\frac{1}{2}}}{\pi(1+e)} >0. 
\end{align*}
Since $r_{n} \rightarrow 0$, choose $n_{2} \in \mathbb{N}$ such that $\forall n \geq n_{2}, r_{n} \leq 1$. Define 
\begin{equation*}
    n_{0}^{''} = \max\{n_{2}, n_{0}' \} \text{ and } b_{\mu,i} := \frac{b_{\alpha_{i}}b_{\beta_{i}}}{\pi(1+e)}.
\end{equation*}
Then, 
\begin{align*}
    \Pi_{n}(|\mu_{\alpha_{i}}-\mu_{\alpha_{i},0}| \leq r_{n}, |\mu_{\beta_{i}}-\mu_{\beta_{i},0}| \leq r_{n}) &\geq \int\limits_{\mu_{\alpha_{i},0}-r_{n}}^{\mu_{\alpha_{i},0}+r_{n}} \int\limits_{\mu_{\beta_{i},0}-r_{n}}^{\mu_{\beta_{i},0}+r_{n}}b_{\mu,i}n^{-\frac{1}{2}} dydx\\
    &= 4b_{\mu,i}n^{-\frac{1}{2}}r_{n}^{2}\\
    &\geq 4b_{\mu,i} (n^{-\frac{1}{2}}r_{n})^{2}. 
\end{align*}
Define $a_{\mu,i}:= 4b_{\mu,i}>0$. Then, 
\begin{equation*}
    \Pi_{n}(|\mu_{\alpha_{i}}-\mu_{\alpha_{i},0}| \leq r_{n}, |\mu_{\beta_{i}}- \mu_{\beta_{i},0}| \leq r_{n}) \geq a_{\mu,i} (n^{-\frac{1}{2}}r_{n})^{2}. 
\end{equation*}
Let $\mathcal{J}$ index the scalar coordinates of $\boldsymbol{\theta},$ excluding $\{ (\mu_{\alpha_{i}}, \mu_{\beta_{i}}): i=1,\dots,N \}$ and the $Nr$ entries of $\Lambda$. Thus, $|\mathcal{J}| = c_{\boldsymbol{\theta}}- Nr-2N$. For every $j \in \mathcal{J}$, let $f_{j}$ denote the prior density of the corresponding coordinate $\theta^{(j)}$. By the prior specification, $f_{j}$ is continuous and strictly positive in a neighborhood of the true value $\theta_{0}^{(j)}$. This includes the transformed-Beta priors on $(-1,1)$, the Gaussian priors on $\mathbb{R}$, and the Gamma priors on $(0, \infty)$. If a coordinate of $\boldsymbol{\theta}$ is parameterized as $\sigma^{2}$, while the prior is specified for $\sigma$, use the induced density
\begin{equation*}
    f_{\sigma^{2}}(z) = \frac{f_{\sigma}(\sqrt{z})}{2\sqrt{z}}, z>0, 
\end{equation*}
which is also continuous and strictly positive in a neighborhood of every positive true value. For each $j \in \mathcal{J}$, choose $\rho_{j}>0$ sufficiently small such that $[\theta_{0}^{(j)}-\rho_{j}, \theta_{0}^{(j)}+\rho_{j}]$ lies entirely in a region on which $f_{j}$ is continuous and strictly positive. Define $b_{j}:= \min\limits_{|x-\theta_{0}^{(j)}| \leq \rho_{j}} f_{j}(x) >0$. Since $\mathcal{J}$ is finite, $\rho_{*} = \min\limits_{j\in \mathcal{J}} \rho_{j}>0$. Hence, $\exists n_{3} \in \mathbb{N}$ such that $r_{n} \leq \rho_{*}, \forall n\geq n_{3}$. Define $n_{0}^{'''} := \max\{ n_{0}^{''}, n_{3},2 \}$. Then, for every $n \geq n_{0}^{'''}$ and every $j \in \mathcal{J}, r_{n} \leq \rho_{j}$. Consequently, 
\begin{align*}
    \Pi_{n}(|\theta^{(j)} - \theta_{0}^{(j)}| \leq r_{n}) &= \int\limits_{\theta_{0}^{(j)}-r_{n}}^{\theta_{0}^{(j)} + r_{n}} f_{j}(x) dx \\
    &\geq 2b_{j}r_{n}\\
    &\geq 2b_{j}n^{-\frac{1}{2}}r_{n}. 
\end{align*}
Define $a_{j}:= 2b_{j}>0$. Then 
\begin{equation*}
    \Pi_{n}(|\theta^{(j)} - \theta_{0}^{(j)}|\leq r_{n}) \geq a_{j}n^{-\frac{1}{2}}r_{n}, j \in \mathcal{J}, \forall n \geq n_{0}^{'''}. 
\end{equation*}
Note 
\begin{equation*}
    f_{\Lambda_{ij}|\tau_{ij}^{2}} (x|t) = \frac{1}{\sqrt{2\pi t}} \exp \{ -\frac{x^{2}}{2t} \}, t>0, 
\end{equation*}
\begin{equation*}
    f_{\tau_{ij}^{2}|\lambda_{i}^{2}} (t|\ell) = \frac{(0.05\ell)^{0.1}}{\Gamma(0.1)}t^{-0.9}e^{-0.05\ell t}, t>0, \ell >0, 
\end{equation*}
and 
\begin{equation*}
    f_{\lambda_{i}^{2}} (\ell) = e^{- \ell}, \ell >0. 
\end{equation*}
Define 
\begin{equation*}
    E_{i}:= \{ \lambda_{i}^{2} \in [1,2], \tau_{ij}^{2} \in [1,2], \forall j = 1, \dots, r  \}. 
\end{equation*}
For $\ell \in [1,2]$, 
\begin{equation*}
    f_{\lambda_{i}^{2}}(\ell) = e^{-\ell} \geq e^{-2}. 
\end{equation*}
Hence 
\begin{equation*}
    \Pi(\lambda_{i}^{2} \in [1,2]) = \int\limits_{1}^{2}e^{- \ell} d\ell = e^{-1} - e^{-2}>0. 
\end{equation*}
For $\ell \in [1,2]$ and $t \in [1,2]$, $f_{\tau_{ij}^{2}|\lambda_{i}^{2}}$ is bounded below by 
\begin{equation*}
    c_{\tau} := \frac{(0.05)^{0.1}}{\Gamma(0.1)}2^{-0.9}e^{-0.2}>0. 
\end{equation*}
Therefore, 
\begin{equation*}
    \Pi(\tau_{ij}^{2}\in [1,2]|\lambda_{i}^{2} = \ell) = \int\limits_{1}^{2} f_{\tau_{ij}^{2}|\lambda_{i}^{2}}(t|\ell) dt \geq c_{\tau}. 
\end{equation*}
Since conditional on $\lambda_{i}^{2}$, the variables $\tau_{i1}^{2}, \dots, \tau_{ir}^{2}$ are independent, which implies
\begin{equation*}
    \Pi(\tau_{ij}^{2} \in [1,2], \forall j| \lambda_{i}^{2} = \ell) \geq c_{\tau}^{r}. 
\end{equation*}
Integrating over $\ell$ gives, 
\begin{equation*}
    \Pi(E_{i}) = \int\limits_{1}^{2} \Pi(\tau_{ij}^{2}\in [1,2], \forall j|\lambda_{i}^{2} = \ell)f_{\lambda_{i}^{2}}(\ell) d\ell \geq (e^{-1} - e^{-2})c_{\tau}^{r}. 
\end{equation*}
So that
\begin{equation*}
    \Pi(E_{i}) \geq c_{E,i}, c_{E,i}:= (e^{-1}-e^{-2})c_{\tau}^{r}>0. 
\end{equation*}
Now fix $i$ and $j$. On $E_{i}$ we have $\tau_{ij}^{2} \in [1,2]$. Also, for $n \geq n_{0}^{'''},r_{n} \leq 1$. Hence, if $|x-\Lambda_{0,ij}| \leq r_{n}$, then 
\begin{equation*}
    |x|\leq |\Lambda_{0,ij}| + |x-\Lambda_{0,ij}|\leq |\Lambda_{0,ij}| + r_{n} \leq |\Lambda_{0,ij}|+1. 
\end{equation*}
Thus, for $t \in [1,2]$ and $|x-\Lambda_{0,ij}| \leq r_{n}$, 
\begin{align*}
    f_{\Lambda_{ij}|\tau_{ij}^{2}}(x|t) &= \frac{1}{\sqrt{2\pi t}}\exp\{ - \frac{x^{2}}{2t} \} \\
    &\geq \frac{1}{\sqrt{4\pi}}\exp\{ -\frac{(|\Lambda_{0,ij}|+1)^{2}}{2} \}>0. 
\end{align*}
Then define, 
\begin{equation*}
    c_{ij} := \frac{1}{\sqrt{4\pi}} \exp \{- \frac{(|\Lambda_{0,ij}|+1)^{2}}{2} \}>0 . 
\end{equation*}
Therefore, whenever $t \in [1,2]$, 
\begin{equation*}
    \Pi(|\Lambda_{ij}-\Lambda_{0,ij}| \leq r_{n}|\tau_{ij}^{2}=t) = \int\limits_{\Lambda_{0,ij}-r_{n}}^{\Lambda_{0,ij}+r_{n}} f_{\Lambda_{ij}|\tau_{ij}^{2}}(x|t)dx \geq 2c_{ij}r_{n}. 
\end{equation*}
Define 
\begin{equation*}
    A_{i,n}:= \bigcap\limits_{j = 1}^{r} \{|\Lambda_{ij} - \Lambda_{0,ij}| \leq r_{n} \}. 
\end{equation*}
Conditional on $\tau_{i1}^{2}, \dots, \tau_{ir}^{2}$, the variables $\Lambda_{i1}, \dots, \Lambda_{ir}$ are independent, so on the event $E_{i}$, 
\begin{equation*}
    \Pi(A_{i,n}|\tau_{i1}^{2}, \dots, \tau_{ir}^{2}, \lambda_{i}^{2}) = \prod\limits_{j = 1}^{r} \Pi(|\Lambda_{ij} - \Lambda_{0,ij}| \leq r_{n} | \tau_{ij}^{2}) \geq \prod\limits_{j=1}^{r} (2 c_{ij}r_{n}). 
\end{equation*}
Thus, 
\begin{equation*}
    \Pi(A_{i,n}) \geq \mathbb{E}[\boldsymbol{1}\{ E_{i} \} \Pi(A_{i,n}|\tau_{i1}^{2}, \dots, \tau_{ir}^{2}, \lambda_{i}^{2})] \geq \Pi(E_{i}) 2^{r}r_{n}^{r}\prod\limits_{j=1}^{r}c_{ij} \geq a_{\Lambda,i}r_{n}^{r}, 
\end{equation*}
where 
\begin{equation*}
    a_{\Lambda,i}:= 2^{r}c_{E,i} \prod\limits_{j=1}^{r} c_{ij}>0.
\end{equation*}
Since
\begin{equation*}
    A_{n}^{\Lambda} := \bigcap\limits_{i=1}^{N} A_{i,n} = \bigcap\limits_{i=1}^{N} \bigcap\limits_{j=1}^{r}\{|\Lambda_{ij} - \Lambda_{0,ij}| \leq r_{n} \}. 
\end{equation*}
and the prior is independent across rows $i$, the events $A_{1,n}, \dots, A_{N,n}$ are independent. Hence, 
\begin{equation*}
    \Pi(A_{n}^{\Lambda}) = \prod\limits_{i=1}^{N} \Pi(A_{i,n}) \geq \prod\limits_{i=1}^{N} a_{\Lambda,i}r_{n}^{Nr} = a_{\Lambda}r_{n}^{Nr}, a_{\Lambda}:= \prod\limits_{i=1}^{N}a_{\Lambda,i}>0. 
\end{equation*}
Finally, because $n^{-\frac{Nr}{2}} \leq 1$, 
\begin{equation*}
    \Pi(A_{n}^{\Lambda}) \geq a_{\Lambda}(n^{-\frac{1}{2}}r_{n})^{Nr}. 
\end{equation*}
Define 
\begin{equation*}
    E_{\mu,i,n} := \{ |\mu_{\alpha_{i}} - \mu_{\alpha_{i},0}| \leq r_{n}, |\mu_{\beta_{i}} - \mu_{\beta_{i},0}| \leq r_{n} \}. 
\end{equation*}
Then, 
\begin{equation*}
    \Pi_{n}(U_{n}) \geq \Pi_{n}(A_{n}^{\Lambda}) \prod\limits_{i=1}^{N}\Pi_{n}(E_{\mu,i,n}) \prod\limits_{j \in \mathcal{J}} \Pi_{n}(|\theta^{(j)} - \theta_{0}^{(j)}| \leq r_{n}). 
\end{equation*}
Define 
\begin{equation*}
    m:= \min\{ a_{\Lambda}^{\frac{1}{Nr}}, a_{\mu,1}^{\frac{1}{2}}, \dots, a_{\mu, N}^{\frac{1}{2}}, \min\limits_{j \in \mathcal{J}} a_{j} \}>0. 
\end{equation*}
Then, 
\begin{align*}
    \Pi_{n}(B_{n}(\boldsymbol{\theta}_{0}, \epsilon_{n};2)) &\geq \Pi_{n}(U_{n})\\
    &\geq m^{Nr} (n^{-\frac{1}{2}}r_{n})^{Nr}m^{2N}(n^{-\frac{1}{2}}r_{n})^{2N} m^{c_{\boldsymbol{\theta}}-Nr-2N}(n^{-\frac{1}{2}}r_{n})^{c_{\boldsymbol{\theta}}-Nr-2N}\\
    &= m^{c_{\boldsymbol{\theta}}}(n^{-\frac{1}{2}}r_{n})^{c_{\boldsymbol{\theta}}}
\end{align*}
Define $\bar{m} := m/c_{4}>0$. Then, for $n \geq n_{0}^{'''}$, 
\begin{equation*}
    \Pi_{n}(B_{n}(\boldsymbol{\theta}_{0}, \epsilon_{n}; 2))\geq (\bar{m}n^{-(a+\frac{1}{2})}\epsilon_{n})^{c_{\boldsymbol{\theta}}}. 
\end{equation*}
Let $\gamma := c_{\boldsymbol{\theta}}/(2+c_{\boldsymbol{\theta}})$. Note that $n\epsilon_{n}^{2} = c_{\epsilon}^{2}n^{\gamma}$. Furthermore, 
\begin{align*}
    -\log(\Pi_{n}(B_{n}(\boldsymbol{\theta}_{0}, \epsilon_{n};2))) &\leq -c_{\boldsymbol{\theta}} \log(\bar{m} n^{-(a+\frac{1}{2})}\epsilon_{n})\\
    &= [c_{\boldsymbol{\theta}}(a+\frac{1}{2}) + \frac{c_{\boldsymbol{\theta}}}{2+c_{\boldsymbol{\theta}}}]\log(n)-c_{\boldsymbol{\theta}}\log(\bar{m}c_{\epsilon})
\end{align*}
Define, 
\begin{equation*}
    A:= c_{\boldsymbol{\theta}}(a+\frac{1}{2}) + \frac{c_{\boldsymbol{\theta}}}{2+c_{\boldsymbol{\theta}}}, B:= -c_{\boldsymbol{\theta}}\log(\bar{m}c_{\epsilon}), \text{ and } B_{+} = \max \{ B,0 \}. 
\end{equation*}
We know that 
\begin{equation*}
    \log(n) \leq \frac{1}{e\gamma} n^{\gamma}
\end{equation*}
and for $n \geq n_{0}^{'''}$, 
\begin{equation*}
    B_{+} \leq \frac{B_{+}}{(n_{0}^{'''})^{\gamma}} n^{\gamma}. 
\end{equation*}
Therefore 
\begin{equation*}
    - \log(\Pi_{n}(B_{n}(\boldsymbol{\theta}_{0}, \epsilon_{n};2))) \leq D{n}^{\gamma }, \text{ where } D:=\frac{A}{e\gamma} + \frac{B_{+}}{(n_{0}^{'''})^{\gamma}}
\end{equation*}
Then, for fixed $K \in (0, 1/2)$, define 
\begin{equation*}
    J_{1} := \max \{ 1, \lceil \sqrt{\frac{2D}{Kc_{\epsilon}^{2}}} \rceil \}. 
\end{equation*}
Then, for $n \geq n_{0}^{'''}$ and $j \geq J_{1}$, 
\begin{align*}
    \frac{1}{\Pi_{n}(B_{n}(\boldsymbol{\theta}_{0}, \epsilon_{n};2))} &\leq \exp\{ Dn^{\gamma} \}\\
    &\leq \exp\{ \frac{1}{2} Kc_{\epsilon}^{2} j^{2} n^{\gamma} \}\\
    &= \exp \{ \frac{1}{2}Kj^{2} n\epsilon_{n}^{2}\}. 
\end{align*}
Since 
\begin{equation*}
    \Pi_{n}(\boldsymbol{\theta} \in \Theta_{n}: j\epsilon_{n} < d_{n,\boldsymbol{s}_{n}}(\boldsymbol{\theta}, \boldsymbol{\theta}_{0}) \leq 2j\epsilon_{n}) \leq 1,
\end{equation*}
it follows that for $n \geq n_{0}^{'''}$ and $j \geq J_{1}$, 
\begin{equation*}
    \frac{\Pi_{n}(\boldsymbol{\theta} \in \Theta_{n}: j\epsilon_{n} < d_{n, \boldsymbol{s}_{n}}(\boldsymbol{\theta}, \boldsymbol{\theta}_{0}) \leq 2j\epsilon_{n})}{\Pi_{n}(B_{n}(\boldsymbol{\theta}_{0}, \epsilon_{n};2))} \leq \exp \{ \frac{1}{2} Kj^{2}n \epsilon_{n}^{2} \}. 
\end{equation*}
Since 
\begin{equation*}
    0 \leq h_{t}^{2}(\boldsymbol{\theta}, \boldsymbol{\theta}_{0}) \leq 2 \text{ and } d_{n,\boldsymbol{s}_{n}}^{2}(\boldsymbol{\theta}, \boldsymbol{\theta}_{0}) = \frac{1}{n}\sum\limits_{t=1}^{n} h_{t}^{2}(\boldsymbol{\theta}, \boldsymbol{\theta}_{0}), 
\end{equation*}
we have for every $n, \boldsymbol{\theta}, \boldsymbol{s}_{n}, d_{n,\boldsymbol{s}_{n}}(\boldsymbol{\theta},\boldsymbol{\theta}_{0}) \leq \sqrt{2}$. Define 
\begin{equation*}
    J_{2} := 1+ \max\limits_{1\leq n <n_{0}^{'''}} \lceil \frac{\sqrt{2}}{\epsilon_{n}} \rceil. 
\end{equation*}
Then, for $1 \leq n < n_{0}^{'''}$ and $j \geq J_{2}, j\epsilon_{n} > \sqrt{2}$. Hence 
\begin{equation*}
    \{ \boldsymbol{\theta} \in \Theta_{n}:j\epsilon_{n} < d_{n,\boldsymbol{s}_{n}}(\boldsymbol{\theta}, \boldsymbol{\theta}_{0}) \leq 2j\epsilon_{n} \} = \varnothing. 
\end{equation*}
Therefore, its prior probability is zero. Define $J_{0} := \max\{ J_{1}, J_{2}\}$. Then, for every $n \in \mathbb{N}$, every $\boldsymbol{s}_{n} \in H_{n}$, and every $j \geq J_{0}$, 
\begin{equation*}
    \frac{\Pi_{n}(\boldsymbol{\theta} \in \Theta_{n}: j\epsilon_{n} < d_{n, \boldsymbol{s}_{n}}(\boldsymbol{\theta}, \boldsymbol{\theta}_{0}) \leq 2j\epsilon_{n})}{\Pi_{n}(B_{n}(\boldsymbol{\theta}_{0}, \epsilon_{n}; 2))} \leq \exp \{ \frac{1}{2} Kj^{2} n\epsilon_{n}^{2} \}. 
\end{equation*}
\end{proof}
\subsection{Proof of Proposition 4}
\begin{proof}
Define 
\begin{align*}
    A_{M,n}&:= \{ \max\limits_{1 \leq t \leq n} |h_{M,t-1}| \leq |\mu_{M,0}| + p\log(\bar{n}) \},\\ 
    A_{\alpha,n} &:= \{ \max\limits_{1 \leq i \leq N, 1 \leq t \leq n-1} |h_{\alpha_{i}, t-1}| \leq |\mu_{\alpha_{i},0}| +\frac{c\log(\bar{n}) + \log(2K_{\alpha,i}(s))}{s} \}, \\
    A_{\beta,n} &:= \{ \max\limits_{1 \leq i \leq N, 1 \leq t \leq n} |\beta_{i,t-1}| \leq \sqrt{(\kappa+1) C_{4}\bar{n}^{\alpha} \log(\bar{n})} \}, \\ 
    A_{\tilde{h},{n}} &:= \{ \max\limits_{1 \leq k \leq r, 1 \leq t \leq n} |\tilde{h}_{k,t-1}| \leq  p\log(\bar{n}) \}, \\ 
    A_{\bar{h},{n}} &:=  \{ \max \limits_{1 \leq i \leq N, 1 \leq t \leq n}|\bar{h}_{i,t-1} | \leq |\bar{\mu}_{i,0}| + p\log(\bar{n}) \}. 
\end{align*}
Then 
\begin{equation*}
    H_{n} = A_{M,n} \cap A_{\alpha,n} \cap A_{\beta,n} \cap A_{\tilde{h}, n} \cap A_{\bar{h}, n}.
\end{equation*}
Then by the union bound
\begin{equation*}
    \mathbb{P}_{\boldsymbol{\theta}_{0}} (H_{n}^{c}) \leq \mathbb{P}_{\boldsymbol{\theta}_{0}}(A_{M,n}^{c}) + \mathbb{P}_{\boldsymbol{\theta}_{0}}( A_{\alpha, n}^{c}) + \mathbb{P}_{\boldsymbol{\theta}_{0}}( A_{\beta,n}^{c}) + \mathbb{P}_{\boldsymbol{\theta}_{0}}(A_{\tilde{h},{n}}^{c}) + \mathbb{P}_{\boldsymbol{\theta}_{0}}( A_{\bar{h},{n}}^{c}). 
\end{equation*}
For the Gaussian AR(1) log-volatility processes, first consider a generic process of the form 
\begin{equation*}
    x_{t} = \mu_{0} + \phi_{0} (x_{t-1} - \mu_{0}) + \sigma_{0} \eta_{t}, \eta_{t} \sim N(0,1), |\phi_{0}| < 1. 
\end{equation*}
Then 
\begin{equation*}
    x_{t} - \mu_{0} = \sum\limits_{j=0}^{\infty} \phi_{0}^{j} \sigma_{0}\eta_{t-j}. 
\end{equation*}
This means 
\begin{equation*}
    x_{t} - \mu_{0} \sim N(0, \frac{\sigma_{0}^{2}}{1-\phi_{0}^{2}}). 
\end{equation*}
So for $u \in \mathbb{R}$, 
\begin{equation*}
    \mathbb{E}_{\boldsymbol{\theta}_{0}}[\exp \{ u(x_{t}-\mu_{0}) \}] = \exp \{ \frac{\sigma_{0}^{2}}{2(1-\phi_{0}^{2})}u^{2} \}.
\end{equation*}
Then by Chernoff's bound 
\begin{equation*}
    \mathbb{P}_{\boldsymbol{\theta}_{0}} (x_{t} - \mu_{0} > y) \leq \exp\{ \frac{\sigma_{0}^{2}u^{2}}{2(1-\phi_{0}^{2})}-uy \}. 
\end{equation*}
Similarly 
\begin{equation*}
    \mathbb{P}_{\boldsymbol{\theta}_{0}} (x_{t} - \mu_{0} < -y) \leq \exp \{ \frac{\sigma_{0}^{2}u^{2}}{2(1- \phi_{0}^{2})}-uy \}. 
\end{equation*}
Minimizing in $u$, with minimum $u  =(1-\phi_{0}^{2})y/\sigma_{0}^{2}$ yields 
\begin{equation*}
    \mathbb{P}_{\boldsymbol{\theta}_{0}} (|x_{t}-\mu_{0}| > y) \leq 2\exp \{ - \frac{(1-\phi_{0}^{2})y^{2}}{2\sigma_{0}^{2}} \}. 
\end{equation*}
Then let $y = p\log(\bar{n})$, so 
\begin{equation*}
    \mathbb{P}_{\boldsymbol{\theta}_{0}}(|x_{t} - \mu_{0}| > p\log(\bar{n})) \leq 2\exp \{ - \frac{(1-\phi_{0}^{2})p^{2}\log^{2}(\bar{n})}{2\sigma_{0}^{2}} \}. 
\end{equation*}
For $h_{M,t},$
\begin{align*}
    \mathbb{P}_{\boldsymbol{\theta}_{0}} (|h_{M,t}| > |\mu_{M,0}| + p\log(\bar{n})) &\leq \mathbb{P}_{\boldsymbol{\theta}_{0}}(|h_{M,t}-\mu_{M,0}|>p\log(\bar{n}))\\
    & \leq  2\exp \{ - \frac{(1-\phi_{M,0}^{2}) p^{2} \log^{2}(\bar{n})}{2\sigma_{M,0}^{2}} \}. 
\end{align*}
By the union bound 
\begin{equation*}
    \mathbb{P}_{\boldsymbol{\theta}_{0}} (A_{M,n}^{c}) \leq 2n\exp\{ - \frac{(1-\phi_{M,0}^{2})p^{2}\log^{2}(\bar{n})}{2\sigma_{M,0}^{2}} \} \rightarrow 0. 
\end{equation*}
For each fixed $k$, 
\begin{equation*}
    \mathbb{P}_{\boldsymbol{\theta}_{0}} (|\tilde{h}_{k,t}| > p\log(\bar{n})) \leq 2\exp\{ - \frac{(1-\tilde{\psi}_{k,0}^{2})p^{2}\log^{2}(\bar{n})}{2\tilde{\sigma}_{k,0}^{2}} \}. 
\end{equation*}
Therefore, 
\begin{align*}
    \mathbb{P}_{\boldsymbol{\theta}_{0}}(A_{\tilde{h},{n}}^{c})  & \leq \sum\limits_{k = 1}^{r} \sum\limits_{t=1}^{n} \mathbb{P}_{\boldsymbol{\theta}_{0}} (|\tilde{h}_{k,t-1}| >  p \log(\bar{n})) \\ 
    & \leq 2rn\max\limits_{1 \leq k \leq r} \exp \{ - \frac{(1-\tilde{\psi}_{k,0}^{2})p^{2} \log^{2}(\bar{n})}{2\tilde{\sigma}_{k,0}^{2}} \} \rightarrow 0. 
\end{align*}
Similarly, for each fixed $i$, 
\begin{equation*}
    \mathbb{P}_{\boldsymbol{\theta}_{0}}(|\bar{h}_{i,t}| > |\bar{\mu}_{i,0}| + p \log(\bar{n})) \leq 2\exp \{ -\frac{(1-\bar{\psi}_{i,0}^{2})p^{2}\log^{2}(\bar{n})}{2\bar{\sigma}_{i,0}^{2}}\}. 
\end{equation*}
Hence
\begin{equation*}
    \mathbb{P}_{\boldsymbol{\theta}_{0}}(A_{\bar{h},{n}}^{c}) \leq 2Nn \max\limits_{1 \leq i \leq N} \exp \{ - \frac{(1-\bar{\psi}_{i,0}^{2})p^{2}\log^{2}(\bar{n})}{2\bar{\sigma}_{i,0}^{2}} \} \rightarrow 0. 
\end{equation*}
Now consider the Z-innovation AR(1) process. Consider a generic process 
\begin{equation*}
    x_{t} = \mu_{0} + \phi_{0}(x_{t-1} - \mu_{0}) + \epsilon_{t}, |\phi_{0}|<1, \epsilon_{t}\sim Z(\frac{1}{2}, \frac{1}{2}, 0,1). 
\end{equation*}
Before proceeding we first need to derive the moment generating function of the four-parameter $Z$-distribution. The moment generating function of $X \sim Z(\frac{1}{2},\frac{1}{2},0,1) $ is  
\begin{equation*}
    M(t) = \mathbb{E}[e^{tx}] = \int\limits_{-\infty}^{\infty} e^{tx}f(x) dx = \frac{1}{\pi}\int\limits_{-\infty}^{\infty}\frac{e^{(t+\frac{1}{2})x}}{1+e^{x}}dx. 
\end{equation*} 
By the change of variables $u = e^{x}$ it follows that 
\begin{equation*}
    M(t) = \frac{1}{\pi}\int\limits_{0}^{\infty}\frac{u^{t+\frac{1}{2}}}{1+u}\frac{du}{u} = \frac{1}{\pi}\int\limits_{0}^{\infty} \frac{u^{t-\frac{1}{2}}}{1+u}du. 
\end{equation*}
Letting $a = t + \frac{1}{2}$
\begin{equation*}M(t) = \frac{1}{\pi}\int\limits_{0}^{\infty}\frac{u^{a-1}}{1+u} du.
\end{equation*}
Another change of variables with $v = \frac{u}{1+u}$ gives 
\begin{align*}
    M(t) &= \frac{1}{\pi}\int\limits_{0}^{1} v^{a-1}(1-v)^{-(a-1)} (1-v)\frac{1}{(1-v)^{2}}dv\\& = \frac{1}{\pi}\int\limits_{0}^{1}v^{a-1}(1-v)^{-(a-1)}(1-v)^{-1}dv \\ & = \frac{1}{\pi}\int\limits_{0}^{1}v^{a-1}(1-v)^{-a}dv\\ & = \frac{1}{\pi} \int\limits_{0}^{1} v^{a-1}(1-v)^{(1-a)-1}dv\\ & = \frac{1}{\pi}B(a,1-a) \\ & = \frac{1}{\pi}B(t+\frac{1}{2}, \frac{1}{2}-t)\\ & = \frac{1}{\pi}\frac{\Gamma(t+\frac{1}{2})\Gamma(\frac{1}{2}-t)}{\Gamma(1)}\\ & = \frac{1}{\pi} \Gamma(t+\frac{1}{2})\Gamma(\frac{1}{2}-t).
\end{align*}
With  $z = t+\frac{1}{2}$,
\begin{equation*}
    M(t) = \frac{1}{\pi}\Gamma(z)\Gamma(1-z) = \frac{1}{\pi}\frac{\pi}{\sin(\pi z)} = \frac{1}{\sin(\pi z)} 
\end{equation*}
by Euler's reflection formula. Thus, $M(t)$, has the convenient form
\begin{equation*}
    M(t) = \frac{1}{\sin(\pi(t+\frac{1}{2}))} =\frac{1}{\sin(\pi t + \frac{\pi}{2})} = \frac{1}{\cos(\pi t)} =\sec(\pi t). 
\end{equation*}
Consider the generic process for $t \geq 1$, 
\begin{equation*}
    x_{t} = \mu_{0} + \phi_{0}(x_{t-1} - \mu_{0}) + \epsilon_{t},
\end{equation*}
where we know $|\phi_{0}| <1$ and $\epsilon_{t} \overset{\text{i.i.d}}{\sim} Z(1/2,1/2,0,1).$ Under the DSP initialization, $x_{0} = \mu_{0} + \epsilon_{0}$. Thus, by iterating the recursion, for $t \geq0$, 
\begin{equation*}
    x_{t}-\mu_{0} = \sum\limits_{j=0}^{t}\phi_{0}^{j}\epsilon_{t-j}. 
\end{equation*}
Fix $s \in (0,1/2)$ and define 
\begin{equation*}
    K_{\phi_{0}}(s):= \prod\limits_{j=0}^{\infty}\sec(\pi \phi_{0}^{j}s).
\end{equation*}
We show that $K_{\phi_{0}}(s) < \infty$. Let $a := \pi s < \pi/2$. Note for every $|x| \leq a$, 
\begin{equation*}
    \log\{ \sec(x) \} = \int\limits_{0}^{|x|} \tan(u)du.
\end{equation*}
Moreover, for $0 \leq u \leq a,$
\begin{equation*}
    \tan(u) = \int\limits_{0}^{u} \sec^{2}(v)dv \leq u\sec^{2}(a), \text{ since } \sec^{2}(v) \leq \sec^{2}(a).
\end{equation*}
Hence 
\begin{equation*}
    0 \leq \log(\sec(x)) \leq \frac{x^{2}}{2}\sec^{2}(a). 
\end{equation*}
Since $|\pi \phi_{0}^{j}s| \leq \pi s = a$, it follows that 
\begin{align*}
    0 \leq \sum\limits_{j=0}^{\infty} \log(\sec(\pi \phi_{0}^{j}s)) \leq \frac{\pi^{2} s^{2}\sec^{2}(\pi s)}{2} \sum\limits_{j=0}^{\infty}\phi_{0}^{2j} = \frac{\pi^{2}s^{2}\sec^{2}(\pi s)}{2(1-\phi_{0}^{2})}<\infty.  
\end{align*}
Therefore,
\begin{equation*}
    K_{\phi_{0}}(s) = \exp \{ \sum\limits_{j=0}^{\infty} \log(\sec(\pi \phi_{0}^{j}s)) \}<\infty. 
\end{equation*}
Then, since $|s\phi_{0}^{j}| < 1/2$, for every $j \geq 0$, and independence of the innovations, we obtain 
\begin{align*}
    \mathbb{E}_{\boldsymbol{\theta}_{0}}[\exp\{ s(x_{t}-\mu_{0}) \}] &= \prod\limits_{j=0}^{t} \mathbb{E}_{\boldsymbol{\theta}_{0}}[\exp\{ s\phi_{0}^{j}\epsilon_{t-j} \}]\\
    &= \prod\limits_{j=0}^{t} \sec(\pi \phi_{0}^{j}s)\\
    &\leq K_{\phi_{0}}(s). 
\end{align*}
Similarly, since $\sec(\cdot)$ is an even function, 
\begin{equation*}
    \mathbb{E}_{\boldsymbol{\theta}_{0}}[\exp\{ -s(x_{t} - \mu_{0}) \}] = \prod\limits_{j=0}^{t} \sec(-\pi \phi_{0}^{j}s) \leq K_{\phi_{0}}(s). 
\end{equation*}
Therefore, by Chernoff's inequality,for every $y>0$,
\begin{equation*}
    \mathbb{P}_{\boldsymbol{\theta}_{0}}(x_{t} - \mu_{0}>y) \leq K_{\phi_{0}}(s) e^{-sy}
\end{equation*}
and
\begin{equation*}
     \mathbb{P}_{\boldsymbol{\theta}_{0}}(x_{t}-\mu_{0}<-y) \leq K_{\phi_{0}}(s) e^{-sy}.
\end{equation*}
Consequently, for $t \geq 0$, 
\begin{equation*}
    \mathbb{P}_{\boldsymbol{\theta}_{0}}(|x_{t}-\mu_{0}| >y) \leq 2K_{\phi_{0}}(s) e^{-sy}. 
\end{equation*}
For each $i = 1, \dots, N$, define 
\begin{equation*}
    K_{\alpha,i}(s) := \prod\limits_{j=0}^{\infty} \sec(\pi \phi_{\alpha_{i},0}^{j}s) < \infty. 
\end{equation*}
For each $i$, define 
\begin{equation*}
    y_{n}^{(\alpha_{i})} := \frac{c\log(\bar{n}) + \log\{ 2K_{\alpha_{i}}(s)\}}{s}.
\end{equation*}
Then, by the inequality above, 
\begin{equation*}
    \mathbb{P}_{\boldsymbol{\theta}_{0}}(|h_{\alpha_{i},t}-\mu_{\alpha_{i},0}|> y_{n}^{(\alpha_{i})}) \leq 2K_{\alpha_{i}}(s) e^{-sy_{n}^{(\alpha_{i})}} = \bar{n}^{-c}. 
\end{equation*}
Hence, 
\begin{equation*}
    \mathbb{P}_{\boldsymbol{\theta}_{0}}(|h_{\alpha_{i},t}|>|\mu_{\alpha_{i},0}| + y_{n}^{(\alpha_{i})}) \leq \bar{n}^{-c}.
\end{equation*}
Therefore, 
\begin{align*}
\mathbb{P}_{\boldsymbol{\theta}_{0}}(A_{\alpha,n}^{c}) &\leq \sum\limits_{i=1}^{N}\sum\limits_{t=1}^{n-1}\mathbb{P}_{\boldsymbol{\theta}_{0}}(|h_{\alpha_{i},t-1}|> |\mu_{\alpha_{i},0}|+y_{n}^{(\alpha_{i})})\\
    &\leq N(n-1)\bar{n}^{-c}\\
    &\rightarrow 0, \text{ since } c>1. 
\end{align*}
Now we will derive a bound on the $\beta-$process. 
For each $i = 1, \dots, N$, define 
\begin{equation*}
    K_{\beta_{i}}(s) := \prod\limits_{j=0}^{\infty} \sec(\pi \phi_{\beta_{i},0}^{j}s) < \infty. 
\end{equation*}
Define the auxiliary event 
\begin{equation*}
    G_{n}:= \{ \max\limits_{1 \leq i \leq N, 1 \leq t \leq n}|h_{\beta_{i},t-1}| \leq |\mu_{\beta_{i},0}| + \frac{c\log(\bar{n}) + \log(2K_{\beta_{i}}(s))}{s}\}.
\end{equation*}
By a similar argument as above,  for every $i$ and every $t\geq0$, 
\begin{equation*}
    \mathbb{P}_{\boldsymbol{\theta}_{0}}(|h_{\beta_{i},t}-\mu_{\beta_{i},0}| > \frac{c\log(\bar{n}) + \log(2K_{\beta_{i}}(s))}{s}) \leq \bar{n}^{-c}.
\end{equation*}
Consequently, 
\begin{equation*}
    \mathbb{P}_{\boldsymbol{\theta}_{0}}(|h_{\beta_{i},t}|> |\mu_{\beta_{i},0}| + \frac{c\log(\bar{n})+\log(2K_{\beta_{i}}(s))}{s}) \leq \bar{n}^{-c}. 
\end{equation*}
Taking the union bound over $i=1,\dots, N$ and $t = 1, \dots, n$ gives
\begin{equation*}
    \mathbb{P}_{\boldsymbol{\theta}_{0}}(G_{n}^{c}) \leq Nn\bar{n}^{-c} \leq N\bar{n}^{1-c}  \rightarrow 0, \text{ since } c>1. 
\end{equation*}
Now fix $ i \in \{ 1, \dots, N \}$ and $t \in \{ 1, \dots, n\}$. Under the model 
\begin{equation*}
    \beta_{i,t} = \beta_{i,0} + \sum\limits_{u=0}^{t-1} \omega_{\beta_{i},u}, 
\end{equation*}
and conditional on the path $\{ h_{\beta_{i,u}} \}_{u=0}^{t-1}$ and $\omega_{\beta_{i},u} \sim N(0, e^{h_{\beta_{i,u}}})$ independently over $u$. 
Therefore, conditional on $\{ h_{\beta_{i},u} \}_{u=0}^{t-1}, \beta_{i,t} - \beta_{i,0} \sim N(0,V_{i,t}), V_{i,t} :=\sum\limits_{u=0}^{t-1} e^{h_{\beta_{i},u}} $. 
On the event $G_{n}$
\begin{equation*}
    h_{\beta_{i,u}} \leq |\mu_{\beta_{i,0}}| + \frac{c\log(\bar{n}) + \log(2K_{\beta_{i}}(s))}{s},
\end{equation*}
for all $u \leq n-1$. Hence, 
\begin{align*}
    V_{i,t} &\leq n \exp \{ |\mu_{\beta_{i,0}}| + \frac{c\log(\bar{n}) + \log(2K_{\beta_{i}}(s))}{s} \}\\ & = \exp\{ |\mu_{\beta_{i},0}| \} (2K_{\beta_{i}}(s))^{\frac{1}{s}} n\bar{n}^{\frac{c}{s}}\\ 
    & \leq \frac{C_{4}}{8}\bar{n}^{\alpha} \text{ on } G_{n},  
\end{align*}
where $\alpha = 1+c/s$ and $C_{4}>0$  is some constant such that
\begin{equation*}
    \max\limits_{1\leq i \leq N}\exp\{ |\mu_{\beta_{i},0}| \} (2K_{\beta_{i}}(s))^{\frac{1}{s}} \leq \frac{C_{4}}{8}.
\end{equation*}
Now define $M_{n} : = \sqrt{(\kappa + 1) C_{4}\bar{n}^{\alpha}\log(\bar{n})}$. Since $M_{n} \rightarrow \infty$ and $\beta_{i,0}$ is fixed, there exists $n_{0}$ such that, for all $n \geq n_{0}, M_{n} \geq 2 \max \limits_{1 \leq i \leq N} |\beta_{i,0}|$. Observe, 
\begin{equation*}
    |\beta_{i,t}| = |(\beta_{i,t} - \beta_{i,0}) + \beta_{i,0}| \leq |\beta_{i,t} - \beta_{i,0}| + |\beta_{i,0}|  \leq \frac{M_{n}}{2} + \frac{M_{n}}{2} = M_{n}. 
\end{equation*}
So, 
\begin{equation*}
    |\beta_{i,t}|> M_{n} \text{ and } M_{n} \geq 2 |\beta_{i,0}| \implies |\beta_{i,t} - \beta_{i,0}| > \frac{M_{n}}{2}. 
\end{equation*}
Hence, for $n \geq n_{0}$ on $G_{n}$, 
\begin{equation*}
    \{ |\beta_{i,t}| > M_{n} \} \subset \{ |\beta_{i,t} - \beta_{i,0}| > \frac{M_{n}}{2} \}. 
\end{equation*}
Now define
\begin{equation*}
    \mathcal{H}_{n} = \sigma(h_{\beta_{\ell},u}: 1 \leq \ell \leq N, 0 \leq u \leq n-1).
\end{equation*}
Then $G_{n} \in \mathcal{H}_{n}$. Moreover, conditional on $\mathcal{H}_{n}$, 
\begin{equation*}
    \beta_{i,t} - \beta_{i,0} \sim N(0, V_{i,t}), \qquad V_{i,t} := \sum\limits_{u=0}^{t-1} e^{h_{\beta_{i},u}}.
\end{equation*}
As shown above, on $G_{n}$, 
\begin{equation*}
    V_{i,t} \leq \frac{C_{4}}{8}\bar{n}^{\alpha}. 
\end{equation*}
Therefore, for $n \geq n_{0}$, 
\begin{align*}
    \mathbb{P}_{\boldsymbol{\theta}_{0}}(\{ |\beta_{i,t}| > M_{n}\} \cap G_{n}) &\leq \mathbb{P}_{\boldsymbol{\theta}_{0}}(\{ |\beta_{i,t} - \beta_{i,0}| > \frac{M_{n}}{2} \} \cap G_{n})\\
    & = \mathbb{E}_{\boldsymbol{\theta}_{0}}[\mathbb{1}_{G_{n}}\mathbb{P}_{\boldsymbol{\theta}_{0}}(|\beta_{i,t} - \beta_{i,0}| > \frac{M_{n}}{2}|\mathcal{H}_{n})]\\
    &\leq 2\exp \{ - \frac{(\frac{M_{n}}{2})^{2}}{2(\frac{C_{4}\bar{n}^{\alpha}}{8})} \}\\
    &= 2\exp \{ - \frac{M_{n}^{2}}{C_{4}\bar{n}^{\alpha}} \} \\
    &= 2\exp\{ - \frac{(\kappa + 1)C_{4}\bar{n}^{\alpha}\log(\bar{n})}{C_{4}\bar{n}^{\alpha}}\} \\
    &= 2\bar{n}^{-(\kappa +1)}.
\end{align*}
Then by the union bound, for $n \geq n_{0}$, 
\begin{align*}
    \mathbb{P}_{\boldsymbol{\theta}_{0}}(A_{\beta,n}^{c}) &= \mathbb{P}_{\boldsymbol{\theta}_{0}}(A_{\beta,n}^{c} \cap G_{n}) + \mathbb{P}_{\boldsymbol{\theta}_{0}}(A_{\beta,n}^{c} \cap G_{n}^{c})\\
    &\leq \mathbb{P}_{\boldsymbol{\theta}_{0}}(A_{\beta,n}^{c} \cap G_{n}) + \mathbb{P}_{\boldsymbol{\theta}_{0}}(G_{n}^{c})\\
    &= \mathbb{P}_{\boldsymbol{\theta}_{0}}(\{ \max\limits_{1 \leq i \leq N, 1 \leq t \leq n} |\beta_{i,t-1}|> M_{n} \} \cap G_{n}) + \mathbb{P}_{\boldsymbol{\theta}_{0}}(G_{n}^{c})\\
    & \leq \sum\limits_{i=1}^{N}\sum\limits_{t=1}^{n} \mathbb{P}_{\boldsymbol{\theta}_{0}}(\{ |\beta_{i,t-1}| > M_{n} \}  \cap G_{n}) + \mathbb{P}_{\boldsymbol{\theta}_{0}}(G_{n}^{c})\\
    & \leq 2Nn\bar{n}^{-(\kappa+1)}+ \mathbb{P}_{\boldsymbol{\theta}_{0}} (G_{n}^{c})\\
    &\rightarrow 0. 
\end{align*}
\end{proof}
\subsection{Proof of Theorem 5}
\begin{proof}
    By Theorem 3 of \cite{ghosal2007convergence}
    \begin{equation*}
        \sup\limits_{\boldsymbol{s}_{n} \in H_{n}} Q_{\boldsymbol{\theta}_{0}, \boldsymbol{s}_{n}}^{(n)} \Pi_{n}(\boldsymbol{\theta} \in \Theta_{n}:d_{n,\boldsymbol{s}_{n}}(\boldsymbol{\theta}, \boldsymbol{\theta}_{0}) \geq M_{n} \epsilon_{n}|\boldsymbol{r}_{1:n}, r_{M,1:n}) \rightarrow 0,  
    \end{equation*}
where in the notation of \cite{ghosal2007convergence} $Pf$ denotes the expectation of $f$ with respect to $P$.Let 
\begin{equation*}
    Z_{n}(\boldsymbol{s}_{n}) := \Pi_{n}(\boldsymbol{\theta} \in \Theta_{n}:d_{n,\boldsymbol{s}_{n}}(\boldsymbol{\theta}, \boldsymbol{\theta}_{0}) \geq M_{n} \epsilon_{n}|\boldsymbol{r}_{1:n}, r_{M,1:n}). 
\end{equation*}
Notice $0 \leq Z_{n}(\boldsymbol{s}_{n}) \leq 1$. Now fix $\delta > 0$. Since $Z_{n}(\boldsymbol{s}_{n}) \geq 0$, Markov's inequality gives 
\begin{equation*}
    Q_{\boldsymbol{\theta}_{0},\boldsymbol{s}_{n}}^{(n)} (Z_{n}(\boldsymbol{s}_{n})>\delta) \leq \frac{1}{\delta}Q_{\boldsymbol{\theta}_{0}, \boldsymbol{s}_{n}}^{(n)} Z_{n}(\boldsymbol{s}_{n}). 
\end{equation*}
That is 
\begin{align*}
    &Q_{\boldsymbol{\theta}_{0}, \boldsymbol{s}_{n}}^{(n)}(\Pi_{n}(\boldsymbol{\theta}\in \Theta_{n}:d_{n,\boldsymbol{s}_{n}}(\boldsymbol{\theta}, \boldsymbol{\theta}_{0})\geq M_{n}\epsilon_{n}|\boldsymbol{r}_{1:n}, r_{M,1:n})>\delta)\\ &\leq \frac{1}{\delta}Q_{\boldsymbol{\theta}_{0}, \boldsymbol{s}_{n}}^{(n)}\Pi_{n}(\boldsymbol{\theta}\in \Theta_{n}:d_{n,\boldsymbol{s}_{n}}(\boldsymbol{\theta}, \boldsymbol{\theta}_{0})\geq M_{n}\epsilon_{n}|\boldsymbol{r}_{1:n}, r_{M,1:n}) \rightarrow 0. 
\end{align*}
This implies 
\begin{equation*}
    \sup\limits_{\boldsymbol{s}_{n} \in H_{n}} Q_{\boldsymbol{\theta}_{0},\boldsymbol{s}_{n}}^{(n)} (\Pi_{n}(\boldsymbol{\theta} \in \Theta_{n}: d_{n,\boldsymbol{s}_{n}}(\boldsymbol{\theta}, \boldsymbol{\theta}_{0})\geq M_{n}\epsilon_{n}|\boldsymbol{r}_{1:n}, r_{M,1:n})>\delta) \rightarrow 0.
\end{equation*}
For $\boldsymbol{s}_{n} \in H_{n}$, define 
\begin{equation*}
    A_{n}(\boldsymbol{s}_{n}): = \{ r_{1:n}: \Pi_{n}(\boldsymbol{\theta} \in \Theta_{n}:d_{n, \boldsymbol{s}_{n}}(\boldsymbol{\theta}, \boldsymbol{\theta}_{0}) \geq M_{n}\epsilon_{n}|\boldsymbol{r}_{1:n}, r_{M,1:n}) > \delta \}. 
\end{equation*}
Then, 
\begin{equation*}
    \sup\limits_{\boldsymbol{s}_{n} \in H_{n}}Q_{\boldsymbol{\theta}_{0},\boldsymbol{s}_{n}}^{(n)}(A_{n}(\boldsymbol{s}_{n})) \rightarrow 0. 
\end{equation*}
Similarly, define 
\begin{equation*}
    A_{n}(\boldsymbol{S}_{0,n}) := \{ r_{1:n}: \Pi_{n}(\boldsymbol{\theta} \in \Theta_{n}: d_{n,\boldsymbol{S}_{0,n}}(\boldsymbol{\theta},\boldsymbol{\theta}_{0}) \geq M_{n}\epsilon_{n}|\boldsymbol{r}_{1:n}, r_{M,1:n})>\delta \}. 
\end{equation*}
So, 
\begin{equation*}
    \mathbb{P}_{\boldsymbol{\theta}_{0}}(\Pi_{n}(\boldsymbol{\theta} \in \Theta_{n}: d_{n,\boldsymbol{S}_{0,n}}(\boldsymbol{\theta}, \boldsymbol{\theta}_{0}) \geq M_{n}\epsilon_{n}|\boldsymbol{r}_{1:n}, r_{M,1:n})> \delta) = \mathbb{P}_{\boldsymbol{\theta}_{0}}(\boldsymbol{r}_{1:n} \in A_{n}(\boldsymbol{S}_{0,n})). 
\end{equation*}
By the law of total expectation 
\begin{align*}
    \mathbb{P}_{\boldsymbol{\theta}_{0}}(\boldsymbol{r}_{1:n}\in A_{n}(\boldsymbol{S}_{0,n})) &= \mathbb{E}_{\boldsymbol{\theta}_{0}}[\mathbb{E}_{\boldsymbol{\theta}_{0}} (\boldsymbol{1}\{ \boldsymbol{r}_{1:n} \in A_{n}(\boldsymbol{S}_{0,n}) \}|\boldsymbol{S}_{0,n}, r_{M,1:n})]\\ 
    & = \mathbb{E}_{\boldsymbol{\theta}_{0}} [Q_{\boldsymbol{\theta}_{0},\boldsymbol{S}_{0,n}}^{(n)}(A_{n}(\boldsymbol{S}_{0,n}))]\\
    &= \mathbb{E}_{\boldsymbol{\theta}_{0}} [Q_{\boldsymbol{\theta}_{0}, \boldsymbol{S}_{0,n}}^{(n)} (A_{n}(\boldsymbol{S}_{0,n})) \boldsymbol{1} \{ \boldsymbol{S}_{0,n} \in H_{n} \} ] + \mathbb{E}_{\boldsymbol{\theta}_{0}}[Q_{\boldsymbol{\theta}_{0}, \boldsymbol{S}_{0,n}}^{(n)} (A_{n}(\boldsymbol{S}_{0,n}))\boldsymbol{1}\{ \boldsymbol{S}_{0,n} \notin H_{n} \}]. 
\end{align*}
Note 
\begin{equation*}
    \mathbb{E}_{\boldsymbol{\theta}_{0}}[Q_{\boldsymbol{\theta}_{0}, \boldsymbol{S}_{0,n}}^{(n)} (A_{n}(\boldsymbol{S}_{0,n})) \boldsymbol{1}\{ \boldsymbol{S}_{0,n} \in H_{n} \}] \leq \sup\limits_{\boldsymbol{s}_{n} \in H_{n}} Q_{\boldsymbol{\theta}_{0}, \boldsymbol{s}_{n}}^{(n)} (A_{n}(\boldsymbol{s}_{n})) \rightarrow 0. 
\end{equation*}
Similarly by the Proposition 4, 
\begin{equation*}
    \mathbb{E}_{\boldsymbol{\theta}_{0}}[Q_{\boldsymbol{\theta}_{0},\boldsymbol{S}_{0,n}}^{(n)}(A_{n}(\boldsymbol{S}_{0,n})) \boldsymbol{1} \{ \boldsymbol{S}_{0,n} \notin H_{n} \}] \leq \mathbb{P}_{\boldsymbol{\theta}_{0}}(H_{n}^{c}) \rightarrow 0. 
\end{equation*}
Therefore 
\begin{equation*}
    \mathbb{P}_{\boldsymbol{\theta}_{0}}(\Pi_{n}(\boldsymbol{\theta} \in \Theta_{n}: d_{n,\boldsymbol{S}_{0,n}}(\boldsymbol{\theta}, \boldsymbol{\theta}_{0}) \geq M_{n}\epsilon_{n} | \boldsymbol{r}_{1:n}, r_{M,1:n}) > \delta) \rightarrow 0. 
\end{equation*}
\end{proof}
\section{MCMC details}

For the stochastic volatility processes, we use the
ancillarity--sufficiency interweaving strategy, which interweaves the
centered and non-centered parameterizations to improve the mixing of
the stochastic volatility parameters. Conditional on the latent factors
and volatilities, the rows of the loading matrix are updated by Bayesian
regression, with deep interweaving improving sampling efficiency; we use
the largest absolute loading in each factor column for the
deep-interweaving transformation.

All of our code is written in the R programming language (\cite{R_language}). We provide a sketch algorithm of the key components of our MCMC sampler. For the dynamic shrinkage components and the multivariate stochastic volatility component we leverage the code from \cite{kowal2019dynamic} and \cite{factorstochvol}. The full sampling algorithm is provided below. We use our sampling algorithm to obtain $n_{\text{total}}$ from the joint posterior distribution for the correlation matrix time series of length $T$. 
\begin{itemize}
    \item  Draw $\{ \text{Var}(r_{M,1})^{(k)}, \dots, \text{Var}(r_{M,T})^{(k)} \}_{k=1}^{n_{\text{total}}}$ from the joint posterior distribution of the $SV(1)$ process. For details of the sampling algorithm see \cite{kastner2016dealing}. 
    \item $\{ \sigma_{\epsilon,a,t}^{(0)} \}_{1 \leq a \leq N}^{1 \leq t \leq T} = \{ \frac{1}{T-1} \sum\limits_{t=1}^{T} (r_{a,t}-\frac{1}{T}\sum\limits_{i=1}^{T}r_{a,t})^{2} \}_{a=1}^{N}$
    \item  \begin{align*}
        Q_{\beta}^{(0)}& = X^\top \text{diag}(\sigma_{1}^{-2},\dots, \sigma_{T}^{-2})X + (L^\top\otimes I_{2}) \text{diag}(\sigma_{\omega,\alpha,1}^{-2},\sigma_{\omega,\beta,1}^{-2}, \dots,\sigma_{\omega,\alpha,T}^{-2},\sigma_{\omega,\beta,T}^{-2}   ) (L\otimes I_{2}), \\ 
        &\sigma_{1}^{2}, \dots, \sigma_{T}^{2},\sigma_{\omega,\alpha,1}^{2},\sigma_{\omega,\beta,1}^{2}, \dots,\sigma_{\omega,\alpha,T}^{2},\sigma_{\omega,\beta,T}^{2}    \sim |N(0,1)|,\\ 
        &L = \begin{pmatrix}
            1 & 0 & 0 & \dots &0 \\ 
            -1 & 1& 0 & \dots &0\\ 
            0&-1&1 &\dots&0\\ 
            \vdots & \vdots & \vdots& \vdots & \vdots\\ 
            0&\dots&0&-1&1
        \end{pmatrix}\\
        & X = \text{blockdiag}(x_{1}^\top, \dots, x_{T}^\top). 
    \end{align*}
    \item For $1 \leq a \leq N: $
\begin{align*}
    &\{ (\alpha_{a,t}^{(0)}, \beta_{a,t}^{(0)})\}_{t=1}^{T} |\{\sigma_{\epsilon_{a,t}}^{(0)} \} _{t=1}^{T},X,\{r_{a,t} \}_{t=1}^{N}, \{r_{a,t} \}_{t=1}^{T} \sim N((Q_{post,a}^{(0)})^{-1}b_{a}^{(0)}, (Q_{post,a}^{(0)})^{-1}) \\ 
    &Q_{obs,a}^{(0)} = X^\top \text{diag}((\sigma_{\epsilon_{a,1}}^{(0)})^{-2}, \dots , (\sigma_{\epsilon_{a,T}}^{(0)})^{-2})X, b^{(0)}_{a} = X^\top \text{diag}((\sigma_{\epsilon_{a,1}}^{(0)})^{-2}, \dots , (\sigma_{\epsilon_{a,T}}^{(0)})^{-2})(r_{a,1}, \dots, r_{a,T})^\top,\\ 
    &D_{a,t}^{(0)} = \text{diag}((\sigma_{\omega,\alpha_{a,t}}^{(0)})^{-2}, (\sigma_{\omega,\beta_{a,t}}^{(0)})^{-2}), \sigma_{\omega,\alpha_{a,t}}^{(0)} =\sigma_{\omega,\beta_{a,t}}^{(0)} =  0.1\sigma_{\epsilon_{a,t}}^{(0)}, \\ 
    &Q_{\text{evol},a}^{(0)} =
\begin{pmatrix}
D_{a,1}^{(0)} + D_{a,2}^{(0)} & -D_{a,2}^{(0)}      & 0         & \cdots    & 0 \\
-D_{a,2}^{(0)}      & D_{a,2}^{(0)} + D_{a,3}^{(0)} & -D_{a,3}^{(0)}       & \ddots    & \vdots \\
0         & -D_{a,3}^{(0)}      & D_{a,3}^{(0)} + D_{a,4}^{(0)}  & \ddots    & 0 \\
\vdots    & \ddots    & \ddots     & \ddots    & -D_{a,T}^{(0)} \\
0         & \cdots    & 0          & -D_{a,T}^{(0)}      & D_{a,T}^{(0)}
\end{pmatrix},  Q_{post,a}^{(0)} = Q_{obs,a}^{(0)} + Q_{evol,a}^{(0)}\\ 
\end{align*}
\item For $1 \leq a \leq N$:
\begin{align*}
    \text{compute } \Omega_{a}^{(0)} &= D_{1}B_{i}^{(0)}\\ 
    D_{1} = \begin{pmatrix}
    -1 & 1 & 0 & \dots & 0 \\ 
    0& -1 & 1 & \dots & 0 \\ 
    \vdots & \vdots & \vdots & \vdots & \vdots\\ 
    0 & \dots &0 & -1&1
\end{pmatrix} &\text{ and }  B_{a}^{(0)} = \begin{pmatrix}
 \alpha_{a,1}^{(0)} & \beta_{a,1}^{(0)} \\ 
 \alpha_{a,2}^{(0)} & \beta_{a,2}^{(0)} \\ 
 \vdots & \vdots \\ 
 \alpha_{a,T}^{(0)} & \beta_{a,T}^{(0)}
\end{pmatrix}.
\end{align*}
\item  For $ 1 \leq a \leq n$: 
\begin{align*}
    &E_{a}^{(0)} = \{ \{ \sigma_{\omega,\alpha_{a,t}}^{(0)},\sigma_{\omega,\beta_{a,t}}^{(0)} \} _{t=2}^{T}, \{ h_{\alpha_{a,t}}^{(0)},h_{\beta_{a,t}}^{(0)} \}_{t=2}^{T},   \phi_{\alpha_{a}}^{(0)}, \phi_{\beta_{a}}^{(0)}, \mu_{\alpha_{a}}^{(0)}, \mu_{\beta_{a}}^{(0)}  , \{ \sigma_{\eta, \alpha_{a,t}}^{(0)}, \sigma_{\eta,\beta_{a},t}^{(0)} \}_{t=3}^{T}, \sigma_{\eta, \alpha_{a,0}}^{(0)}, \sigma_{\eta, \beta_{a,0}}^{(0)}, \mu_{a,0}^{(0)} \} \\ 
    & \omega_{\alpha_{a,t}}^{(0)} = \alpha_{a,t}^{(0)} - \alpha_{a,t-1}^{(0)}, \omega_{\beta_{a,t}}^{(0)} = \beta_{a,t}^{(0)} - \beta_{a,t-1}^{(0)}, h_{\alpha_{a,t}}^{(0)} = \log((\omega_{\alpha_{a,t}}^{(0)})^{2} + 10^{-4}), h_{\beta_{a,t}}^{(0)} = \log((\omega_{\beta_{a,t}}^{(0)})^{2} + 10^{-4}), \\ 
    & \sigma_{\omega, \alpha_{a,t}}^{(0)} = \exp \{\frac{h_{\alpha_{a,t}}^{(0)}}{2} \}, \sigma_{\omega, \beta_{a,t}}^{(0)} = \exp \{\frac{h_{\beta_{a,t}}^{(0)}}{2} \},t = 2,\dots , T. \\ 
    &\text{Fit an } AR(1) \text{ for } \{h_{j_{a,t}}^{(0)}\}_{t=2}^{T},
    h_{j_{a,t}}^{(0)} -\mu_{j_{a}}^{(0)} = \phi_{j_{a}}^{(0)} (h_{j_{a,t-1}}^{(0)}-\mu_{j_{a}}^{(0)}) + \eta_{j_{a,t}},j \in \{ \alpha, \beta \}, \phi_{j_{a}}^{(0)} := \hat{\phi_{j_{a}}} , \mu_{j_{a}}^{(0)}:= \hat{\mu}_{j_{a}}. 
 \\
    & \text{For } t = 3\dots T, \sigma_{\eta,\alpha_{a,t}}^{(0)} = \sigma_{\eta,\beta_{a,t}}^{(0)} =  \sigma_{\eta,\alpha_{a,0}}^{(0)} = \sigma_{\eta,\beta_{a,0}}^{(0)} = \pi, \mu_{a,0}^{(0)} = \frac{1}{2}(\mu_{\alpha_{a}}^{(0)} + \mu_{\beta_{a}}^{(0)}). 
\end{align*}
\item For $1\leq a \leq N, E_{0,a}^{(0)} = \{ \sigma_{\omega_{0}, \alpha_{a}}^{(0)},\sigma_{\omega_{0}, \beta_{a}}^{(0)}, \rho_{\omega_{0}, \alpha_{a}}^{(0)},\rho_{\omega_{0}, \beta_{a}}^{(0)},\sigma_{00,a}^{(0)}, \rho_{00,a}^{(0)} \},  \sigma_{\omega_{0}, \alpha_{a}}^{(0)} = |\alpha_{a,1}^{(0)}|, \sigma_{\omega_{0}, \beta_{a}}^{(0)} = |\beta_{a,1}^{(0)}|, \rho_{\omega_{0}, \alpha_{a}}^{(0)} = \rho_{\omega_{0}, \beta_{a}}^{(0)} =\sigma_{00,a}^{(0)} = \rho_{00,a} ^{(0)} =  1.$
\item Draw $\{\Sigma_{\epsilon,t}^{(0)} \}_{t=1}^{T}$ from the joint posterior distribution of the time-varying observation error covariance matrices of $[\boldsymbol{r}_{1,t}, \dots, \boldsymbol{r}_{N,t}] - [\boldsymbol{\alpha}_{1,t}^{(0)}, \dots, \boldsymbol{\alpha}_{N,t}^{(0)}] - r_{M,t}[\boldsymbol{\beta}_{1,t}^{(0)}, \dots, \boldsymbol{\beta}_{N,t}^{(0)}]\}_{t=1}^{T}$. Then we redefine the initial values of the observation error variances, particularly $\{ \sigma_{a,t}^{(0)} \}_{1 \leq t \leq T}^{1 \leq a \leq N} = \{[\Sigma_{\epsilon,t}^{(0)}]_{a,a} \}_{a=1}^{N}$. For more details on the sampling algorithm see \cite{factorstochvol}.
\end{itemize}
For $1 \leq k \leq n_{total}:$
\begin{itemize}
    \item For $1 \leq a \leq N$, \begin{align*}
& \{ (\alpha_{a,t}^{(k+1)}, \beta_{a,t}^{(k+1)})\}_{t=1}^{T} |\{\sigma_{\epsilon_{a,t}}^{(k)} \} _{t=1}^{T}, X, \{ r_{a,t}\}_{t=1}^{T}, \{ (\sigma_{\omega,\alpha_{a,t}}^{(k)}, \sigma_{\omega,\beta_{a,t}}^{(k)}) \}_{t=1}^{T}\sim N((Q_{post,a}^{(k)})^{-1}b_{a}^{(k)}, (Q_{post,a}^{(k)})^{-1}) \\
    &Q_{obs,a}^{(k)} = X^\top \text{diag}((\sigma_{\epsilon_{a,1}}^{(k)})^{-2}, \dots , (\sigma_{\epsilon_{a,T}}^{(k)})^{-2})X, b^{(k)}_{a} = X^\top \text{diag}((\sigma_{\epsilon_{a,1}}^{(k)})^{-2}, \dots , (\sigma_{\epsilon_{a,T}}^{(k)})^{-2})(r_{a,1}, \dots, r_{a,T})^\top,\\ 
    &D_{a,t}^{(k)} = \text{diag}((\sigma_{\omega,\alpha_{a,t}}^{(k)})^{-2}, (\sigma_{\omega,\beta_{a,t}}^{(k)})^{-2}) \\ 
    &Q_{\text{evol},a}^{(k)} =
\begin{pmatrix}
D_{a,1}^{(k)} + D_{a,2}^{(k)} & -D_{a,2}^{(k)}      & 0         & \cdots    & 0 \\
-D_{a,2}^{(k)}      & D_{a,2}^{(k)} + D_{a,3}^{(k)} & -D_{a,3}^{(k)}       & \ddots    & \vdots \\
0         & -D_{a,3}^{(k)}      & D_{a,3}^{(k)} + D_{a,4}^{(k)}  & \ddots    & 0 \\
\vdots    & \ddots    & \ddots     & \ddots    & -D_{a,T}^{(k)} \\
0         & \cdots    & 0          & -D_{a,T}^{(k)}      & D_{a,T}^{(k)}
\end{pmatrix},  Q_{post,a}^{(k)} = Q_{obs,a}^{(k)} + Q_{evol,a}^{(k)}\\ 
&(\sigma_{\omega,\alpha_{a,t}}^{(k)})^{2},\,(\sigma_{\omega,\beta_{a,t}}^{(k)})^{2}
=
\begin{cases}
(\sigma_{\omega_0,\alpha_a}^{(k)})^{2},\,(\sigma_{\omega_0,\beta_a}^{(k)})^{2}, 
& t=1,\ldots,D,\\[4pt]
(\sigma_{\omega,\alpha_{a,t}}^{(k)})^{2},\,(\sigma_{\omega,\beta_{a,t}}^{(k)})^{2}\ , \text{ from } E_{a}^{(k)}, 
& t=D+1,\ldots,T.
\end{cases}
\end{align*}
\item For $1 \leq a \leq N, \Omega_{a}^{(k+1)} = D_{1} B_{a}^{(k+1)}. $
\item For $1 \leq a \leq N$
\begin{enumerate}
    \item $\delta_{a} = \boldsymbol{1} \{\min \{ (\alpha_{a,1}^{(k+1)})^{2},(\beta_{a,1}^{(k+1)})^{2}  \} < 10^{-16} \}\max(10^{-8}, 1.4826\cdot\frac{|\alpha_{a,1}^{(k+1)}-\beta_{a,1}^{(k+1)}|}{2\times 10^{6}})$
    \item $(\tilde{\alpha}_{a}^{(k+1)})^{2} =  (\alpha_{a,1}^{(k+1)})^{2} + \delta_{a} \text{ and } (\tilde{\beta}_{a}^{(k+1)})^{2} =  (\beta_{a,1}^{(k+1)})^{2} + \delta_{a}$ 
    \item $(\sigma_{\omega_{0}, \alpha_{a}}^{(k+1)})^{-2} \sim Ga(1, \frac{(\tilde{\alpha}_{a}^{(k+1)})^{2}}{2} + \kappa_{\omega_{0}, \alpha_{a}}^{(k)}) \text{ and } (\sigma_{\omega_{0}, \beta_{a}}^{(k+1)})^{-2} \sim Ga(1, \frac{(\tilde{\beta}_{a}^{(k+1)})^{2}}{2} + \kappa_{\omega_{0}, \beta_{a}}^{(k)}) $
    \item $\kappa_{\omega_{0}, \alpha_{a}}^{(k+1)} \sim Ga(1, \frac{1}{(\sigma_{\omega_{0}, \alpha_{a}}^{(k+1)})^{2}}+\frac{1}{(\sigma_{00,a}^{(k)})^{2}}) \text{ and } \kappa_{\omega_{0}, \beta_{a}}^{(k+1)} \sim Ga(1, \frac{1}{(\sigma_{\omega_{0}, \beta_{a}}^{(k+1)})^{2}}+\frac{1}{(\sigma_{00,a}^{(k)})^{2}})$
    \item $(\sigma_{00,a}^{(k+1)})^{-2} \sim Ga(\frac{3}{2}, \kappa_{\omega_{0}, \alpha_{a}}^{(k+1)} + \kappa_{\omega_{0}, \beta_{a}} ^{(k+1)}+ \kappa_{00,a}^{(k)})$
    \item $\kappa_{00,a}^{(k+1)} \sim Ga(1, \frac{1}{(\sigma_{00,a}^{(k+1)})^{2}}+1).$
    \item $E_{a,0} ^{(k+1)} = \{ \sigma_{\omega_{0}, \alpha_{a}} ^{k+1}, \sigma_{\omega_{0}, \beta_{a}} ^{k+1}, \kappa_{\omega_{0}, \alpha_{a}}^{(k+1)}, \kappa_{\omega_{0}, \beta_{a}}^{(k+1)}, \sigma_{00,a}^{(k+1)}, \kappa_{00,a}^{(k+1)} \}$
\end{enumerate}
\item $j \in \{ \alpha, \beta\}$
\item $\delta_{j_{a}}^{(k+1)} = \boldsymbol{1} \{ \exists t \in \{2, \dots, T \}: (\omega_{j_{a},t}^{(k+1)})^{2} < 10^{-16} \}\max\{ 10^{-8}, 10^{-6}MAD((\omega_{j_{a},2:T})^{(k+1)}) \}, y_{j_{a,t}}^{* (k+1)} = \log((\omega_{j_{a,t}}^{k+1})^{2} + \delta_{j_{a}}^{(k+1)}), t = 2, \dots, T.$
\item Introduce mixture indicators $J_{j_{a,t}}^{(k+1)} \in \{ 1, \dots, 10\}$ drawn from the 10 component Gaussian mixture (\cite{omori2007stochastic}) with the specific mixtures means $m_{j_{a,t}}^{(k+1)}$ and variances $v_{j_{a,t}}^{2(k+1)}$ from the paper. 
\item Let $j \in \{\alpha,\beta\}$ and set $n=T-1$. For notational convenience, re-index time by
$s=1,\dots,n$ corresponding to original times $t=s+1 \in \{2,\dots,T\}$.
Define
\[
y^{*(k+1)}_{j,a,s} = \log\!\big((\omega_{j,a,s+1}^{(k+1)})^2 + \delta_{j,a}^{(k+1)}\big), \qquad s=1,\dots,n.
\]
Given the (\cite{omori2007stochastic}) mixture indicators $J_{j,a,s}^{(k+1)}\in\{1,\dots,10\}$, let
\[
m_{j,a,s}^{(k+1)} = m_{J_{j,a,s}^{(k+1)}}, \qquad v^{2(k+1)}_{j,a,s} = v^2_{J_{j,a,s}^{(k+1)}}.
\]
Collect
\[
\boldsymbol{y}^{*(k+1)}_{j,a}=(y^*_{j,a,1},\dots,y^*_{j,a,n})',\quad
\boldsymbol{m}_{j,a}^{(k+1)}=(m_{j,a,1}^{(k+1)},\dots,m_{j,a,n}^{(k+1)})',\quad
\boldsymbol{v}^{2(k+1)}_{j,a}=(v^{2(k+1)}_{j,a,1},\dots,v^{2(k+1)}_{j,a,n})'.
\]

\item Define the centered log-volatility state
\[
\tilde h_{j,a,s}^{(k+1)} = h_{j,a,s+1}^{(k+1)} - \mu_{j,a}^{(k)}, \qquad s=1,\dots,n,
\]
and the linear term
\[
\ell_{j,a,s}^{(k+1)} = \frac{y^{*(k+1)}_{j,a,s} - m_{j,a,s}^{(k+1)} - \mu_{j,a}^{(k)}}{v^{2(k+1)}_{j,a,s}}, \qquad s=1,\dots,n.
\]
Let $\boldsymbol{\ell}_{j,a}^{(k+1)}=(\ell_{j,a,1}^{(k+1)},\dots,\ell_{j,a,n}^{(k+1)})'$.

\item Define the evolution precisions
\[
q_{j,a,1}^{(k)} = (\sigma_{\eta,j,a,0}^{(k)})^{-2}, \qquad
q_{j,a,s}^{(k)} = (\sigma_{\eta,j,a,s+1}^{(k)})^{-2}, \quad s=2,\dots,n.
\]
Then the precision matrix $Q_{j,a}^{(k+1)}\in\mathbb{R}^{n\times n}$ is tridiagonal with
\[
(Q_{j,a}^{(k+1)})_{s,s}=
\begin{cases}
\dfrac{1}{v^{2(k+1)}_{j,a,s}} + q_{j,a,s}^{(k)} + (\phi_{j,a}^{(k)})^2 q_{j,a,s+1}^{(k)}, & s=1,\dots,n-1,\\[8pt]
\dfrac{1}{v^{2(k+1)}_{j,a,n}} + q_{j,a,n}^{(k)}, & s=n,
\end{cases}
\]
and off-diagonal entries
\[
(Q_{j,a}^{(k+1)})_{s,s+1} = (Q_{j,a}^{(k+1)})_{s+1,s} = -\phi_{j,a}^{(k)}\, q_{j,a,s+1}^{(k)},
\qquad s=1,\dots,n-1.
\]

\item The centered state is sampled as
\[
\tilde{\boldsymbol h}_{j,a}^{(k+1)}\mid \cdots \sim
N\!\big((Q_{j,a}^{(k+1)})^{-1}\boldsymbol{\ell}_{j,a}^{(k+1)},\, (Q_{j,a}^{(k+1)})^{-1}\big),
\]
and hence the (uncentered) log-volatility vector is
\[
\boldsymbol h_{j,a}^{(k+1)}
=
\mu_{j,a}^{(k)}\mathbf{1}_n + \tilde{\boldsymbol h}_{j,a}^{k+1}
\sim
N\!\big(\mu_{j,a}^{(k)}\mathbf{1}_n + (Q_{j,a}^{(k+1)})^{-1}\boldsymbol{\ell}_{j,a}^{(k+1)},\, (Q_{j,a}^{(k+1)})^{-1}\big).
\]
\item Let $\tilde{h}_{j_{a,t}}^{(k+1)} = h_{j_{a,t}}^{(k+1)} - \mu_{j_{a}}^{(k)}$ and $\eta_{j_{a,t}}^{k+1} =\tilde{h}_{j_{a,t}}^{(k+1)} - \phi_{j_{a}}^{(k+1)}\tilde{h}_{j_{a,t-1}}^{(k+1)}, t = 3,\dots,T.$
\item If $\phi_{\alpha_{a}}^{(k)} = \phi_{\beta_{a}}^{(k)} = 0$ or $\phi_{\alpha_{a}}^{(k)} = \phi_{\beta_{a}}^{(k)}=1$ then do not resample the phi's. Otherwise for each $j \in \{ \alpha, \beta \}$, define $\phi_{j_{a}}^{(01,k+1)} = \frac{\phi_{j_{a}}^{(k+1)}+1}{2} \in (0,1).$ The draw is from the posterior, implemented by splice sampling on $(0,1)$ with \begin{equation*}
    \phi_{j_{a}}^{(01,k+1)} \propto \exp \{ - \frac{1}{2}\sum\limits_{t=3}^{T}(\frac{\tilde{h}_{j_{a,t}}^{(k+1)}-(2\phi_{j_{a}}^{(01)}-1)\tilde{h}_{j_{a,t-1}}^{(k+1)}}{\sigma_{\eta,j_{a},t}^{(k)}})^{2}\} \cdot Beta(\phi_{j_{a}}^{(01)};10,2), \phi_{j_{a}}^{(k+1)} = 2\phi_{j_{a}}^{(01,k+1)}-1. 
\end{equation*}
\item For $j \in \{ \alpha, \beta \}, (\sigma_{\eta,j_{a}, 0}^{(k+1)})^{-2} \sim PG(1, \tilde{h}_{j_{a},2}^{(k+1)})$. For $t = 3,\dots, T, (\sigma_{\eta,j_{a,t}}^{(k+1)}) ^{-2} \sim PG(1 , \eta_{j_{a,t}})$
\item Let $j \in \{ \alpha, \beta \} $ and $\Xi_{j_{a}}^{(k+1)} \sim PG (1, \mu_{{j,a}}^{(k)} - \mu_{0,{a}}^{(k)})$. Then $y_{1,j_{a}}^{(k)} = \frac{h_{j_{a},2}^{(k)}}{\sigma_{\eta,j_{a},0}^{(k+1)}}, x_{1,j_{a}}^{(k+1)} = \frac{1}{\sigma_{\eta,j_{a},0}^{(k+1)}}, y_{t,j_{a}}^{(k+1)} = \frac{h_{j_{a},t+1}^{(k+1)}- \phi_{j_{a}}^{(k+1)}h_{j_{a,t}}^{(k+1)}}{\sigma_{\eta,j_{a,t+1}}^{(k+1)}}, x_{t,j_{a}}^{(k+1)} = \frac{1-\phi_{j_{a}}^{(k+1)}}{\sigma_{\eta,j_{a,t+1}}^{(k+1)}}, t= 2,..., T-1$. Then $\mu_{j_{a}}^{(k+1)} \sim N(m_{j_{a}}^{(k+1)}, s_{j_{a}}^{2(k+1)}), s_{j_{a}}^{2(k+1)} = (\sum\limits_{t=1}^{T-1}(x_{t,j_{a}}^{(k+1)})^{2} + \Xi_{j_{a}}^{(k+1)})^{-1}, m_{j_{a}}^{(k+1)} = (\sum\limits_{t=1}^{T-1} x_{t,j_{a}}^{(k+1)}y_{t,j_{a}}^{(k+1)} + \mu_{0,{a}}^{(k)}\Xi_{j_{a}}^{(k+1)}) s_{j_{a}}^{2(k+1)}$.If $\phi_{\alpha_{a}}^{(k+1)} = \phi_{\beta_{a}}^{(k+1)} = 1$ then $\mu_{\alpha_{a}}^{(k+1)} = \mu_{\beta_{a}}^{(k+1)}=0$.
\item Let $\Xi_{\alpha_{a}}$ and $\Xi_{\beta_{a}}$ be the Polya-Gamma precisions drawn in the previous step. Then, $\mu_{0,{a}}^{(k+1)} \sim N(m_{0,{a}}^{(k+1)}, s_{0,{a}}^{2(k+1)}), s_{0,{a}}^{2(k+1)} = (\Xi_{\alpha_{a}}^{(k+1)} + \Xi_{\beta_{a}}^{(k+1)} + \Xi_{0_{a}}^{(k+1)})^{-1},$ and $m_{0,a}^{(k+1)} = s_{0,a}^{2(k+1)} (\Xi_{\alpha_{a}}^{(k+1)}\mu_{\alpha_{a}}^{(k+1)} + \Xi_{\beta_{a}}^{(k+1)} \mu_{\beta_{a}}^{(k+1)} + \Xi_{0,a}^{(k+1)} \log(\sigma_{e}^{2})), \Xi_{0,a}^{(k+1)} \sim PG(1,\mu_{0,a}^{(k)}-\log(\sigma_{e}^{2})).$ If $\phi_{\alpha_{a}}^{(k+1)} = \phi_{\beta_{a}}^{(k+1)} = 1$ then $\mu_{0,{a}}^{(k+1)} = 0.$
\item Draw $\{\Sigma_{\epsilon,t}^{(k+1)} \}_{t=1}^{T}$ from the joint posterior distribution of the time-varying observation error covariance matrices of $[\boldsymbol{r}_{1,t}, \dots, \boldsymbol{r}_{N,t}] - [\boldsymbol{\alpha}_{1,t}^{(k+1)}, \dots, \boldsymbol{\alpha}_{N,t}^{(k+1)}] - r_{M,t}[\boldsymbol{\beta}_{1,t}^{(k+1)}, \dots, \boldsymbol{\beta}_{N,t}^{(k+1)}]\}_{t=1}^{T}$. Then we redefine the initial values of the observation error variances, particularly $\{ \sigma_{a,t}^{(+1)} \}_{1 \leq t \leq T}^{1 \leq a \leq N} = \{[\Sigma_{\epsilon,t}^{(k+1)}]_{a,a} \}_{a=1}^{N}$. For more details on the sampling algorithm see \cite{factorstochvol}.
\item \begin{align*}
    &\{ \Sigma_{t}^{(k+1)} \}_{t=1}^{T} =  \{ \text{diag}([(Q_{post,1}^{(k+1)})^{-1}]_{2t-1,2t-1}, \dots, [(Q_{post,N}^{(k+1)})^{-1}]_{2t-1,2t-1})\\ & + \text{Var}(r_{M,t})^{(k+1)} ([(Q_{post,1}^{(k+1)})^{-1} b_{1}^{(k+1)}]_{2t}, \dots,  [(Q_{post,N}^{(k+1)})^{-1} b_{N}^{(k+1)}]_{2t})^\top\\ &([(Q_{post,1}^{(k+1)})^{-1} b_{1}^{(k+1)}]_{2t}, \dots,  [(Q_{post,N}^{(k+1)})^{-1} b_{N}^{(k+1)}]_{2t})+ \Sigma_{\epsilon,t}^{(k+1)}  \}_{t=1}^{T}.
\end{align*}

 Then compute the associated correlation matrix time series $\{R_{t}^{(k+1)} \}_{t=1}^{T}$. Then compute $\{ \text{score}(R_{t}^{(k+1)}) \}_{t=1}^{T}$.
\end{itemize}
\addcontentsline{toc}{section}{References}
\setcitestyle{aysep={,}}
\bibliography{bibliography}
\end{document}